\documentclass{article}
\usepackage{amsmath}
\usepackage{amssymb}
\usepackage{slashed}
\usepackage{hyperref}

\newtheorem{theorem}{Theorem}[section]

\newtheorem{definition}[theorem]{Definition}

\newtheorem{lemma}[theorem]{Lemma}
\newtheorem{notation}[theorem]{Notation}

\newtheorem{proposition}[theorem]{Proposition}
\newtheorem{remark}[theorem]{Remark}

\newenvironment{proof}[1][Proof]{\noindent\textbf{#1.} }{\ \rule{0.5em}{0.5em}}
\renewcommand{\theequation}{\thesection.\arabic{equation}}
\input{tcilatex}
\begin{document}

\title{$G_{2}$-structures for $N=1$ supersymmetric $\func{AdS}_{4}$
solutions of $M$-theory}
\author{Sergey Grigorian \\
School of Mathematical and Statistical Sciences\\
University of Texas Rio Grande Valley\\
1201 W. University Drive\\
Edinburg, TX 78539}
\maketitle

\begin{abstract}
We study the $N=1$ supersymmetric solutions of $D=11$ supergravity obtained
as a warped product of four-dimensional anti-de-Sitter space with a
seven-dimensional Riemannian manifold $M$. Using the octonion bundle
structure on $M$ we reformulate the Killing spinor equations in terms of
sections of the octonion bundle on $M$. The solutions then define a single
complexified $G_{2}$-structure on $M$ or equivalently two real $G_{2}$%
-structures. We then study the torsion of these $G_{2}$-structures and the
relationships between them.
\end{abstract}

\tableofcontents

\section{Introduction}

\setcounter{equation}{0}It is known that supersymmetric vacua of String
Theory and M-theory are determined by nowhere-vanishing spinors that satisfy
a first-order differential equation, known as a Killing spinor equation. The
precise form of the equation depends on particular details of the theory. In
particular, in a theory with vanishing fluxes, the Killing spinor equation
just reduces to an equation for a parallel spinor with respect to the
Levi-Civita connection. Since existence of non-trivial parallel spinors is
equivalent to a reduced holonomy group, this gives the key implication that
if the full $D$-dimensional spacetime is a product of the $4$-dimensional
Minkowski space and correspondingly, a $\left( D-4\right) $-dimensional
internal space, then the internal space needs to be a $6$-dimensional
Calabi-Yau manifold if $D=10$ and must be a $7$-dimensional manifold with
holonomy contained in $G_{2}$ if $D=11$. In terms of $G$-structures, this is
equivalent to a torsion-free $SU\left( 3\right) $-structure in $6$
dimensions and a torsion-free $G_{2}$-structure in $7$ dimensions. However,
if there are non-trivial fluxes, the Killing spinor equation will have
additional terms that are defined by the fluxes and the resulting Killing
spinors will no longer be parallel. In the setting of a (warped)\ product $%
M^{4}\times M^{D-4}$ this will then induce a non-flat metric on $M^{4}$ and
the spinors on the internal space will no longer define a torsion-free $G$%
-structure, instead there will be some torsion components. Moreover, if
quantum corrections are also considered, then the Killing spinor equation
will also contain additional terms, which will also lead to corrections away
from torsion-free structures. Since the early 1980's there has been great
progress in understanding the geometric implications of flux
compactifications of different theories with different details, such as
compactifications of type $IIA$, type $IIB$, and heterotic supergravities
for $D=10$, and compactifications of $D=11$ supergravities with different
amounts of supersymmetry \cite%
{Awada:1982pk,DuffPope:1986hr,Papadopoulos:1995da,Becker:2000rz,DallAgata:2003txk,Kaste:2003zd,Behrndt:2003uq,Behrndt:2003zg,Behrndt:2004bh,DallAgata:2005zlf,Gauntlett:2004hs,Lukas:2004ip,Lust:2004ig,Behrndt:2005bv,CveticSUSYflux,Grana:2005jc,Micu:2006ey,SparksFlux,Fernandez:2008wla,Ashmore:2015joa,Triendl:2015rta,Macpherson:2016xwk,Passias:2017yke}%
. There have also been attempts in including higher-order corrections and
approaches that do not necessarily factor Killing spinors into $4$%
-dimensional and $\left( D-4\right) $-dimensional pieces \cite%
{Gutowski:2014ova,Gran:2016zxk}. Other approaches also involved higher
derivative corrections \cite{Lu:2003ze,McOrist:2012yc,Becker:2014rea}. There
have also been multiple approaches using different aspects of generalized
geometry \cite%
{Grana:2005sn,Hull:2007zu,Coimbra:2015nha,Malek:2016bpu,Malek:2017njj,Coimbra:2017fqj}%
.

In this paper we will primarily consider a warped product compactification
of $11$-dimensional $N=1$ supergravity on $\func{AdS}_{4}\times M$, where $M$
is a $7$-dimensional space, however many of the results can also be adapted
for a Minkowski background. Our emphasis is to gain a better understanding
of the torsion of $G_{2}$-structures that are imposed by the supersymmetry
constraints. The properties of $G_{2}$-structures in supergravity
compactifications have been studied in both mathematical and physical
literature, but for different details of the theory \cite%
{Awada:1982pk,Kaste:2003zd,Behrndt:2004bh,Duff:2002rw,Friedrich:2001yp,Gutowski:2001fm,WittenBeasley,Agricola:2003by,AcharyaGukov,House:2004pm,GrigorianYau1,GrigorianG2Review}%
. Moreover, in the physics literature the results are usually expressed in
terms of an $SU\left( 3\right) $-structure or an $SU\left( 2\right) $%
-structure on the $7$-dimensional manifold \cite%
{DallAgata:2003txk,Behrndt:2004bh,Lukas:2004ip,CveticSUSYflux,Micu:2006ey,SparksFlux,Triendl:2015rta}%
. While this approach has its merits since it more readily allows to compare 
$11$-dimensional theories with $10$-dimensional ones, it does not use some
of the unique feature of $7$-dimensional $G_{2}$-structures.

Seven-dimensional manifolds with a $G_{2}$-structure have been of great
interest in differential geometry ever since Alfred Gray studied vector
cross products on orientable manifolds in 1969 \cite{Gray-VCP}. It turns out
that a $2$-fold vector cross product - that is, one that takes two vectors
and outputs another one, exists only in $3$ dimensions and in $7$
dimensions. The $3$-dimensional vector cross product is very well known in $%
\mathbb{R}^{3}$ and on a general oriented $3$-manifold it comes from the
volume $3$-form, so it is a special case of a $\left( n-1\right) $-fold
vector cross product in a $n$-dimensional space, where it also comes from
the volume form. In $7$ dimensions, however, the vector cross product
structure is even more special, since it is not part of a generic sequence.
The $3$-dimensional and $7$-dimensional vector cross products do however
have something in common since they are closely related to the normed
division algebras - the quaternions and octonions, which are $4$- and $8$%
-dimensional, respectively. A good review of vector cross product geometries
can also be found in \cite{ConanChapter}. On a $7$-manifold, the group that
preserves the vector cross product is precisely the $14$-dimensional Lie
group $G_{2}$ - this is the automorphism group of the octonion algebra. In
order to be able to define a vector cross product globally on a $7$%
-manifold, we need to introduce a $G_{2}$\emph{-structure}, which is now a
reduction of the frame bundle to $G_{2}.$ A necessary and sufficient
condition on the $7$-manifold to admit a $G_{2}$-structure is that it is
orientable and admits a spin structure. There are purely topological
conditions of the vanishing of the first two Stiefel-Whitney classes \cite%
{FernandezGray, FriedrichNPG2}. Once we have a $7$-manifold that satisfies
these conditions, any Riemannian metric $g$ will give rise to an $SO\left(
7\right) $-structure, and this could then be reduced to a $G_{2}$-structure.
By specifying a $G_{2}$-principal bundle, we are effectively also defining a 
$G_{2}$-invariant $3$-form $\varphi $ that is compatible with $g$ and gives
rise to the structure constants for the vector cross product. As it is
well-known \cite{bryant-2003} for each metric $g$ there is a family of
compatible $G_{2}$-structures. Pointwise, such a family is parametrized by $%
SO\left( 7\right) /G_{2}\cong \mathbb{R}P^{7}.$

Due to the close relationship between $G_{2}$ and octonions, it is natural
to introduce an octonionic structure on a $7$-manifold with a $G_{2}$%
-structure. In \cite{GrigorianOctobundle}, the notion of an octonion bundle
on a $7$-manifold with $G_{2}$-structure has been introduced. A number of
properties of $G_{2}$-structures are re-expressed in a very natural form
using the octonion formalism, which we review in Section \ref{secg2struct}
of this paper. In many ways, the structure of the octonion bundle mirrors
that of the spinor bundle on a $7$-manifold. It is well known that a $G_{2}$%
-structure may be defined by a unit spinor on the manifold. Under the
correspondence between the spinor bundle and the octonion bundle, the fixed
spinor is then mapped to $1.$ A change of the unit spinor then corresponds
to a transformation of the $G_{2}$-structure within the same metric class.
As it is well known, the enveloping algebra of the octonions, i.e. the
algebra of left multiplication maps by an octonion under composition, gives
an isomorphism with the spinor bundle as Clifford modules \cite{HarveyBook}.
However, the Clifford module structure is by definition associative, so this
correspondence only captures part of the structure of the octonion bundle.
The full non-associative structure of the octonion bundle cannot be seen in
the spinor bundle, therefore it is expected that the octonion bundle carries
more information than the spinor bundle, although the difference is subtle.
In particular, while there is no natural binary operation on the spinor
bundle, we can multiply octonions. In fact, Clifford multiplication of a
vector and a spinor translates to multiplication of two octonions -
therefore we are implicitly using the \emph{triality} correspondence between
vector and spinor representations which is unique to $7$ dimensions \cite%
{BaezOcto}. This gives a non-associative division algebra structure,
something that we generically do not have on spinor bundles. Moreover, as
shown in \cite{GrigorianOctobundle}, isometric $G_{2}$-structures, i.e.
those that are compatible with the same metric, are readily described in
terms of octonions and their torsion tensors are then re-interpreted as $1$%
-forms with values in the imaginary octonion bundle. Since $N=1$
supersymmetry in $D=11$ is most readily described using a \emph{complex}
spinor on the $7$-manifold, in Section \ref{secCxOcto} we also introduce
complexified octonions and corresponding complexified $G_{2}$-structures.
Complexified octonions no longer form a division algebra and contain zero
divisors.

Due to the close relationship between spinors and octonions in $7$
dimensions, any spinorial expression on a $7$-manifold with a $G_{2}$%
-structure can be recast in terms of octonions. In Section \ref{secoctosusy}
we rewrite the Killing spinor equation in $7$ dimensions in terms of a
complexified octonion section $Z$. It turns out that a key property of $Z$
is whether or not it is a zero divisor. On an $\func{AdS}_{4}$ background
with non-vanishing $4$-dimensional flux, it is shown that $Z$ is nowhere a
zero divisor. On the other hand, if the $4$-dimensional flux vanishes, there
could be a locus where $Z$ is a zero divisor, however in any neighborhood, $%
Z $ will be not be everywhere a zero divisor. If the $4$-dimensional
spacetime is Minkowski, then $Z$ is either everywhere a zero divisor or is
nowhere a zero divisor. In any case, wherever $Z$ is not a zero divisor, we
are able to define a complexified $G_{2}$-structure corresponding to $Z$,
and we show that its torsion lies in the class $\mathbf{1}_{\mathbb{C}%
}\oplus \mathbf{7}_{\mathbb{C}}\oplus \mathbf{27}_{\mathbb{C}}$ (here we are
referring to the complex representations of $G_{2}$), with the component in $%
\mathbf{7}_{\mathbb{C}}$ moreover being a exact $1$-form. It is also easy to
see in this case that the Killing equation implies that the $7$-dimensional
flux $4$-form $G$ has to satisfy the Bianchi identity $dG=0$.

The complexified octonion section $Z$ decomposes as $X+iY$ in terms of $%
\mathbb{C}$-real octonions $X$ and $Y$, which define real $G_{2}$-structures 
$\varphi _{X}$ and $\varphi _{Y}$, respectively, both of which correspond to
the same metric. The octonion section $W=YX^{-1}$ then defines the
transformation from $\varphi _{X}$ to $\varphi _{Y}.$ The property of $Z$
being a zero divisor is then equivalent to $\left\vert W\right\vert ^{2}=1$
and the $\mathbb{O}$-real part $w_{0}$ of $W\ $vanishing. Assumptions that
are equivalent to these have sometimes been made in the literature, however $%
\left\vert W\right\vert ^{2}$ and $w_{0}$ satisfy some particular
differential equations and cannot be arbitrarily set to particular values.
The vector field $w$ which is the $\mathbb{O}$-imaginary part of $W$, is
precisely the vector field that has been used in the literature to reduce
the $G_{2}$-structure on the internal $7$-manifold to an $SU\left( 3\right) $%
-structure. However, in Section \ref{secoctosusy} we find that if the
internal manifold is compact and the $4$-dimensional spacetime is $\func{AdS}%
_{4},$ then $w$ cannot be nowhere-vanishing. Therefore, in these cases, $%
SU\left( 3\right) $-structures or Sasakian structures on the $7$-manifold
would be degenerate and only valid locally, since they depend on a
nowhere-vanishing vector field.

In Section \ref{secExteq}, we then use the external Killing spinor equations
to obtain expressions for some of the components of the $4$-form $G$ in
terms of $W$ and the warp factor. This leads on to Section \ref{secInteq},
where the internal Killing spinor equations in terms of octonions are used
to derive the torsion tensors $T^{\left( X\right) }$ and $T^{\left( Y\right)
}$ of the real $G_{2}$-structures $\varphi _{X}$ and $\varphi _{Y}.$ In
particular, we find that if $4$-dimensional flux is non-vanishing, then both 
$T^{\left( X\right) }$ and $T^{\left( Y\right) }$ are in the generic torsion
class $\mathbf{1}\oplus \mathbf{7}\oplus \mathbf{14\oplus 27},$ however the $%
\mathbf{7}$ components are exact, so both $T^{\left( X\right) }$ and $%
T^{\left( Y\right) }$ are conformally in the torsion class $\mathbf{1}\oplus 
\mathbf{14\oplus 27}.$ If the $4$-dimensional flux vanishes, then both $%
T^{\left( X\right) }$ and $T^{\left( Y\right) }$ are in the torsion class $%
\mathbf{7}\oplus \mathbf{14\oplus 27},$ and moreover if $w$ is not
proportional to the gradient of warp factor, then all the torsion components
are non-zero. In the case of a Minkowski background, the torsion is
conformally in the class $\mathbf{14\oplus 27}$. In the particular case of
the Minkowski background with $\left\vert W\right\vert ^{2}=1$ and $w_{0}=0$%
, our expression for the torsion reduces to the one obtained in \cite%
{KasteMinasianFlux}, however, to the author's knowledge, the torsion of a $%
G_{2}$-structure obtained from a compactification to $\func{AdS}_{4}$ has
not appeared previously in the literature. We then also work out the
covariant derivative of $W.$

In Section \ref{secInteg}, we consider the integrability conditions for $%
G_{2}$-structure torsion and outline how they relate to the equations of
motion and Einstein's equations. In general, the integrability conditions
for Killing spinor equations in $D=11$ have been considered in \cite%
{Gauntlett:2002fz} where it has been shown that if the Bianchi identity and
the equations of motion for the flux in $D=11$ is satisfied, then the
Killing spinor equations also imply Einstein's equations. This result has
also been used in \cite{SparksFlux} to conclude that in a $N=2$
compactification of $D=11$ supergravity to $\func{AdS}_{4}$ the Killing
spinor equations imply the Bianchi identity, the equations of motion, and
Einstein's equations. Here we outline how to obtain all of this directly
from the integrability conditions for the $G_{2}$-structures $\varphi _{X}$
and $\varphi _{Y}.$

\begin{description}
\item[Conventions] In this paper we will be using the following convention
for Ricci and Riemann curvature:%
\begin{equation}
\func{Ric}_{jl}=g^{ik}\func{Riem}_{ijkl}  \label{riemconvention}
\end{equation}%
Also, the convention for the orientation of a $G_{2}$-structure will same as
the one adopted by Bryant \cite{bryant-2003} and follows the author's
previous papers. In particular, this causes $\psi =\ast \varphi $ to have an
opposite sign compared to the works of Karigiannis, so many identities and
definitions cited from \cite%
{karigiannis-2005-57,karigiannis-2006notes,karigiannis-2007} may have
differing signs.
\end{description}

\section{Supersymmetry}

\setcounter{equation}{0}Consider the basic bosonic action of
eleven-dimensional supergravity \cite{Cremmer:1978km}, which is supposed to
describe low-energy $M$-theory :

\begin{equation}
S=\frac{1}{2}\int \check{R}\check{\func{vol}}-\int \check{G}\wedge \check{%
\ast}\check{G}-\frac{1}{3}\int \check{C}\wedge \check{G}\wedge \check{G}
\label{sugraction}
\end{equation}%
where $\check{R}$ is the scalar curvature, $\check{\func{vol}}$ is the
volume, and $\check{\ast}$ is the Hodge star of the metric $\check{g}$ on
the $11$-dimensional spacetime; and $\check{C}$ is a local $3$-form
potential with field strength $\check{G}=d\check{C},$ so that $dG=0$. Note
that additional higher-order terms may also be present in (\ref{sugraction}) 
\cite{Lu:2003ze,Becker:2014rea,Coimbra:2017fqj}. From (\ref{sugraction}),
the bosonic equations of motion are found to be 
\begin{subequations}%
\label{D11eqs} 
\begin{eqnarray}
\check{\func{Ric}} &=&\frac{1}{3}\left( \check{G}\check{G}-2\check{g}%
\left\vert \check{G}\right\vert _{\check{g}}^{2}\right)  \label{EinsteinD11}
\\
d\ast G &=&-G\wedge G  \label{eomD11}
\end{eqnarray}%
\end{subequations}%
where $\left( \check{G}\check{G}\right) _{AB}=\check{G}_{AMNP}\check{G}%
_{B}^{\ MNP}\ $and $\left\vert \check{G}\right\vert _{\check{g}}^{2}=\frac{1%
}{24}\check{G}_{AMNP}\check{G}^{\ AMNP}$ with $\check{g}$ used to contract
indices. Note that we will be using Latin upper case letters for
11-dimensional indices, Latin lower case letters for $7$-dimensional
indices, and Greek letters for $4$-dimensional indices. Moreover, the
supersymmetry variations of the gravitino field $\check{\Psi}$ give the
following Killing equation:%
\begin{equation}
\delta \check{\Psi}_{A}=\left\{ \check{\nabla}_{A}^{S}+\frac{1}{144}\check{G}%
_{BCDE}\left( \check{\gamma}_{\ \ \ \ \ \ \ \ \ A}^{BCDE}-8\check{\gamma}%
^{CDE}\hat{g}_{\ A}^{B}\right) \right\} \check{\varepsilon}
\label{11dimsusytrans}
\end{equation}%
where all the checked objects are $11$-dimensional: $\check{\nabla}^{S}$ is
the spinorial Levi-Civita connection, $\check{\gamma}$ are gamma matrices,
and $\hat{\varepsilon}$ is a Majorana spinor. Nowhere-vanishing solutions to 
$\delta \check{\Psi}=0$ are precisely the Killing spinors. A supersymmetric
vacuum solution of $11$-dimensional $N=1$ supergravity would then be a
triple $\left( \check{g},\check{G},\check{\varepsilon}\right) $ where $%
\check{g}\ $is metric and $\check{G}$ is a closed $4$-form that satisfy
equations (\ref{D11eqs}) and $\check{\varepsilon}$ is a nowhere-vanishing
spinor that satisfies (\ref{11dimsusytrans}). Correspondingly, $N=2$
supersymmetry would have two independent solutions to (\ref{11dimsusytrans}%
). Note that it is known \cite{Gauntlett:2002fz} that if $\left( \check{g},%
\check{G},\check{\varepsilon}\right) $ satisfy $d\check{G}=0$, (\ref{eomD11}%
), and (\ref{11dimsusytrans}), then (\ref{EinsteinD11}) follows as an
integrability condition.

Assuming a warped product compactification from $11$ dimensions to $4,$ the $%
11$-dimensional metric $\check{g}$ is written as%
\begin{equation}
\check{g}=e^{2\Delta }\left( \eta \oplus \tilde{g}\right)
\label{11dimmetric}
\end{equation}%
where $\eta $ is the $4$-dimensional Minkowski or $\func{AdS}_{4}$ metric
and $\tilde{g}$ is a Riemannian metric on the $7$-dimensional
\textquotedblleft internal manifold\textquotedblright\ $M$. The warp factor $%
\Delta $ is assumed to be a real-valued smooth function on the internal
space $M$. Note that for convenience we will work with the \textquotedblleft
warped\textquotedblright\ metric $g=e^{2\Delta }\tilde{g}$ on $M$, so in
fact our $11$-dimensional metric will be taken to be 
\begin{equation}
\check{g}=\left( e^{2\Delta }\eta \right) \oplus g  \label{11dimmetric2}
\end{equation}%
The $11$-dimensional $4$-form flux $\check{G}$ is decomposed as 
\begin{equation}
\hat{G}=3\mu \func{vol}_{4}+G  \label{11dimGhat}
\end{equation}%
where $\mu $ is a real constant, $\func{vol}_{4}$ is the volume form of the
metric $\eta $ on the $4$-dimensional spacetime and $G$ is a 4-form on $M$.
Note that the condition $d\hat{G}=0$ forces $\mu $ to be constant and $dG=0$%
. With these ans\"{a}tze for the metric and the $4$-form, from the $11$%
-dimensional equation of motion (\ref{eomD11}) we obtain the following $7$%
-dimensional equation 
\begin{equation}
d\left( e^{4\Delta }\ast G\right) =-6\mu G  \label{eom7D}
\end{equation}%
and Einstein's equation (\ref{EinsteinD11}) gives us 
\begin{subequations}%
\label{riceinstein} 
\begin{eqnarray}
\check{\func{Ric}}_{\alpha \beta } &=&-\left( 12\mu ^{2}e^{-6\Delta }+\frac{2%
}{3}e^{2\Delta }\left\vert G\right\vert ^{2}\right) \eta _{\alpha \beta } \\
\check{\func{Ric}}_{ab} &=&\frac{1}{3}\left( GG\right) _{ab}+\left( 6\mu
^{2}e^{-8\Delta }-\frac{2}{3}\left\vert G\right\vert ^{2}\right) g_{ab}
\end{eqnarray}%
\end{subequations}%
where now $\left\vert G\right\vert ^{2}$ is with respect to warped metric $g$
on $M$. On the other hand, using a standard conformal transformation
formula, $\check{\func{Ric}}$ is related to the Ricci curvatures $\func{Ric}%
_{4}$ and $\widetilde{\func{Ric}}$ of the unwarped metrics in $4$ and $7$
dimensions in the following way 
\begin{equation*}
\check{\func{Ric}}=\func{Ric}_{4}+\widetilde{\func{Ric}}-9\left( \tilde{%
\nabla}d\Delta -d\Delta d\Delta \right) -\left( \tilde{\nabla}^{2}\Delta
+9\left\vert d\Delta \right\vert _{\tilde{g}}^{2}\right) \left( \eta +\tilde{%
g}\right)
\end{equation*}%
where $\tilde{\nabla}$ is the $7$-dimensional Levi-Civita connection for the
metric $\tilde{g}$. However, the Ricci curvature $\func{Ric}^{\left(
4\right) }$ of the unwarped $4$-dimensional metric $\eta $ is by assumption
given by 
\begin{equation}
\func{Ric}_{4}=-12\left\vert \lambda \right\vert ^{2}\eta  \label{RicAdS}
\end{equation}%
where $\lambda $ is a complex constant, which is of course $0$ if $\eta $ is
the Minkowski metric. Moreover, rewriting $\tilde{\nabla}$ and $\widetilde{%
\func{Ric}}$ in terms of the warped metric $g$ on $M$, and separating $%
\check{\func{Ric}}$ into the $4$-dimensional and $7$-dimensional parts, we
find 
\begin{subequations}%
\begin{eqnarray}
\check{\func{Ric}}^{\left( 4\right) } &=&-e^{2\Delta }\left( \nabla
^{2}\Delta +4\left\vert d\Delta \right\vert ^{2}+12\left\vert \lambda
\right\vert ^{2}e^{-2\Delta }\right) \eta \\
\check{\func{Ric}}^{\left( 7\right) } &=&\func{Ric}-4\left( \nabla d\Delta
+d\Delta d\Delta \right)
\end{eqnarray}%
\end{subequations}%
Now comparing with (\ref{riceinstein}), we obtain the following equation 
\begin{subequations}
\begin{eqnarray}
\nabla ^{2}\Delta &=&12\mu ^{2}e^{-8\Delta }+\frac{2}{3}\left\vert
G\right\vert ^{2}-4\left\vert d\Delta \right\vert ^{2}-12\left\vert \lambda
\right\vert ^{2}e^{-2\Delta }  \label{ric4cond} \\
\func{Ric} &=&4\left( \nabla d\Delta +d\Delta d\Delta \right) +\frac{1}{3}%
GG+\left( 6\mu ^{2}e^{-8\Delta }-\frac{2}{3}\left\vert G\right\vert
^{2}\right) g  \label{ric7cond}
\end{eqnarray}%
\end{subequations}%
So in particular, the $7$-dimensional scalar curvature $R$ satisfies the
following:%
\begin{equation}
\frac{1}{4}R=\nabla ^{2}\Delta +\left\vert d\Delta \right\vert ^{2}+\frac{5}{%
6}\left\vert G\right\vert ^{2}+\frac{21}{2}\mu ^{2}e^{-8\Delta }
\label{scal7}
\end{equation}%
We can see that the condition (\ref{ric4cond}) forces the Ricci curvature of
the $4$-dimensional spacetime to be negative if the internal space is
compact. This is a well-known `no-go' theorem from \cite{Gibbons:1984kp}.
Also, (\ref{ric4cond}) implies that if $\Delta $ is constant (and without
loss of generality can be set to $0$), then 
\begin{equation}
\left\vert G\right\vert ^{2}=18\left( \left\vert \lambda \right\vert
^{2}-\mu ^{2}\right)
\end{equation}%
In particular, this implies that in this case, on a Minkowski background,
with $\left\vert \lambda \right\vert =0$, both $G$ and $\mu $ have to
vanish. On an $\func{AdS}_{4}$ background, this still gives the restrictive
requirement that $\left\vert G\right\vert ^{2}$ has to be constant.
Moreover, from (\ref{scal7}) we then see that this implies that the $7$%
-dimensional scalar curvature has to be a positive constant.

The $11$-dimensional gamma matrices $\check{\gamma}$ are decomposed as 
\begin{equation}
\check{\gamma}=e^{\Delta }\tilde{\gamma}\otimes \mathbb{I+\tilde{\gamma}}%
^{\left( 5\right) }\otimes \gamma  \label{11dimgamhat}
\end{equation}%
where $\tilde{\gamma}\ $are real $4$-dimensional gamma matrices, $\mathbb{%
\tilde{\gamma}}^{\left( 5\right) }=i\tilde{\gamma}^{\left( 1\right) }\tilde{%
\gamma}^{\left( 2\right) }\tilde{\gamma}^{\left( 3\right) }\tilde{\gamma}%
^{\left( 4\right) }$ is the four-dimensional chirality operator, and $\gamma 
$ are imaginary gamma matrices in $7$ dimensions that satisfy the standard
Clifford algebra identity:

\begin{equation}
\gamma _{a}\gamma _{b}+\gamma _{b}\gamma _{a}=2g_{ab}\func{Id}
\label{Cliffgam}
\end{equation}%
The $11$-dimensional spinor $\check{\varepsilon}$ is decomposed as 
\begin{equation}
\check{\varepsilon}=\varepsilon \otimes \theta +\varepsilon ^{\ast }\otimes
\theta ^{\ast }  \label{11dimepshat}
\end{equation}%
where $\theta \ $is a complex spinor on the $7$-dimensional internal space,
and $\varepsilon $ is a $4$-dimensional Weyl spinor that also satisfies the
Killing spinor equation in $4$ dimensions:%
\begin{equation}
\tilde{\nabla}^{\mathcal{S}}\varepsilon =\tilde{\gamma}\lambda ^{\ast
}\varepsilon ^{\ast }  \label{4Dkilling}
\end{equation}%
where $\tilde{\nabla}^{\mathcal{S}}$ is the $4$-dimensional Levi-Civita
connection on spinors with respect to the metric $\eta $ and $\lambda $ is a
complex constant - known as a superpotential or the gravitino mass. In
particular, $\left\vert \lambda \right\vert ^{2}$ appears in the expression
for the Ricci curvature (\ref{RicAdS}).

Using the decomposition (\ref{11dimepshat}) as well as (\ref{4Dkilling}),
the Killing equation (\ref{11dimsusytrans}) gives us two sets of equations
on $\theta $: 
\begin{subequations}%
\label{7dimsusyeq}%
\begin{eqnarray}
0 &=&\lambda e^{-\Delta }\theta ^{\ast }+\left[ \left( \mu e^{-4\Delta }i+%
\frac{1}{2}\left( \partial _{c}\Delta \right) \gamma ^{c}\right) +\frac{1}{%
144}G_{bcde}\gamma ^{bcde}\right] \theta  \label{7dimsusyeq1} \\
\nabla _{a}^{\mathcal{S}}\theta &=&\left[ \frac{i}{2}\mu e^{-4\Delta }\gamma
_{a}-\frac{1}{144}\left( G_{bcde}\gamma _{\ \ \ \ \ \ \
a}^{bcde}-8G_{abcd}\gamma ^{bcd}\right) \right] \theta  \label{7dimsusyeq2}
\end{eqnarray}

\end{subequations}
Here the equation (\ref{7dimsusyeq1}) comes from the \textquotedblleft
external\textquotedblright\ part of (\ref{11dimsusytrans}), i.e. when $%
A=0,..,3$, and (\ref{11dimsusytrans}) comes from the \textquotedblleft
internal\textquotedblright\ part of (\ref{11dimsusytrans}), i.e. when $%
A=4,...,11$. Using standard gamma matrix identities, the equation (\ref%
{7dimsusyeq2}) can be rewritten as 
\begin{equation}
\nabla _{a}^{\mathcal{S}}\theta =\left[ \frac{i}{2}\mu e^{-4\Delta }\gamma
_{a}-\frac{1}{144}\gamma _{a}\left( G_{bcde}\gamma _{\ \ \ \ \ \ \
}^{bcde}\right) +\frac{1}{12}G_{abcd}\gamma ^{bcd}\right] \theta
\label{7dimsusyeq2a}
\end{equation}

Since we require the solutions of (\ref{11dimsusytrans}) and hence (\ref%
{7dimsusyeq}) to be nowhere-vanishing, this implies that the $7$-dimensional
internal manifold $M$ must admit a $G_{2}$-structure. In particular, given
any metric $g$ on $M,$ there will exist a family of $G_{2}$-structures that
are compatible with $g$. In terms of $G_{2}$-structures, the equations (\ref%
{7dimsusyeq}) then say that a $G_{2}$-structure with a particular torsion
must exist within this metric class of $G_{2}$-structures. In \cite%
{Kaste:2003zd} and \cite{Behrndt:2003zg}, the expressions for the torsion of
a particular member of the metric class were explicitly calculated in terms
of the $4$-form $G$ under some assumptions on $\theta _{\pm }$. In this
paper, we will use the octonion bundle formalism developed in \cite%
{GrigorianOctobundle} to reformulate equations (\ref{7dimsusyeq}) in terms
of octonions on $M$ and will use that to derive properties of the torsion of
the corresponding $G_{2}$-structures.

\section{$G_{2}$-structures and octonion bundles}

\setcounter{equation}{0}\label{secg2struct}

The 14-dimensional group $G_{2}$ is the smallest of the five exceptional Lie
groups and is closely related to the octonions, which is the noncommutative,
nonassociative, $8$-dimensional normed division algebra. In particular, $%
G_{2}$ can be defined as the automorphism group of the octonion algebra.
Given the octonion algebra $\mathbb{O}$, there exists a unique orthogonal
decomposition into a real part, that is isomorphic to $\mathbb{R},$ and an 
\emph{imaginary }(or \emph{pure}) part, that is isomorphic to $\mathbb{R}%
^{7} $:%
\begin{equation}
\mathbb{O}\cong \mathbb{R}\oplus \mathbb{R}^{7}
\end{equation}%
Octonion multiplication to define a vector cross product $\times $ on $%
\mathbb{R}^{7}$. Given vectors $u,v\in \mathbb{R}^{7}$, we regard them as
octonions in $\func{Im}\mathbb{O}$, multiply them together using octonion
multiplication, and then project the result to $\func{Im}\mathbb{O}$ to
obtain a new vector in $\mathbb{R}^{7}$:%
\begin{equation}
u\times v=\func{Im}\left( uv\right) .  \label{octovp}
\end{equation}%
The subgroup of $GL\left( 7,\mathbb{R}\right) $ that preserves this vector
cross product is then precisely the group $G_{2}$. A detailed account of the
properties of the octonions and their relationship to exceptional Lie groups
is given by John Baez in \cite{BaezOcto}. The structure constants of the
vector cross product define a $3$-form on $\mathbb{R}^{7}$, hence $G_{2}$ is
alternatively defined as the subgroup of $GL\left( 7,\mathbb{R}\right) $
that preserves a particular $3$-form $\varphi _{0}\in \Lambda ^{3}\left( 
\mathbb{R}^{7}\right) ^{\ast }$ \cite{Joycebook}.

In general, given a $n$-dimensional manifold $M$, a $G$-structure on $M$ for
some Lie subgroup $G$ of $GL\left( n,\mathbb{R}\right) $ is a reduction of
the frame bundle $F$ over $M$ to a principal subbundle $P$ with fibre $G$. A 
$G_{2}$-structure is then a reduction of the frame bundle on a $7$%
-dimensional manifold $M$ to a $G_{2}$-principal subbundle. The obstructions
for the existence of a $G_{2}$-structure are purely topological. Given a $7$%
-dimensional smooth manifold that is both orientable ($w_{1}=0$) and admits
a spin structure ($w_{2}=0$), there always exists a $G_{2}$-structure on it 
\cite{FernandezGray, FriedrichNPG2, Gray-VCP}.

There is a $1$-$1$ correspondence between $G_{2}$-structures on a $7$%
-manifold and smooth $3$-forms $\varphi $ for which the $7$-form-valued
bilinear form $B_{\varphi }$ as defined by (\ref{Bphi}) is positive definite
(for more details, see \cite{Bryant-1987} and the arXiv version of \cite%
{Hitchin:2000jd}). 
\begin{equation}
B_{\varphi }\left( u,v\right) =\frac{1}{6}\left( u\lrcorner \varphi \right)
\wedge \left( v\lrcorner \varphi \right) \wedge \varphi  \label{Bphi}
\end{equation}%
Here the symbol $\lrcorner $ denotes contraction of a vector with the
differential form: 
\begin{equation}
\left( u\lrcorner \varphi \right) _{mn}=u^{a}\varphi _{amn}.
\label{contractdef}
\end{equation}%
Note that we will also use this symbol for contractions of differential
forms using the metric.

A smooth $3$-form $\varphi $ is said to be \emph{positive }if $B_{\varphi }$
is the tensor product of a positive-definite bilinear form and a
nowhere-vanishing $7$-form. In this case, it defines a unique Riemannian
metric $g_{\varphi }$ and volume form $\func{vol}_{\varphi }$ such that for
vectors $u$ and $v$, the following holds 
\begin{equation}
g_{\varphi }\left( u,v\right) \func{vol}_{\varphi }=\frac{1}{6}\left(
u\lrcorner \varphi \right) \wedge \left( v\lrcorner \varphi \right) \wedge
\varphi  \label{gphi}
\end{equation}%
An equivalent way of defining a positive $3$-form $\varphi $, is to say that
at every point, $\varphi $ is in the $GL\left( 7,\mathbb{R}\right) $-orbit
of $\varphi _{0}$. It can be easily checked that the metric (\ref{gphi}) for 
$\varphi =\varphi _{0}$ is in fact precisely the standard Euclidean metric $%
g_{0}$ on $\mathbb{R}^{7}$. Therefore, every $\varphi $ that is in the $%
GL\left( 7,\mathbb{R}\right) $-orbit of $\varphi _{0}$ has an \emph{%
associated} Riemannian metric $g$, that is in the $GL\left( 7,\mathbb{R}%
\right) $-orbit of $g_{0}.$ The only difference is that the stabilizer of $%
g_{0}$ (along with orientation) in this orbit is the group $SO\left(
7\right) $, whereas the stabilizer of $\varphi _{0}$ is $G_{2}\subset
SO\left( 7\right) $. This shows that positive $3$-forms forms that
correspond to the same metric, i.e., are \emph{isometric}, are parametrized
by $SO\left( 7\right) /G_{2}\cong \mathbb{RP}^{7}\cong S^{7}/\mathbb{Z}_{2}$%
. Therefore, on a Riemannian manifold, metric-compatible $G_{2}$-structures
are parametrized by sections of an $\mathbb{RP}^{7}$-bundle, or
alternatively, by sections of an $S^{7}$-bundle, with antipodal points
identified.

\begin{definition}
If two $G_{2}$-structures $\varphi _{1}$ and $\varphi _{2}$ on $M$ have the
same associated metric $g$, we say that $\varphi _{1}$ and $\varphi _{2}$
are in the same \emph{metric class}.
\end{definition}

Let $\left( M,g\right) $ be a smooth $7$-dimensional Riemannian manifold,
with $w_{1}=w_{2}=0.$ We know $M$ admits $G_{2}$-structures. In particular,
let $\varphi $ be a $G_{2}$-structure for $M$ for which $g$ is the
associated metric. We also use $g$ to define the Levi-Civita connection $%
\nabla $, and the Hodge star $\ast $. In particular, $\ast \varphi $ is a $4$%
-form dual to $\varphi ,$ which we will denote by $\psi $. We will give a
brief overview of the properties of octonion bundles as developed in \cite%
{GrigorianOctobundle}.

\begin{definition}
The \emph{octonion bundle }$\mathbb{O}M$ on $M$ is the rank $8$ real vector
bundle given by 
\begin{equation}
\mathbb{O}M\cong \Lambda ^{0}\oplus TM  \label{OMdef}
\end{equation}%
where $\Lambda ^{0}\cong M\times \mathbb{R}\ $is a trivial line bundle. At
each point $p\in M$, $\mathbb{O}_{p}M\cong \mathbb{R}\oplus T_{p}M.$
\end{definition}

The definition (\ref{OMdef}) simply mimics the decomposition of octonions
into real and imaginary parts. The bundle $\mathbb{O}M$ is defined as a real
bundle, but which will have additional structure as discussed below. Now let 
$A\in \Gamma \left( \mathbb{O}M\right) $ be a section of the octonion
bundle. We will call $A$ simply an \emph{octonion} \emph{on }$M.$ From (\ref%
{OMdef}), $A$ has a real part $\func{Re}A,$ which is a scalar function on $M$%
, as well as an imaginary part $\func{Im}A$ which is a vector field on $M$.
For convenience, we may also write $A=\left( \func{Re}A,\func{Im}A\right) $
or $A=\left( 
\begin{array}{c}
\func{Re}A \\ 
\func{Im}A%
\end{array}%
\right) $. Also, let us define octonion conjugation given by $\bar{A}=\left( 
\func{Re}A,-\func{Im}A\right) .$ Whenever there is a risk of ambiguity, we
will use $\hat{1}$ to denote the generator of $\Gamma \left( \func{Re}%
\mathbb{O}M\right) $.

Since $\mathbb{O}M$ is defined as a tensor bundle, the Riemannian metric $g$
on $M$ induces a metric on $\mathbb{O}M.$ Let $A=\left( a,\alpha \right) \in
\Gamma \left( \mathbb{O}M\right) .$ Then, 
\begin{eqnarray}
\left\vert A\right\vert ^{2} &=&\left\langle A,A\right\rangle =a^{2}+g\left(
\alpha ,\alpha \right)  \notag \\
&=&a^{2}+\left\vert \alpha \right\vert ^{2}  \label{Omet}
\end{eqnarray}%
We will be using the same notation for the norm, metric and inner product
for sections of $\mathbb{O}M$ as for standard tensors on $M$. It will be
clear from the context which is being used. If however we need to specify
that only the octonion inner product is used, we will use the notation $%
\left\langle \cdot ,\cdot \right\rangle _{\mathbb{O}}.$ The metric (\ref%
{Omet}) ensures that the real and imaginary parts are orthogonal to each
other.

\begin{definition}
Given the $G_{2}$-structure $\varphi $ on $M,$ we define a \emph{vector
cross product with respect to }$\varphi $ on $M.$ Let $\alpha $ and $\beta $
be two vector fields, then define%
\begin{equation}
\left\langle \alpha \times _{\varphi }\beta ,\gamma \right\rangle =\varphi
\left( \alpha ,\beta ,\gamma \right)  \label{vcrossdef}
\end{equation}%
for any vector field $\gamma $ \cite{Gray-VCP,karigiannis-2005-57}.
\end{definition}

In index notation, we can thus write 
\begin{equation}
\left( \alpha \times _{\varphi }\beta \right) ^{a}=\varphi _{\ bc}^{a}\alpha
^{b}\beta ^{c}  \label{vcrossdef2}
\end{equation}%
Note that $\alpha \times _{\varphi }\beta =-\beta \times _{\varphi }\alpha $%
. If there is no ambiguity as to which $G_{2}$-structure is being used to
define the cross product, we will simply denote it by $\times $, dropping
the subscript.

Using the contraction identity for $\varphi $ \cite%
{bryant-2003,GrigorianG2Review, karigiannis-2005-57}%
\begin{equation}
\varphi _{abc}\varphi _{mn}^{\ \ \ \ c}=g_{am}g_{bn}-g_{an}g_{bm}+\psi
_{abmn}  \label{phiphi1}
\end{equation}%
we obtain the following identity for the double cross product.

\begin{lemma}
Let $\alpha ,\beta ,\gamma $ be vector fields, then 
\begin{equation}
\alpha \times \left( \beta \times \gamma \right) =\left\langle \alpha
,\gamma \right\rangle \beta -\left\langle \alpha ,\beta \right\rangle \gamma
+\psi \left( \cdot ^{\sharp },\alpha ,\beta ,\gamma \right)
\label{doublecrossprod}
\end{equation}%
where $^{\sharp }$ means that we raise the index using the inverse metric $%
g^{-1}$.
\end{lemma}

Using the inner product and the cross product, we can now define the \emph{%
octonion product }on $\mathbb{O}M$.

\begin{definition}
Let $A,B\in \Gamma \left( \mathbb{O}M\right) .$ Suppose $A=\left( a,\alpha
\right) $ and $B=\left( b,\beta \right) $. Given the vector cross product (%
\ref{vcrossdef}) on $M,$ we define the \emph{octonion product }$A\circ
_{\varphi }B$\emph{\ with respect to }$\varphi $\emph{\ }as follows:%
\begin{equation}
A\circ _{\varphi }B=\left( 
\begin{array}{c}
ab-\left\langle \alpha ,\beta \right\rangle \\ 
a\beta +b\alpha +\alpha \times _{\varphi }\beta%
\end{array}%
\right)  \label{octoproddef}
\end{equation}
\end{definition}

If there is no ambiguity as to which $G_{2}$-structure is being used to
define the octonion product, for convenience, we will simply write$\ AB$ to
denote it. The octonion product behaves as expected with respect to
conjugation - $\overline{AB}=\overline{B}\overline{A}$ and the inner product 
$\left\langle AB,C\right\rangle =\left\langle B,\bar{A}C\right\rangle $ \cite%
{GrigorianOctobundle}.

A key property of the octonions is their non-associativity and given three
octonions $A=\left( a,\alpha \right) ,$ $B=\left( b,\beta \right) ,$ $%
C=\left( c,\gamma \right) $, a measure of their non-associativity is their 
\emph{associator} $\left[ A,B,C\right] $ defined by 
\begin{eqnarray}
\left[ A,B,C\right] &=&A\left( BC\right) -\left( AB\right) C  \notag \\
&=&2\psi \left( \cdot ^{\sharp },\alpha ,\beta ,\gamma \right)
\end{eqnarray}

We can use octonions to parametrize isometric $G_{2}$-structures is
well-known. If $A=\left( a,\alpha \right) $ is any nowhere-vanishing
octonion, then define 
\begin{equation}
\sigma _{A}\left( \varphi \right) =\frac{1}{\left\vert A\right\vert ^{2}}%
\left( \left( a^{2}-\left\vert \alpha \right\vert ^{2}\right) \varphi
-2a\alpha \lrcorner \left( \ast \varphi \right) +2\alpha \wedge \left(
\alpha \lrcorner \varphi \right) \right)  \label{sigmaAdef}
\end{equation}%
It is well-known \cite{bryant-2003} that given a fixed $G_{2}$-structure $%
\varphi $ that is associated with the metric $g,$ then any other $G_{2}$%
-structure in this metric class is given by (\ref{sigmaAdef}) for some
octonion section $A$. It has been shown \cite[Thm 4.8]{GrigorianOctobundle}
that for any nowhere-vanishing octonions $U$ and $V,$ 
\begin{equation}
\sigma _{U}\left( \sigma _{V}\varphi \right) =\sigma _{UV}\left( \varphi
\right)
\end{equation}%
Thus we can move freely between different $G_{2}$-structures in the same
metric class, which we will refer to as a \emph{change of gauge}. It has
also been shown in \cite{GrigorianOctobundle} that the octonion product $%
A\circ _{V}B$ that corresponds to the $G_{2}$-structure $\sigma _{V}\left(
\varphi \right) $ is given by 
\begin{equation}
A\circ _{V}B:=\left( AV\right) \left( V^{-1}B\right) =AB+\left[ A,B,V\right]
V^{-1}  \label{OctoVAB}
\end{equation}%
In particular, a change of gauge can be used to rewrite octonion expressions
that involve multiple products. Consider the expression $A\left( BC\right) $%
. Then, assuming $C$ is non-vanishing, we have 
\begin{eqnarray}
A\left( BC\right) &=&\left( AB\right) C+\left[ A,B,C\right]  \notag \\
&=&\left( A\circ _{C}B\right) C  \label{associd1}
\end{eqnarray}%
Similarly, 
\begin{eqnarray}
\left( AB\right) C &=&A\left( BC\right) -\left[ A,B,C\right]  \notag \\
&=&A\left( BC-A^{-1}\left[ A,B,C\right] \right)  \notag \\
&=&A\left( B\circ _{A^{-1}}C\right)  \label{associd2}
\end{eqnarray}%
Here we have used some properties of the associator from \cite[Lemma 3.9]%
{GrigorianOctobundle}.

Also, note that the cross product $A\times _{V}B$ of two $\mathbb{O}$%
-imaginary octonions with respect to $\sigma _{V}\left( \varphi \right) $
can be obtained from (\ref{OctoVAB}) as 
\begin{eqnarray}
A\times _{V}B &=&\frac{1}{2}\left( A\circ _{V}B-B\circ _{V}A\right)  \notag
\\
&=&A\times B+\left[ A,B,V\right] V^{-1}  \label{OctoVcross}
\end{eqnarray}

\subsection{Torsion and octonion covariant derivative}

\label{sectTors}Given a $G_{2}$-structure $\varphi $ with an associated
metric $g$, we may use the metric to define the Levi-Civita connection $%
\nabla $. The \emph{intrinsic torsion }of a $G_{2}$-structure is then
defined by $\nabla \varphi $. Following \cite%
{GrigorianG2Torsion1,karigiannis-2007}, we can write 
\begin{equation}
\nabla _{a}\varphi _{bcd}=2T_{a}^{\ e}\psi _{ebcd}^{{}}  \label{codiffphi}
\end{equation}%
where $T_{ab}$ is the \emph{full torsion tensor}. Similarly, we can also
write 
\begin{equation}
\nabla _{a}\psi _{bcde}=-8T_{a[b}\varphi _{cde]}  \label{psitorsion}
\end{equation}%
We can also invert (\ref{codiffphi}) to get an explicit expression for $T$ 
\begin{equation}
T_{a}^{\ m}=\frac{1}{48}\left( \nabla _{a}\varphi _{bcd}\right) \psi ^{mbcd}.
\end{equation}%
This $2$-tensor fully defines $\nabla \varphi $ \cite{GrigorianG2Torsion1}.
The torsion tensor $T$ is defined here as in \cite{GrigorianOctobundle}, but
actually corresponds to $\frac{1}{2}T$ in \cite{GrigorianG2Torsion1} and $-%
\frac{1}{2}T$ in \cite{karigiannis-2007}. With respect to the octonion
product, we get 
\begin{equation}
\nabla _{X}\left( AB\right) =\left( \nabla _{X}A\right) B+A\left( \nabla
_{X}B\right) -\left[ T_{X},A,B\right]  \label{nablaXAB}
\end{equation}

In general we can obtain an orthogonal decomposition of $T_{ab}$ according
to representations of $G_{2}$ into \emph{torsion components}: 
\begin{equation}
T=\tau _{1}g+\tau _{7}\lrcorner \varphi +\tau _{14}+\tau _{27}
\end{equation}%
where $\tau _{1}$ is a function, and gives the $\mathbf{1}$ component of $T$%
. We also have $\tau _{7}$, which is a $1$-form and hence gives the $\mathbf{%
7}$ component, and, $\tau _{14}\in \Lambda _{14}^{2}$ gives the $\mathbf{14}$
component and $\tau _{27}$ is traceless symmetric, giving the $\mathbf{27}$
component. As it was originally shown by Fern\'{a}ndez and Gray \cite%
{FernandezGray}, there are in fact a total of 16 torsion classes of $G_{2}$%
-structures that arise depending on which of the components are non-zero.
Moreover, as shown in \cite{karigiannis-2007}, the torsion components $\tau
_{i}$ relate directly to the expression for $d\varphi $ and $d\psi $ 
\begin{subequations}%
\label{dphidpsi}

\begin{eqnarray}
d\varphi &=&8\tau _{1}\psi -6\tau _{7}\wedge \varphi +8\iota _{\psi }\left(
\tau _{27}\right)  \label{dphi} \\
d\psi &=&-8\tau _{7}\wedge \psi -4\ast \tau _{14}.  \label{dpsi}
\end{eqnarray}

\end{subequations}%
where $\iota _{\psi }\left( h\right) _{abcd}=h_{[a}^{\ m}\psi _{\left\vert
m\right\vert bcd]}$ for any $2$-tensor $h$. In terms of the octonion bundle,
we can think of $T$ as an $\func{Im}\mathbb{O}$-valued $1$-form. If we
change gauge to a different $G_{2}$-structure $\varphi _{V}$, then the
torsion transforms in the following way.

\begin{theorem}[{\protect\cite[Thm 7.2]{GrigorianOctobundle}}]
\label{ThmTorsV}Let $M$ be a smooth $7$-dimensional manifold with a $G_{2}$%
-structure $\left( \varphi ,g\right) $ with torsion $T\in \Omega ^{1}\left( 
\func{Im}\mathbb{O}M\right) $. For a nowhere-vanishing $V\in \Gamma \left( 
\mathbb{O}M\right) ,$ consider the $G_{2}$-structure $\sigma _{V}\left(
\varphi \right) .$ Then, the torsion $T^{\left( V\right) }$ of $\sigma
_{V}\left( \varphi \right) $ is given by%
\begin{eqnarray}
T^{\left( V\right) } &=&\func{Im}\left( \func{Ad}_{V}T-\left( \nabla
V\right) V^{-1}\right)  \label{TorsV2} \\
&=&\func{Ad}_{V}T-\left( \nabla V\right) V^{-1}+\frac{1}{2}\frac{1}{%
\left\vert V\right\vert ^{2}}\left( d\left\vert V\right\vert ^{2}\right) 
\hat{1}  \notag
\end{eqnarray}
\end{theorem}

Here $\func{Ad}_{V}$ is the adjoint map $\func{Ad}_{V}:\Gamma \left( \mathbb{%
O}M\right) \longrightarrow \Gamma \left( \mathbb{O}M\right) $ given by $%
\func{Ad}_{V}\left( A\right) =VAV^{-1}$. Then, as in \cite%
{GrigorianOctobundle} let us define a connection on $\mathbb{O}M$, using $T$
as a \textquotedblleft connection\textquotedblright\ $1$-form.

\begin{definition}
Define the octonion covariant derivative $D$ such for any $X\in \Gamma
\left( TM\right) ,$ 
\begin{equation*}
D_{X}:\Gamma \left( \mathbb{O}M\right) \longrightarrow \Gamma \left( \mathbb{%
O}M\right)
\end{equation*}%
given by 
\begin{equation}
D_{X}A=\nabla _{X}A-AT_{X}  \label{octocov}
\end{equation}%
for any $A\in \Gamma \left( \mathbb{O}M\right) .$
\end{definition}

Note that in particular, $D_{X}1=-T_{X}.$Important properties of this
covariant derivative are that is satisfies a \emph{partial} derivation
property with respect to the octonion product and that it is metric
compatible

\begin{proposition}[{\protect\cite[Prop 6.4 and 6.5]{GrigorianOctobundle}}]
\label{propDXprod}Suppose $A,B\in \Gamma \left( \mathbb{O}M\right) $ and $%
X\in \Gamma \left( TM\right) $, then%
\begin{equation}
D_{X}\left( AB\right) =\left( \nabla _{X}A\right) B+A\left( D_{X}B\right)
\label{octocovprod}
\end{equation}%
and%
\begin{equation}
\nabla _{X}\left( g\left( A,B\right) \right) =g\left( D_{X}A,B\right)
+g\left( A,D_{X}B\right) \,  \label{DXmet}
\end{equation}
\end{proposition}

Moreover, $D_{X}$ is covariant with respect to change of gauge.

\begin{proposition}
\label{propDV}Suppose $\left( \varphi ,g\right) $ is a $G_{2}$-structure on
a $7$-manifold $M$, with torsion $T$ and corresponding octonion covariant
derivative $D$. Suppose $V$ is a unit octonion section, and $\tilde{\varphi}%
=\sigma _{V}\left( \varphi \right) $ is the corresponding $G_{2}$-structure,
that has torsion $\tilde{T}$, given by (\ref{TorsV2}), and an octonion
covariant derivative $\tilde{D}$. Then, for any octonion section $A$, we
have 
\begin{equation}
D\left( AV\right) =\left( \tilde{D}A\right) V  \label{DtildeAV2}
\end{equation}
\end{proposition}

Note that if $V$ does not have unit norm, then we have an additional term:%
\begin{equation}
D\left( AV\right) =\left( \tilde{D}A\right) V+\left( \frac{1}{2}\frac{1}{%
\left\vert V\right\vert ^{2}}d\left\vert V\right\vert ^{2}\right) AV
\label{DtildeAV3}
\end{equation}%
In \cite{GrigorianOctobundle}, it was also shown that the octonion bundle
structure has a close relationship with the spinor bundle $\mathcal{S}$ on $%
M $. It is well-known (e.g. \cite{BaezOcto,HarveyBook}) that octonions have
a very close relationship with spinors in $7$ dimensions, respectively. In
particular, the \emph{enveloping algebra} of the octonion algebra is
isomorphic to the Clifford algebra in $7$ dimensions. The (left) enveloping
algebra of $\mathbb{O}$ consists of left multiplication maps $%
L_{A}:V\longrightarrow AV$ for $A,V\in \mathbb{O},$ under composition \cite%
{SchaferNonassoc}. Similarly, a right enveloping algebra may also be
defined. Since the binary operation in the enveloping algebra is defined to
be composition, it is associative. More concretely, for $\mathbb{O}$%
-imaginary octonions $A$ and $B$, 
\begin{equation}
L_{A}L_{B}+L_{B}L_{A}=-2\left\langle A,B\right\rangle \func{Id}
\label{octoenvelopCliff}
\end{equation}%
which is the defining identity for a Clifford algebra. We see that while the
octonion algebra does give rise to the Clifford algebra, in the process we
lose the nonassociative structure, and hence the octonion algebra has more
structure than the corresponding Clifford algebra. Note that also due to
non-associativity of the octonions, in general that $L_{A}L_{B}\neq L_{AB}.$
In fact, we have the following relationship.

It is then well-known that a nowhere-vanishing spinor on $M$ defines a $%
G_{2} $-structure via a bilinear expression involving Clifford
multiplication. In fact, given a unit norm spinor $\xi \in \Gamma \left( 
\mathcal{S}\right) ,$ we may define 
\begin{equation}
\varphi _{\xi }\left( \alpha ,\beta ,\gamma \right) =-\left\langle \xi
,\alpha \cdot \left( \beta \cdot \left( \gamma \cdot \xi \right) \right)
\right\rangle _{S}  \label{phixispin}
\end{equation}%
where $\cdot $ denotes Clifford multiplication, $\alpha ,\beta ,\gamma $ are
arbitrary vector fields and $\left\langle \cdot ,\cdot \right\rangle _{S}$
is the inner product on the spinor bundle. Now suppose a $G_{2}$-structure $%
\varphi _{\xi }$ is defined by a unit norm spinor $\xi $ using (\ref%
{phixispin}). This choice of a reference $G_{2}$-structure then induces a
correspondence between spinors and octonions. We can define the linear map $%
j_{\xi }:\Gamma \left( \mathcal{S}\right) \longrightarrow \Gamma \left( 
\mathbb{O}M\right) $ by 
\begin{subequations}%
\label{jxiprops} 
\begin{eqnarray}
j_{\xi }\left( \xi \right) &=&1 \\
j_{\xi }\left( V\cdot \eta \right) &=&V\circ _{\varphi _{\xi }}j_{\xi
}\left( \eta \right) .  \label{jxiprops2}
\end{eqnarray}%
\end{subequations}%
As shown in \cite{GrigorianOctobundle} the map $j_{\xi }$ is in fact
pointwise an isomorphism of real vector spaces from spinors to octonions.
Moreover, it is metric-preserving, and under this map we get a relationship
between the spinor covariant derivative $\nabla ^{\mathcal{S}}$ and the
octonion covariant derivative $D.$ More precisely, for any $\eta \in \Gamma
\left( \mathcal{S}\right) $ 
\begin{equation}
j_{\xi }\left( \nabla _{X}^{S}\eta \right) =D_{X}\left( j_{\xi }\left( \eta
\right) \right)  \label{jxispinoctder}
\end{equation}%
where $D$ is the octonion covariant derivative (\ref{octocov}) with respect
to the $G_{2}$-structure $\varphi _{\xi }.$Therefore, the octonion bundle
retains all of the information from the spinor bundle, but has an additional
non-associative division algebra structure.

We may define a distinguished $\func{Im}\mathbb{O}$-valued $1$-form $\delta
\in \Omega ^{1}\left( \func{Im}\mathbb{O}M\right) $ such that for any vector 
$X$ on $M,$ $\delta \left( X\right) \mathcal{\in }\Gamma \left( \func{Im}%
\mathbb{O}M\right) ,$ with components given by 
\begin{equation}
\delta \left( X\right) =\left( 0,X\right) .  \label{deltadef}
\end{equation}%
Therefore in particular, $\delta $ is the isomorphism that takes vectors to
imaginary octonions.

\begin{notation}
For convenience, given a vector field $X,$ we will denote $\delta \left(
X\right) \mathcal{\in }\Gamma \left( \func{Im}\mathbb{O}M\right) $ by $\hat{X%
}$.
\end{notation}

In components, the imaginary part of $\delta $ is simply represented by the
Kronecker delta:%
\begin{equation}
\delta _{i}=\left( 0,\delta _{i}^{\ \alpha }\right) .  \label{deltadef2}
\end{equation}

Below are some properties of $\delta $

\begin{lemma}[\protect\cite{GrigorianOctobundle}]
\label{lemDelProps} Suppose $\delta \in \Omega ^{1}\left( \func{Im}\mathbb{O}%
M\right) $ is defined by (\ref{deltadef2}) on a $7$-manifold $M$ with $G_{2}$%
-structure $\varphi $ and metric $g$. It then satisfies the following
properties, where octonion multiplication is with respect to $\varphi $

\begin{enumerate}
\item $\delta _{i}\delta _{j}=\left( -g_{ij},\varphi _{ij}^{\ \ \alpha
}\right) $

\item $\delta _{i}\left( \delta _{j}\delta _{k}\right) =\left( -\varphi
_{ijk},\psi _{\ ijk}^{\alpha }-\delta _{i}^{\alpha }g_{jk}+\delta _{\
j}^{\alpha }g_{ik}-\delta _{\ k}^{\alpha }g_{ij}\right) $

\item For any $A=\left( a_{0},\alpha \right) \in \Gamma \left( \mathbb{O}%
M\right) ,$ 
\begin{equation}
\delta A=\left( 
\begin{array}{c}
-\alpha \\ 
a_{0}\delta -\alpha \lrcorner \varphi%
\end{array}%
\right)  \label{deltaA}
\end{equation}
\end{enumerate}
\end{lemma}

From Lemma \ref{lemDelProps} we see that left multiplication by $\delta $
gives a representation of the Clifford algebra. In particular, $\delta $
satisfy the Clifford algebra identity (\ref{octoenvelopCliff})%
\begin{equation}
\delta _{a}\delta _{b}+\delta _{b}\delta _{a}=-2g_{ab}\func{Id}
\label{octoenvelopCliff2}
\end{equation}%
In fact, using the map $j_{\xi },$ we can relate the Clifford algebra
representation on spinors in terms of imaginary gamma matrices to $\delta $: 
\begin{equation}
j_{\xi }\left( \gamma _{a}\right) =i\delta _{a}  \label{jgamdelta}
\end{equation}

We can now define the octonion Dirac operator $\slashed{D}$ using $\delta $
and the octonion covariant derivative $D$ (\ref{octocov}). Let $A\in \Gamma
\left( \mathbb{O}M\right) ,$ then define $\slashed{D}A$ as 
\begin{equation}
\slashed{D}A=\delta ^{i}\left( D_{i}A\right)  \label{defDiracOp}
\end{equation}%
This operator is precisely what we obtain by applying the map $j_{\xi }$ to
the standard Dirac operator on the spinor bundle.

A very useful property of $\slashed{D}$ is that it gives the $\tau _{1}$ and 
$\tau _{7}$ components of the torsion \cite{GrigorianOctobundle}. Suppose $V$
is a nowhere-vanishing octonion section, and suppose $\tilde{\varphi}=\sigma
_{V}\left( \varphi \right) $ has torsion tensor $\tilde{T}$ with $1$%
-dimensional and $7$-dimensional components $\tilde{\tau}_{1}\ $and $\tilde{%
\tau}_{7}.$ Then, 
\begin{equation}
\slashed{D}V=\left( 
\begin{array}{c}
7\tilde{\tau}_{1} \\ 
\frac{1}{\left\vert V\right\vert }\nabla \left\vert V\right\vert -6\tilde{%
\tau}_{7}%
\end{array}%
\right) V  \label{DirV}
\end{equation}%
In particular, 
\begin{equation}
\slashed{D}1=\left( 
\begin{array}{c}
7\tau _{1} \\ 
-6\tau _{7}%
\end{array}%
\right)  \label{Dir1}
\end{equation}%
where $\tau _{1}$ and $\tau _{7}$ are the corresponding components of $T$ -
the torsion tensor of the $G_{2}$-structure $\varphi .$

Let us also introduce some additional notation that will be useful later on.

\begin{notation}
Let $A\in \Omega ^{1}\left( \mathbb{O}M\right) ,$ Define the transpose $%
A^{t} $ of $A$ as 
\begin{equation}
\left( A^{t}\right) _{i}=\left( \func{Re}A\right) _{i}\hat{1}+\left\langle
\delta _{i},A_{j}\right\rangle \delta ^{j}  \label{Atranspose}
\end{equation}%
Thus, $A^{t}$ is an $\mathbb{O}M$-valued $1$-form with the same $\mathbb{O}$%
-real part as $A,$ but the $\mathbb{O}$-imaginary part is transposed.
Similarly, define%
\begin{equation}
\left( \func{Sym}A\right) _{i}=\frac{1}{2}\left( A+A^{t}\right) =\left( 
\func{Re}A\right) _{i}\hat{1}+\frac{1}{2}\left( \left\langle \delta
_{i},A_{j}\right\rangle \delta ^{j}+\left( \func{Im}A\right) _{i}\right)
\label{Symdef}
\end{equation}%
and 
\begin{equation}
\left( \func{Skew}A\right) _{i}=\frac{1}{2}\left( A-A^{t}\right) =\frac{1}{2}%
\left( \left\langle \delta _{i},A_{j}\right\rangle \delta ^{j}-\left( \func{%
Im}A\right) _{i}\right)  \label{Asymdef}
\end{equation}
\end{notation}

A useful property of the transpose that follows from (\ref{deltaA}) is the
following%
\begin{equation}
\left( \delta A\right) ^{t}=A\delta  \label{deltaAtr}
\end{equation}

\subsection{Complexified octonions}

\label{secCxOcto}Recall that the Lie group $G_{2}$ also has a non-compact
complex form $G_{2}^{\mathbb{C}}$, which is a complexification of the
compact real form $G_{2}.$ In particular, $G_{2}^{\mathbb{C}}$ is a subgroup
of $SO\left( 7,\mathbb{C}\right) $ - the complexified special orthogonal
group. The group $G_{2}^{\mathbb{C}}$ then preserves the complexified vector
product on $\mathbb{C}^{7}$ and is the automorphism group of the
complexified octonion algebra - the bioctonions \cite{BaezOcto,AkbulutCan}.
We define the \emph{bioctonion algebra }$\mathbb{O}_{\mathbb{C}}$ as the
complexification $\mathbb{C}\otimes _{\mathbb{R}}\mathbb{O}$ of the real
octonion algebra $\mathbb{O}.$ This can be regarded as an octonion algebra
over the field $\mathbb{C}.$ In fact, it is known that any octonion algebra
over $\mathbb{C}$ is isomorphic to $\mathbb{O}_{\mathbb{C}}$ \cite%
{AkbulutCan}. Any $Z\in \mathbb{O}_{\mathbb{C}}$ can be uniquely written as $%
Z=X+iY,$ where $X,Y\in \mathbb{O}.$ Since $\mathbb{C}$ is the base field for
this algebra, the imaginary unit $i$ commutes with octonions. To avoid
ambiguity we will make a distinction between the complex real and imaginary
parts and the octonion real and imaginary parts, so whenever there could be
ambiguity, we will refer to $\mathbb{C}$-real and $\mathbb{C}$-imaginary
parts or $\mathbb{O}$-real and $\mathbb{O}$-imaginary parts of a complex
octonion. Also, we will make a distinction between complex conjugation,
which will be denoted by $^{\ast }$ and octonionic conjugation that will be
denoted by $\overline{}$. That is, for $Z=X+iY\in \mathbb{O}_{\mathbb{C}},$
where $X,Y$ are $\mathbb{C}$-real octonions, we have 
\begin{subequations}
\begin{eqnarray}
\bar{Z} &=&\bar{X}+i\bar{Y} \\
Z^{\ast } &=&X-iY
\end{eqnarray}%
\end{subequations}%
Let us also linearly extend the inner product $\left\langle \cdot ,\cdot
\right\rangle $ from $\mathbb{O}$ to a complex-valued bilinear form on $%
\mathbb{O}_{\mathbb{C}}$. So if $Z_{1}=X_{1}+iY_{1}$ and $Z_{2}=X_{2}+iY_{2}$%
, we define 
\begin{eqnarray}
\left\langle Z_{1},Z_{2}\right\rangle &=&\left\langle
X_{1}+iY_{1},X_{2}+iY_{2}\right\rangle  \notag \\
&=&\left\langle X_{1},X_{2}\right\rangle -\left\langle
Y_{1},Y_{2}\right\rangle +i\left\langle X_{1},Y_{2}\right\rangle
+i\left\langle Y_{1},X_{2}\right\rangle  \label{cxinprod}
\end{eqnarray}%
Furthermore, define the complex-valued quadratic form $N$ on $\mathbb{O}_{%
\mathbb{C}}$ given by 
\begin{eqnarray}
N\left( Z\right) &=&Z\bar{Z}=\left\langle Z,Z\right\rangle  \notag \\
&=&\left( \left\vert X\right\vert ^{2}-\left\vert Y\right\vert ^{2}\right)
+2i\left\langle X,Y\right\rangle  \label{cxqform}
\end{eqnarray}%
for any $Z=X+iY\in \mathbb{O}_{\mathbb{C}}.$ Note that while $Z\bar{Z}$
gives a $\mathbb{O}$-real expression, taking $Z\overline{Z^{\ast }}$ is in
general not $\mathbb{O}$-real:%
\begin{equation}
Z\overline{Z^{\ast }}=\left( \left\vert X\right\vert ^{2}+\left\vert
Y\right\vert ^{2}\right) +i\left( Y\bar{X}-X\bar{Y}\right)  \label{ZZstbar}
\end{equation}%
In fact, it is curious that $\func{Re}_{\mathbb{C}}\left( Z\overline{Z^{\ast
}}\right) =\func{Re}_{\mathbb{O}}\left( Z\overline{Z^{\ast }}\right) $ and $%
\func{Im}_{\mathbb{C}}\left( Z\overline{Z^{\ast }}\right) =\func{Im}_{%
\mathbb{O}}\left( Z\overline{Z^{\ast }}\right) .$

It is easy to check that $N$ gives $\mathbb{O}_{\mathbb{C}}$ the structure
of a composition algebra, that is, for any $Z_{1},Z_{2}\in \mathbb{O}_{%
\mathbb{C}}$, 
\begin{equation}
N\left( Z_{1}Z_{2}\right) =N\left( Z_{1}\right) N\left( Z_{2}\right)
\end{equation}%
This quadratic form is moreover \emph{isotropic}, that is there exist
non-zero $Z\in \mathbb{O}_{\mathbb{C}}$ for which $N\left( Z\right) =0.$
Therefore, $\mathbb{O}_{\mathbb{C}}$ is a \emph{split composition algebra}.
Indeed, $N\left( Z\right) =0$ if and only if $Z=X+iY$ with $X,Y\in \mathbb{O}
$ such that 
\begin{equation}
\left\vert X\right\vert =\left\vert Y\right\vert \ \ \text{and }\left\langle
X,Y\right\rangle =0  \label{zeroN}
\end{equation}%
In fact a complex octonion $Z\in \mathbb{O}_{\mathbb{C}}$ is a zero divisor
if and only if $N\left( Z\right) =0.$ Suppose now $Z\in \mathbb{O}_{\mathbb{C%
}}$ is not a zero divisor. We can then define the inverse of $Z$ in a
standard way:%
\begin{equation}
Z^{-1}=\frac{\bar{Z}}{N\left( Z\right) }  \label{Zinv}
\end{equation}

Given a real manifold $M$, we may always complexify the tangent bundle as $%
TM\otimes \mathbb{C},$ and then if $M$ has a real Riemannian metric and a
choice of orientation, it can be uniquely linearly extended to $TM\otimes 
\mathbb{C}$, thus defining a $SO\left( n,\mathbb{C}\right) $-structure.
Similarly, given a real $G_{2}$-structure $\varphi $ on a $7$-manifold $M$
we can extend it to $TM\otimes \mathbb{C}$ and this will define a reduction
of the $SO\left( n,\mathbb{C}\right) $-structure to a $G_{2}^{\mathbb{C}}$%
-structure. The corresponding octonion bundle can now be also complexified,
to give the bioctonion bundle $\mathbb{O}_{\mathbb{C}}M,$ with $g$ defining
the quadratic form $N$ on $\mathbb{O}_{\mathbb{C}}M.$

However note that now, given a $Z\in \Gamma \left( \mathbb{O}_{\mathbb{C}%
}M\right) $ that is nowhere a zero divisor, we may define a complex-valued $%
3 $-form $\varphi _{Z}$ by 
\begin{equation}
\varphi _{Z}=\sigma _{Z}\left( \varphi \right)
\end{equation}%
Similarly as for real octonions, this will still be compatible with $g$, but
will define a new $G_{2}^{\mathbb{C}}$-structure. This is now invariant
under $Z\longrightarrow fZ$ for any nowhere-vanishing $\mathbb{C}$-valued
function $f$, and hence this explicitly shows that isometric $G_{2}^{\mathbb{%
C}}$-structures are pointwise parametrized by $\mathbb{C}P^{7}\cong SO\left(
n,\mathbb{C}\right) /G_{2}^{\mathbb{C}}$.

It is easy to see that all the standard $G_{2}$-structure identities are
satisfied by $\varphi _{Z}$ and its corresponding $4$-form $\psi _{Z}.$
Using $\varphi _{Z}$ we can define a new product structure on $\mathbb{O}_{%
\mathbb{C}}M$ in the same way as (\ref{OctoVAB}):%
\begin{equation}
A\circ _{Z}B:=\left( AZ\right) \left( Z^{-1}B\right) =AB+\left[ A,B,Z\right]
Z^{-1}  \label{defAZB}
\end{equation}%
The product $\circ _{Z}$ for complex $Z$ satisfies the same properties as
for $\mathbb{C}$-real $Z,$ namely an important property is the following:

\begin{lemma}
\label{LemABZ}Suppose $A,B,Z\in \Gamma \left( \mathbb{O}_{\mathbb{C}%
}M\right) $, such that $Z$ is nowhere a zero divisor, then 
\begin{equation}
A\left( BZ\right) =\left( A\circ _{Z}B\right) Z  \label{ABZ}
\end{equation}
\end{lemma}

The irreducible representations of $G_{2}^{\mathbb{C}}$ are just the
complexifications of the corresponding representations of the real form of $%
G_{2}$, so we still have the same decompositions of (complexified)
differential forms. We will denote the relevant complexified representations
with a superscript $\mathbb{C}.$

Finally, we can also define the torsion $T^{\left( Z\right) }$ of $\varphi
^{\left( Z\right) }$ as in (\ref{TorsV2}) 
\begin{equation}
T^{\left( Z\right) }=-\left( DZ\right) Z^{-1}+\frac{1}{2}\frac{1}{N\left(
Z\right) }d\left( N\left( Z\right) \right) \hat{1}  \label{TorsZ}
\end{equation}%
With respect to representations of $G_{2}^{\mathbb{C}}$, this will have a
similar decomposition into components as the torsion of a real $G_{2}$%
-structure, with $d\varphi _{Z}$ and $d\psi _{Z}$ satisfying a complexified
version of (\ref{dphidpsi}): 
\begin{subequations}%
\label{dphidpsiz}

\begin{eqnarray}
d\varphi _{Z} &=&8\tau _{1}^{\left( Z\right) }\psi _{Z}-6\tau _{7}^{\left(
Z\right) }\wedge \varphi _{Z}+8\iota _{\psi _{Z}}\left( \tau _{27}^{\left(
Z\right) }\right)  \label{dphiz} \\
d\psi _{Z} &=&-8\tau _{7}^{\left( Z\right) }\wedge \psi _{Z}-4\ast \tau
_{14}^{\left( Z\right) }.  \label{dpsiz}
\end{eqnarray}

\end{subequations}
where $\tau _{1}^{\left( Z\right) },\tau _{7}^{\left( Z\right) },\tau
_{14}^{\left( Z\right) },$ and $\tau _{27}^{\left( Z\right) }$ are the
components of $T^{\left( Z\right) }$ in representations $\mathbf{1}^{\mathbb{%
C}},\mathbf{7}^{\mathbb{C}},\mathbf{14}^{\mathbb{C}},$ and $\mathbf{27}^{%
\mathbb{C}}$, respectively, with respect to the $G_{2}^{\mathbb{C}}$%
-structure $\varphi _{Z}$.

\section{Supersymmetry equations in terms of octonions}

\setcounter{equation}{0}\label{secoctosusy}Recall that the supersymmetry
equations in 7 dimensions for a spinor $\theta $ are given by 
\begin{subequations}
\begin{eqnarray}
0 &=&\lambda e^{-\Delta }\theta ^{\ast }+\left[ \left( \mu ie^{-4\Delta }+%
\frac{1}{2}\left( \partial _{c}\Delta \right) \gamma ^{c}\right) +\frac{1}{%
144}G_{bcde}\gamma ^{bcde}\right] \theta \\
\nabla _{a}^{\mathcal{S}}\theta &=&\left[ \frac{i}{2}\mu e^{-4\Delta }\gamma
_{a}-\frac{1}{144}\gamma _{a}\left( G_{bcde}\gamma _{\ \ \ \ \ \ \
}^{bcde}\right) +\frac{1}{12}G_{abcd}\gamma ^{bcd}\right] \theta .
\end{eqnarray}

\end{subequations}
Suppose we have a fixed $G_{2}$-structure $\varphi $ that is defined by a
unit spinor $\xi .$ Then, using the ideas in Section \ref{sectTors}, we can
apply the map $j_{\xi }$ to the above equations, in order to obtain a set of
equations for a complex octonion $Z=$ $j_{\xi }\left( \theta \right) .$
We'll suppose that $Z=X+iY.$ 
\begin{subequations}%
\label{susyeqoct}%
\begin{eqnarray}
0 &=&\lambda e^{-\Delta }Z^{\ast }+i\left( \mu e^{-4\Delta }+\frac{1}{2}%
\left( \partial _{c}\Delta \right) \delta ^{c}\right) Z+\frac{1}{144}%
G_{bcde}\delta ^{b}\left( \delta ^{c}\left( \delta ^{d}\left( \delta
^{e}Z\right) \right) \right)  \label{susyeqoct1} \\
D_{a}Z &=&-\frac{1}{2}\mu e^{-4\Delta }\delta _{a}Z-\frac{i}{144}%
G_{bcde}\delta _{a}\left( \delta ^{b}\left( \delta ^{c}\left( \delta
^{d}\left( \delta ^{e}Z\right) \right) \right) \right) -\frac{i}{12}%
G_{abcd}\left( \delta ^{b}\left( \delta ^{c}\left( \delta ^{d}Z\right)
\right) \right)  \label{susyeqoct2}
\end{eqnarray}%
\end{subequations}%
Let $M=\mu e^{-4\Delta }+\frac{1}{2}\left( \widehat{\partial \Delta }\right)
=\left( 
\begin{array}{c}
\mu e^{-4\Delta } \\ 
\frac{1}{2}\partial \Delta%
\end{array}%
\right) ,$ and also define the following operators on $\Gamma \left( \mathbb{%
O}_{\mathbb{C}}M\right) $ 
\begin{subequations}
\begin{eqnarray}
-\frac{1}{144}G_{bcde}\delta ^{b}\left( \delta ^{c}\left( \delta ^{d}\left(
\delta ^{e}A\right) \right) \right) &=&G^{\left( 4\right) }\left( A\right)
\label{defG4} \\
-\frac{1}{144}G_{abcd}\left( \delta ^{b}\left( \delta ^{c}\left( \delta
^{d}A\right) \right) \right) &=&G_{a}^{\left( 3\right) }\left( A\right)
\label{defG3}
\end{eqnarray}%
\end{subequations}%
for any octonion $A.$ These operators satisfy some properties that are an
easy consequence of the Clifford algebra identity (\ref{octoenvelopCliff2}).

\begin{lemma}
\label{LemG3G4}The operators $G^{\left( 4\right) }$ and $G^{\left( 3\right)
} $ satisfy the following properties for any (bi)octonion sections $A$ and $%
B $ 
\begin{subequations}
\begin{eqnarray}
\delta _{a}\left( G^{\left( 4\right) }\left( A\right) \right) &=&G^{\left(
4\right) }\left( \delta _{a}A\right) -8G_{a}^{\left( 3\right) }\left(
A\right) \\
\left\langle G^{\left( 4\right) }\left( A\right) ,B\right\rangle
&=&\left\langle A,G^{\left( 4\right) }\left( B\right) \right\rangle \\
\left\langle G_{a}^{\left( 3\right) }\left( A\right) ,B\right\rangle
&=&\left\langle A,G_{a}^{\left( 3\right) }\left( B\right) \right\rangle \\
\left\langle \delta _{a}G^{\left( 4\right) }\left( A\right) ,B\right\rangle
&=&\left\langle \delta _{a}A,G^{\left( 4\right) }\left( B\right)
\right\rangle -8\left\langle A,G_{a}^{\left( 3\right) }\left( B\right)
\right\rangle \\
\left\langle \delta _{a}G^{\left( 4\right) }\left( A\right) ,A\right\rangle
&=&-4\left\langle A,G_{a}^{\left( 3\right) }\left( A\right) \right\rangle
\end{eqnarray}%
\end{subequations}%
\end{lemma}

Now we can rewrite (\ref{susyeqoct}) more succinctly as 
\begin{subequations}%
\label{susyeqoct1a2a}%
\begin{eqnarray}
G^{\left( 4\right) }\left( Z\right) &=&iMZ+\lambda e^{-\Delta }Z^{\ast }
\label{susyeqoct1a} \\
D_{a}Z &=&-\frac{1}{2}\mu e^{-4\Delta }\delta _{a}Z+i\delta _{a}G^{\left(
4\right) }\left( Z\right) +12iG_{a}^{\left( 3\right) }\left( Z\right)
\label{susyeqoct2a}
\end{eqnarray}

\end{subequations}%
Recall that given a $G_{2}$-structure $\varphi $, we can decompose the $4$%
-form $G$ according to the representations of $G_{2}$ as%
\begin{eqnarray}
G &=&Q_{0}\psi +Q_{1}^{\flat }\wedge \varphi _{X}+\iota _{\psi _{X}}\left(
Q_{2}\right)  \label{Gdecomp} \\
&=&\iota _{\psi _{X}}\left( Q_{0}g-\frac{4}{3}Q_{1}\lrcorner \varphi
_{Z}+4Q_{2}\right)  \label{Gdecomp2}
\end{eqnarray}%
where $Q_{0}$ is a scalar function that represents the $1$-dimensional
component of $G,$ $Q_{1}$ is a vector that represents the $7$-dimensional
component, and $Q_{2}$ is a symmetric $2$-tensor that represents the $27$%
-dimensional component. Then, a straightforward calculation, gives us $%
G^{\left( 4\right) }\left( 1\right) $ and $G^{\left( 3\right) }\left(
1\right) $ in terms of these components.

\begin{lemma}
\label{LemG4G3}In terms of the decomposition (\ref{Gdecomp}) of the $4$-form 
$G$, the quantities $G^{\left( 4\right) }$ and $G^{3}$ are given by 
\begin{subequations}%
\label{G4G3Q} 
\begin{eqnarray}
G^{\left( 4\right) }\left( 1\right) &=&\left( 
\begin{array}{c}
\frac{7}{6}Q_{0} \\ 
\frac{2}{3}Q_{1}%
\end{array}%
\right)  \label{G41} \\
G^{\left( 3\right) }\left( 1\right) &=&\left( 
\begin{array}{c}
\frac{1}{6}Q_{1} \\ 
-\frac{1}{8}Q_{0}\delta +\frac{1}{8}Q_{1}\lrcorner \varphi +\frac{1}{12}Q_{2}%
\end{array}%
\right)  \label{G31}
\end{eqnarray}%
\end{subequations}%
\end{lemma}

Suppose $Z$ is not a zero divisor, that is, suppose we are working in a
neighborhood where $N\left( Z\right) \neq 0$. We will justify this
assumption later on. In that case, we can define a complexified $G_{2}$%
-structure $\varphi _{Z}=\sigma _{Z}\left( \varphi \right) $. Similarly to (%
\ref{Gdecomp}), with respect to $\varphi _{Z}$ we can decompose the $4$-form 
$G$ according to the representations of $G_{2}$ as%
\begin{equation}
G=Q_{0}^{\left( Z\right) }\psi _{Z}+Q_{1}^{\left( Z\right) }\wedge \varphi
_{Z}+\iota _{\psi _{Z}}\left( Q_{2}^{\left( Z\right) }\right)
\label{GtotalZ}
\end{equation}%
where the components $Q_{i}^{\left( Z\right) }$ will in general now be
complex-valued, even though $G$ is of course still $\mathbb{C}$-real.
Applying Lemma \ref{LemABZ}, The equation (\ref{susyeqoct2a}) can be
rewritten with respect to $\varphi _{Z}$ as 
\begin{equation}
DZ=-\frac{1}{2}\mu e^{-4\Delta }\delta Z+i\left( \delta \circ _{Z}G^{\left(
4\right) \left( Z\right) }\left( 1\right) \right) Z+12i\left( G^{\left(
3\right) \left( Z\right) }\left( 1\right) \right) Z,  \label{DZ2}
\end{equation}%
where $G^{\left( 4\right) \left( Z\right) }$ and $G^{\left( 3\right) \left(
Z\right) }$ are the operators $G^{\left( 4\right) }$ and $G^{\left( 3\right)
}$ with respect to the $G_{2}$-structure $\varphi _{Z}$. Now applying Lemma %
\ref{LemG4G3} to $G^{\left( 4\right) \left( Z\right) }\left( 1\right) $ and $%
G^{\left( 3\right) \left( Z\right) }\left( 1\right) $ we see that 
\begin{equation}
DZ=i\left( 
\begin{array}{c}
\frac{4}{3}Q_{1}^{\left( Z\right) } \\ 
\left( \frac{i}{2}\mu e^{-4\Delta }-\frac{5}{6}Q_{0}^{\left( Z\right)
}\right) \delta +\frac{1}{3}Q_{1}^{\left( Z\right) }\lrcorner \varphi
_{Z}-Q_{2}^{\left( Z\right) }%
\end{array}%
\right) Z\ .  \label{DZ3}
\end{equation}%
Taking the inner product of (\ref{DZ3}) with $Z,$ we obtain 
\begin{equation}
Q_{1}^{\left( Z\right) }=-\frac{3}{8}\frac{i}{N\left( Z\right) }d\left(
N\left( Z\right) \right)  \label{Q1Z}
\end{equation}%
and moreover, from (\ref{TorsZ}) we find that the torsion of $\varphi _{Z}$
is given by 
\begin{equation}
T^{\left( Z\right) }=\left( \frac{1}{2}\mu e^{-4\Delta }+\frac{5}{6}%
iQ_{0}^{\left( Z\right) }\right) g-\frac{i}{3}Q_{1}^{\left( Z\right)
}\lrcorner \varphi _{Z}+iQ_{2}^{\left( Z\right) }.  \label{TZ3}
\end{equation}%
Hence, the internal equation (\ref{susyeqoct2a}) gives us the following
information regarding the components of $T^{\left( Z\right) }$:%
\begin{subequations}%
\label{TZs}%
\begin{eqnarray}
\tau _{1}^{\left( Z\right) } &=&\frac{1}{2}\mu e^{-4\Delta }+\frac{5}{6}%
iQ_{0}^{\left( Z\right) }  \label{tau1Z} \\
\tau _{7}^{\left( Z\right) } &=&-\frac{i}{3}Q_{1}^{\left( Z\right) }=-\frac{1%
}{8}\frac{1}{N\left( Z\right) }d\left( N\left( Z\right) \right)
\label{tau7Z} \\
\tau _{14}^{\left( Z\right) } &=&0 \\
\tau _{27}^{\left( Z\right) } &=&iQ_{2}^{\left( Z\right) }.  \label{tau27Z}
\end{eqnarray}%
\end{subequations}
In particular, the $14$-dimensional component of the torsion vanishes and $%
\tau _{7}^{\left( Z\right) }$ is an exact $1$-form. Note that the $G_{2}$%
-structure $\varphi _{Z}$, and hence its torsion, depend on $Z$ only
projectively, that is they are invariant under scalings $Z\longrightarrow
\kappa Z$ for any nowhere-vanishing complex-valued functions $\kappa $. On
the other hand, the original equation (\ref{susyeqoct2}) is only invariant
under constant scalings. The expression (\ref{Q1Z}) then determines $N\left(
Z\right) $ up to a constant factor.

Using (\ref{Gdecomp2}) we can write the $4$-form $G$ in terms of the linear
operator $\iota _{\psi _{Z}}$ as%
\begin{equation*}
G=\iota _{\psi _{Z}}\left( Q_{0}^{\left( Z\right) }g-\frac{4}{3}%
Q_{1}^{\left( Z\right) }\lrcorner \varphi _{Z}+4Q_{2}^{\left( Z\right)
}\right)
\end{equation*}%
and therefore, using (\ref{TZ3}), 
\begin{equation}
G=-4i\iota _{\psi _{Z}}\left( T^{\left( Z\right) }\right) +\left( 2i\mu
e^{-4\Delta }-\frac{7}{3}Q_{0}^{\left( Z\right) }\right) \psi _{Z}\ ,
\end{equation}%
but from the definition (\ref{codiffphi}) of the torsion, 
\begin{equation}
d\varphi _{Z}=8\iota _{\psi _{Z}}\left( T^{\left( Z\right) }\right) .
\label{dphiZTpsi}
\end{equation}%
However, in terms of $\tau _{1}^{\left( Z\right) },$ 
\begin{equation*}
2i\mu e^{-4\Delta }-\frac{7}{3}Q_{0}^{\left( Z\right) }=\frac{14}{5}i\tau
_{1}^{\left( Z\right) }+\frac{3}{5}i\mu e^{-4\Delta }.
\end{equation*}%
Therefore, 
\begin{equation}
G=-\frac{1}{2}id\varphi _{Z}+\frac{i}{5}\left( 14\tau _{1}^{\left( Z\right)
}+3\mu e^{-4\Delta }\right) \psi _{Z}\ .  \label{GZexp}
\end{equation}%
We can summarize our findings so far.

\begin{theorem}
\label{ThmTorsZ}Let $\left( \varphi ,g\right) \ $be a real $G_{2}$-structure
on a $7$-manifold $\dot{M}.$ Then, a non-zero divisor section $Z\in \Gamma
\left( \mathbb{O}_{\mathbb{C}}M\right) $ satisfies the equation (\ref%
{susyeqoct2a}) for a $4$-form $G$, real constant $\mu $, and a warp
parameter $\Delta $, if the torsion of the complexified $G_{2}$-structure $%
\varphi _{Z}=\sigma _{Z}\left( \varphi \right) $ satisfies (\ref{TZs}) and $%
N\left( Z\right) $ satisfies (\ref{Q1Z}). Conversely, if $\varphi _{Z}$ is
in the torsion class $\mathbf{1}^{\mathbb{C}}\oplus \mathbf{7}^{\mathbb{C}%
}\oplus \mathbf{27}^{\mathbb{C}}$, with $\tau _{7}^{\left( Z\right) }=$ $-%
\frac{1}{8}\frac{1}{N\left( Z\right) }d\left( N\left( Z\right) \right) $,
and the $4$-form $G=-\frac{1}{2}id\varphi _{Z}+\frac{i}{5}\left( 14\tau
_{1}^{\left( Z\right) }+3\mu e^{-4\Delta }\right) \psi _{Z}$ is $\mathbb{C}$%
-real, then $Z$ satisfies the equation (\ref{susyeqoct2a}) with $4$-form $G$%
, constant $\mu $, and warp parameter $\Delta .$
\end{theorem}

From (\ref{GZexp}), (\ref{tau7Z}), and (\ref{dpsi}), we see that $dG=0$ if
and only if 
\begin{equation}
d\left( \left( 14\tau _{1}^{\left( Z\right) }+3\mu e^{-4\Delta }\right)
N\left( Z\right) \right) =0.  \label{DG0cond}
\end{equation}

So far we have only used equation (\ref{susyeqoct2a}). Using (\ref%
{susyeqoct1a}) will give us additional information. First, let us take the
inner product of (\ref{susyeqoct1a}) with the $\mathbb{C}$-conjugate $%
Z^{\ast }=X-iY$. Then, we get 
\begin{equation}
\left\langle G^{\left( 4\right) }\left( Z\right) ,Z^{\ast }\right\rangle
=i\left\langle MZ,Z^{\ast }\right\rangle +\left\langle \lambda e^{-\Delta
}Z^{\ast },Z^{\ast }\right\rangle  \label{MZZst}
\end{equation}%
However since $G$ is a real $4$-form, the left hand side of (\ref{MZZst}) is 
$\mathbb{C}$-real, so the $\mathbb{C}$-imaginary part of the right hand side
must vanish. Now, $\left\langle MZ,Z^{\ast }\right\rangle =\left\langle
M,Z^{\ast }\bar{Z}\right\rangle $ and $M$ is $\mathbb{C}$-real. Moreover,
from (\ref{ZZstbar}) we see that that the $\mathbb{C}$-real part of $Z^{\ast
}\bar{Z}$ is $\left\vert X\right\vert ^{2}+\left\vert Y\right\vert ^{2}$,
which is $\mathbb{O}$-real, so we find 
\begin{equation}
\mu e^{-4\Delta }\left( \left\vert X\right\vert ^{2}+\left\vert Y\right\vert
^{2}\right) +\func{Im}_{\mathbb{C}}\left( \lambda e^{-\Delta }N\left(
Z\right) ^{\ast }\right) =0  \label{muXsqYsq}
\end{equation}%
In particular, we can rewrite (\ref{susyeqoct1a}) as%
\begin{equation}
G^{\left( 4\right) \left( Z\right) }\left( 1\right) =iM+\lambda e^{-\Delta
}Z^{\ast }Z^{-1}.  \label{susyeqoct1a2}
\end{equation}%
Now from (\ref{ZZstbar}) we have 
\begin{equation*}
Z^{\ast }\overline{Z}=\left( \left\vert X\right\vert ^{2}+\left\vert
Y\right\vert ^{2}\right) -i\left( Y\bar{X}-X\bar{Y}\right) =\left(
\left\vert X\right\vert ^{2}+\left\vert Y\right\vert ^{2}\right) -2i\func{Im}%
_{\mathbb{O}}\left( Y\bar{X}\right) .
\end{equation*}%
Therefore, together with (\ref{G41}), (\ref{susyeqoct1a2}) gives us 
\begin{equation}
\left( 
\begin{array}{c}
\frac{7}{6}Q_{0}^{\left( Z\right) } \\ 
\frac{2}{3}Q_{1}^{\left( Z\right) }%
\end{array}%
\right) =\left( 
\begin{array}{c}
i\mu e^{-4\Delta }+\frac{\lambda e^{-\Delta }}{N\left( Z\right) }\left(
\left\vert X\right\vert ^{2}+\left\vert Y\right\vert ^{2}\right) \\ 
\frac{1}{2}i\partial \Delta -\frac{2\lambda ie^{-\Delta }}{N\left( Z\right) }%
\func{Im}_{\mathbb{O}}\left( Y\bar{X}\right)%
\end{array}%
\right) .  \label{Q0Q1Z}
\end{equation}%
In particular, comparing with (\ref{Q1Z}), we obtain 
\begin{equation}
dN\left( Z\right) =-2\left( d\Delta \right) N\left( Z\right) +8\lambda
e^{-\Delta }\func{Im}_{\mathbb{O}}\left( Y\bar{X}\right)  \label{dNZ}
\end{equation}%
We obtained (\ref{dNZ}) by switching to the $Z$-frame, and hence implicitly
assuming that $N\left( Z\right) $ is non-zero. However, (\ref{dNZ}) can also
be obtained by directly considering $d\left\langle Z,Z\right\rangle $. From (%
\ref{dNZ}) we can also find that 
\begin{equation}
d\left\vert N\left( Z\right) \right\vert ^{2}=-4\left( d\Delta \right)
\left\vert N\left( Z\right) \right\vert ^{2}+16\func{Re}\left( \lambda
^{\ast }N\left( Z\right) \right) e^{-\Delta }\func{Im}_{\mathbb{O}}\left( Y%
\bar{X}\right)  \label{dNZsq}
\end{equation}

Note that if $\lambda \neq 0$, (\ref{dNZ}) can be written as 
\begin{equation}
d\left( e^{2\Delta }\lambda ^{\ast }N\left( Z\right) \right) =8\left\vert
\lambda \right\vert ^{2}e^{\Delta }\func{Im}_{\mathbb{O}}\left( Y\bar{X}%
\right)   \label{dNZ2a}
\end{equation}%
Since the right-hand side of (\ref{dNZ2a}) is $\mathbb{C}$-real, this then
shows that 
\begin{equation*}
d\left( \func{Im}_{\mathbb{C}}\left( e^{2\Delta }\lambda ^{\ast }N\left(
Z\right) \right) \right) =0
\end{equation*}%
Together with (\ref{muXsqYsq}) this then shows that if $\mu \neq 0,$ then $%
\left( \left\vert X\right\vert ^{2}+\left\vert Y\right\vert ^{2}\right)
e^{-\Delta }$ is constant. Without loss of generality we can then say that 
\begin{subequations}
\begin{eqnarray}
\left\vert X\right\vert ^{2}+\left\vert Y\right\vert ^{2} &=&e^{\Delta }
\label{normxsqysq} \\
\func{Im}_{\mathbb{C}}\left( e^{2\Delta }\lambda ^{\ast }N\left( Z\right)
\right)  &=&\mu .  \label{dNZ2b}
\end{eqnarray}%
\end{subequations}%
If $\mu =0,$ then (\ref{dNZ2b}) still holds, however (\ref{normxsqysq}) will
not follow from (\ref{muXsqYsq}), but we can still reach the same conclusion
by considering $\nabla \left\langle Z,Z^{\ast }\right\rangle .$ This is
shown in Lemma \ref{LemNablaZZst} in the Appendix. Comparing, for example,
with equation (2.16) in \cite{SparksFlux}, if there we set $\chi _{-}=\chi
_{+}$ to reduce to $N=1$ supersymmetry, then (\ref{normxsqysq}) is
equivalent to the equation $\bar{\chi}_{+}\chi =1$ and $N\left( Z\right) $
is equivalent to the quantity $S$.

Differentiating (\ref{dNZ2a}) one more time (and assuming $\lambda \neq 0$),
we obtain 
\begin{equation}
d\left( e^{\Delta }\func{Im}_{\mathbb{O}}\left( Y\bar{X}\right) \right) =0.
\label{ddyx0}
\end{equation}

If we write $N\left( Z\right) =\rho e^{i\theta }$ for some positive
real-valued function $\rho $ and a real-valued function $\theta $, then from
(\ref{dNZ}), (\ref{dNZ2b}), and (\ref{dNZsq}) we obtain 
\begin{subequations}%
\begin{eqnarray}
d\left( e^{2\Delta }\rho \right)  &=&8\left( \lambda _{1}\cos \theta
+\lambda _{2}\sin \theta \right) e^{\Delta }\func{Im}_{\mathbb{O}}\left( Y%
\bar{X}\right)   \label{drho} \\
d\theta  &=&-\frac{8\mu e^{-3\Delta }}{\rho ^{2}}\func{Im}_{\mathbb{O}%
}\left( Y\bar{X}\right) .  \label{dtheta}
\end{eqnarray}

\end{subequations}
In particular, this shows that $\mu =0$ implies a constant phase factor $%
\theta $ for $N\left( Z\right) $. By a change of basis of Killing spinors on 
$\func{AdS}_{4}$, we may assume in that case that $N\left( Z\right) =\rho $
and $\lambda =\lambda _{1}$.

\begin{remark}
\label{RemarkNZ0}For $\lambda \neq 0,$ (\ref{dNZ2b}) also shows that if $\mu
\neq 0$, then $N\left( Z\right) $ is nowhere-vanishing. If $\mu =0$, then $%
N\left( Z\right) $ may vanish at a point. However in (\ref{dNZ}) suppose $%
\lambda $ is nonzero and $N\left( Z\right) =0$ at a point. Then, at that
point $\func{Im}_{\mathbb{O}}\left( Y\bar{X}\right) \neq 0$ (otherwise $%
X=Y=0 $) and hence $dN\left( Z\right) \neq 0$. Therefore, $N\left( Z\right) $
cannot be identically zero in a neighborhood. So for nonzero $\lambda $, we
conclude that $N\left( Z\right) \neq 0$ almost everywhere. On the other
hand, if $\lambda =0,$ i.e. if the $4$-dimensional space is Minkowski then (%
\ref{dNZ2b}) forces $\mu =0$ and (\ref{dNZ}) shows that $N\left( Z\right)
=ke^{-2\Delta }$ for some $k=k_{1}+ik_{2}\in \mathbb{C},$ and is thus either
nowhere-vanishing or zero everywhere. The case where the $4$-dimensional
space is Minkowski and $N\left( Z\right) =0$ everywhere is precisely the
case that was considered in detail in \cite{KasteMinasianFlux}. Therefore,
in all cases except this one, we may assume that $N\left( Z\right) \neq 0$
at least locally. We will also assume that $\lambda \neq 0,$ unless
specified otherwise. Whenever $N\left( Z\right) \neq 0$, from (\ref{dtheta})
we have $e^{i\theta }$ constant in the case $\mu =0$. However, by continuity
of $dN\left( Z\right) ,$ we find that $e^{i\theta }$ has to be constant
everywhere. Hence, we can still apply a change of basis of spinors to assume
that $N\left( Z\right) =\rho $ (with $\rho \geq 0$) and $\lambda =\lambda
_{1}$.
\end{remark}

Now using (\ref{Q0Q1Z}) and (\ref{normxsqysq}) we can rewrite $\tau
_{1}^{\left( Z\right) }$ as 
\begin{equation}
\tau _{1}^{\left( Z\right) }=-\frac{3}{14}\mu e^{-4\Delta }+\frac{5}{7}\frac{%
\lambda i}{N\left( Z\right) }  \label{tau1Z2}
\end{equation}%
and the expression (\ref{GZexp}) for $G$ can now be rewritten as 
\begin{equation}
G=-\frac{1}{2}id\varphi _{Z}-\frac{2\lambda }{N\left( Z\right) }\psi _{Z}.
\label{Gexp}
\end{equation}%
This now allows to verify that $dG=0$, since (\ref{tau1Z2}) gives us 
\begin{equation}
\left( 14\tau _{1}^{\left( Z\right) }+3\mu e^{-4\Delta }\right) N\left(
Z\right) =10\lambda i,  \label{lambdaNZ}
\end{equation}%
and hence (\ref{DG0cond}) is satisfied.

Note that since $G$ is real, the $\mathbb{C}$-imaginary part of (\ref{Gexp})
must vanish. Hence we get:%
\begin{equation}
\func{Im}_{\mathbb{C}}\left( \frac{\lambda }{N\left( Z\right) }\psi
_{Z}\right) =-\frac{1}{2}d\left( \func{Re}\varphi _{Z}\right) .
\end{equation}
Moreover, recall that the cohomology class $\left[ G\right] $ of the $4$%
-form flux must satisfy the flux quantization condition \cite{Witten:1996md}:%
\begin{equation}
\frac{\left[ G\right] }{2\pi }-\frac{p_{1}\left( X\right) }{4}\in
H^{4}\left( X,\mathbb{Z}\right) .
\end{equation}%
This implies that  $\left[ \func{Re}_{\mathbb{C}}\left( \frac{\lambda }{\pi
N\left( Z\right) }\psi _{Z}\right) \right] +\frac{p_{1}\left( X\right) }{4}$
must be a integer cohomology class.

If $\mu =0$, we will not be able to define $\varphi _{Z}$ and $T^{\left(
Z\right) }$ whenever $N\left( Z\right) =0$. However in neighborhoods where $%
N\left( Z\right) \neq 0$ we can still do this and obtain $dG=0$. On the
other hand, if $\lambda \neq 0$, we know $N\left( Z\right) $ cannot vanish
on any open set, hence by continuity must have $dG=0$ everywhere. We can now
extend Theorem \ref{ThmTorsZ} with conditions to satisfy equation (\ref%
{susyeqoct1a}).

\begin{theorem}
\label{ThmTorsZ2}Let $\left( \varphi ,g\right) \ $be a real $G_{2}$%
-structure on a $7$-manifold $\dot{M}.$ Then, a non-zero divisor section $%
Z=X+iY\in \Gamma \left( \mathbb{O}_{\mathbb{C}}M\right) $ satisfies
equations (\ref{susyeqoct1a}) and (\ref{susyeqoct2a}) for a closed $4$-form $%
G$, real constant $\mu $, complex constant $\lambda $, and a warp parameter $%
\Delta $, if the torsion of the complexified $G_{2}$-structure $\varphi
_{Z}=\sigma _{Z}\left( \varphi \right) $ satisfies the conditions of Theorem %
\ref{ThmTorsZ} and moreover $Z$ satisfies%
\begin{subequations}%
\begin{eqnarray}
d\left( e^{2\Delta }N\left( Z\right) \right)  &=&8\lambda e^{\Delta }\func{Im%
}_{\mathbb{O}}\left( Y\bar{X}\right)   \label{dNZimYX} \\
\func{Im}_{\mathbb{C}}\left( e^{2\Delta }\lambda ^{\ast }N\left( Z\right)
\right)  &=&\mu   \label{muNZ} \\
\left\vert X\right\vert ^{2}+\left\vert Y\right\vert ^{2} &=&e^{\Delta }
\label{expdeltaXY} \\
d\left( e^{\Delta }\func{Im}_{\mathbb{O}}\left( Y\bar{X}\right) \right)  &=&0
\end{eqnarray}%
\end{subequations}%
for some real-valued function $f\left( \Delta \right) $. Conversely, if $%
\varphi _{Z}$ is in the torsion class $\mathbf{1}^{\mathbb{C}}\oplus \mathbf{%
7}^{\mathbb{C}}\oplus \mathbf{27}^{\mathbb{C}}$, with $\tau _{7}^{\left(
Z\right) }=$ $-\frac{1}{8}\frac{1}{N\left( Z\right) }d\left( N\left(
Z\right) \right) $, the $4$-form $G=-\frac{1}{2}id\varphi _{Z}+\frac{i}{5}%
\left( 14\tau _{1}^{\left( Z\right) }+3\mu e^{-4\Delta }\right) \psi _{Z}$
is $\mathbb{C}$-real, and $Z$ satisfies (\ref{dNZimYX}), then it satisfies
the equations (\ref{susyeqoct1a}) and (\ref{susyeqoct2a}) with the $4$-form $%
G$, constant $\mu $ defined by (\ref{muNZ}), constant $\lambda $ defined by (%
\ref{lambdaNZ}), and warp parameter $\Delta $ defined by (\ref{expdeltaXY}).
\end{theorem}

\subsection{Decomposition into real and imaginary parts}

To more concretely understand what are the properties of the corresponding
real $G_{2}$-structures, we need to decompose everything into $\mathbb{C}$%
-real and $\mathbb{C}$-imaginary parts. However for convenience, without
loss of generality, we may change the reference $G_{2}$-structure to $%
\varphi _{X}:=\sigma _{X}\left( \varphi \right) ,$ setting 
\begin{equation}
A=ZX^{-1}=1+iW  \label{AZX}
\end{equation}%
where we defined $W=w_{0}+\hat{w}:=YX^{-1}$. Then, using (\ref{DtildeAV3}), 
\begin{equation*}
DZ=D\left( AX\right) =\left( D^{\left( X\right) }A\right) X+\left( d\ln
\left\vert X\right\vert \right) AX
\end{equation*}%
and 
\begin{equation*}
i\delta G^{\left( 4\right) }\left( Z\right) +12iG^{\left( 3\right) }\left(
Z\right) =\left( i\delta \circ _{X}G^{\left( 4\right) \left( X\right)
}\left( A\right) +12iG^{\left( 3\right) \left( X\right) }\left( A\right)
\right) X
\end{equation*}%
where $D^{\left( X\right) },$ $\circ _{X},$ $G^{\left( 3\right) \left(
X\right) },$ and $G^{\left( 4\right) \left( X\right) }$ denote quantities
with respect to $\varphi _{X}.$ From now on, we'll drop the $X$ subscripts
and superscripts, since we will take $\varphi _{X}$ to be the standard $%
G_{2} $-structure. Overall, the equations (\ref{susyeqoct1a2a}) are
equivalent to%
\begin{subequations}%
\label{susyeqocts} 
\begin{eqnarray}
G^{\left( 4\right) }\left( A\right) &=&iMA+\lambda e^{-\Delta }A^{\ast }
\label{susyeqoct1b} \\
DA &=&-\frac{1}{2}\mu e^{-4\Delta }\delta A+i\delta G^{\left( 4\right)
}\left( A\right) +12iG^{\left( 3\right) }\left( A\right) -\left( d\ln
\left\vert X\right\vert \right) A.  \label{susyeqoct2c}
\end{eqnarray}%
\end{subequations}%
Note that the choice of the $G_{2}$-structure as $\varphi _{X}$ is somewhat
arbitrary - we could have alternatively chosen to work with $\varphi _{Y}$
and then defined $B=iZY^{-1}=\left( 1-iW^{-1}\right) $. The corresponding
equations for $B$ would then be 
\begin{subequations}%

\begin{eqnarray}
G^{\left( 4\right) }\left( B\right) &=&iMB-\lambda e^{-\Delta }B^{\ast } \\
D_{a}B &=&-\frac{1}{2}\mu e^{-4\Delta }\delta _{a}B+i\delta _{a}G^{\left(
4\right) }\left( B\right) +12iG_{a}^{\left( 3\right) }\left( B\right)
-\left( d\ln \left\vert Y\right\vert \right) B.
\end{eqnarray}%
\end{subequations}%
Therefore, to obtain corresponding results for $Y,$ we just need to perform
the transformation 
\begin{equation}
\left\{ X\longrightarrow Y,W\longrightarrow -W^{-1},\lambda \longrightarrow
-\lambda \right\} .
\end{equation}

From (\ref{nablazzst}) and (\ref{dNZ}), we find that 
\begin{subequations}
\begin{eqnarray}
d\left\vert X\right\vert ^{2} &=&\frac{1}{2}\left( 3\left\vert Y\right\vert
^{2}-\left\vert X\right\vert ^{2}\right) d\Delta +4\lambda _{1}e^{-\Delta }%
\func{Im}_{\mathbb{O}}\left( Y\bar{X}\right) \\
d\left\vert Y\right\vert ^{2} &=&\frac{1}{2}\left( 3\left\vert X\right\vert
^{2}-\left\vert Y\right\vert ^{2}\right) d\Delta -4\lambda _{1}e^{-\Delta }%
\func{Im}_{\mathbb{O}}\left( Y\bar{X}\right) .
\end{eqnarray}%
\end{subequations}%
However, $\func{Im}_{\mathbb{O}}\left( Y\bar{X}\right) =\left\vert
X\right\vert ^{2}w$, thus, 
\begin{subequations}
\begin{eqnarray}
d\ln \left\vert X\right\vert &=&\frac{1}{4}\left( 3\left\vert W\right\vert
^{2}-1\right) d\Delta +2\lambda _{1}e^{-\Delta }w  \label{dlnX} \\
d\ln \left\vert Y\right\vert &=&\left\vert W\right\vert ^{-2}\left( \frac{1}{%
4}\left( 3-\left\vert W\right\vert ^{2}\right) d\Delta -2\lambda
_{1}e^{-\Delta }w\right)  \label{dlnY}
\end{eqnarray}%
\end{subequations}%
and hence, 
\begin{equation}
d\left\vert W\right\vert ^{2}=\frac{3}{2}\left( 1-\left\vert W\right\vert
^{4}\right) d\Delta -4\lambda _{1}e^{-\Delta }\left( 1+\left\vert
W\right\vert ^{2}\right) w.  \label{dWsq}
\end{equation}%
From (\ref{dWsq}), we see that necessary conditions for $\left\vert
W\right\vert ^{2}\equiv 1$ are $\lambda _{1}=0$ or $w=0$.

Now, (\ref{dWsq}) can be used to rewrite (\ref{dlnX}) and (\ref{dlnY}) as%
\begin{subequations}%

\begin{eqnarray}
2d\ln \left\vert X\right\vert &=&d\Delta -d\left( \ln \left( 1+\left\vert
W\right\vert ^{2}\right) \right) \\
2d\ln \left\vert Y\right\vert &=&d\Delta +d\ln \frac{\left\vert W\right\vert
^{2}}{1+\left\vert W\right\vert ^{2}}.
\end{eqnarray}%
\end{subequations}%
Together with (\ref{normxsqysq}) and $\left\vert W\right\vert ^{2}=\frac{%
\left\vert Y\right\vert ^{2}}{\left\vert X\right\vert ^{2}}$ these equations
imply 
\begin{subequations}%
\begin{eqnarray}
\left\vert X\right\vert ^{2} &=&\frac{e^{\Delta }}{1+\left\vert W\right\vert
^{2}}  \label{normXsq} \\
\left\vert Y\right\vert ^{2} &=&\frac{e^{\Delta }\left\vert W\right\vert ^{2}%
}{1+\left\vert W\right\vert ^{2}}.
\end{eqnarray}%
\end{subequations}%

Using the $\mathbb{C}$-imaginary part of (\ref{dNZ}) and the fact that $%
w_{0}=\frac{\left\langle X,Y\right\rangle }{\left\vert X\right\vert ^{2}},$
we obtain an expression for $dw_{0}$: 
\begin{equation}
dw_{0}=-\frac{3}{2}\left( 1+\left\vert W\right\vert ^{2}\right) w_{0}d\Delta
-4\left( w_{0}\lambda _{1}-\lambda _{2}\right) e^{-\Delta }w.  \label{dw0}
\end{equation}%
This shows that a necessary condition for $w_{0}=0$ is $\lambda _{2}=0$.
Note that if $w=0$, then $\left\vert W\right\vert ^{2}=w_{0}^{2},$ so we
cannot have $w_{0}=0$. In fact, it is easy to see from (\ref{dWsq}) and (\ref%
{dw0}) that $w=0$ is only consistent when $d\Delta =0,$ and hence $w_{0}$ is
also constant. This case is known to reduce to the standard Freund-Rubin
solution on $S^{7}$ \cite{Behrndt:2003zg,CveticSUSYflux}, so we will not
consider it. Hence, we will assume that $w$ does not vanish identically. 

If $w_{0}\neq 0$ and $\left\vert W\right\vert ^{2}\neq 1,$ we can use (\ref%
{dWsq}) to rewrite (\ref{dw0}) as 
\begin{equation}
d\left( \ln \frac{w_{0}}{1-\left\vert W\right\vert ^{2}}\right) =4\left( 
\frac{\lambda _{2}}{w_{0}}-\frac{2\lambda _{1}}{1-\left\vert W\right\vert
^{2}}\right) e^{-\Delta }w.  \label{dlnw01}
\end{equation}%
However, from (\ref{dNZ2b}) and (\ref{normXsq}), 
\begin{equation}
2w_{0}\lambda _{1}-\lambda _{2}\left( 1-\left\vert W\right\vert ^{2}\right)
=\mu e^{-3\Delta }\left( 1+\left\vert W\right\vert ^{2}\right) .
\label{l1l2mu}
\end{equation}%
Hence, if $\lambda _{2}$ and $\mu $ are non-zero, we see that 
\begin{equation}
\left\vert W\right\vert ^{2}=\frac{\lambda _{2}+\mu e^{-3\Delta }-2\lambda
_{1}w_{0}}{\lambda _{2}-\mu e^{-3\Delta }}.  \label{W2exp}
\end{equation}

\begin{remark}
From (\ref{l1l2mu}), we see that if $w_{0}\equiv 0$, and hence $\lambda
_{2}=0$ (from (\ref{dw0})), then we must also have $\mu =0$. As we noted in
Remark \ref{RemarkNZ0}, we can assume the converse is also true, that is if $%
\mu =0$, then $\lambda _{2}=0$, and $N\left( Z\right) $ is real, which means 
$w_{0}=0$. Similarly, if $\left\vert W\right\vert ^{2}\equiv 1,$ and hence $%
\lambda _{1}=0$, then we must also have $\mu =0$. It should be emphasized
that (\ref{l1l2mu}) is a very important relation between $w_{0}$ and $%
\left\vert W\right\vert ^{2}$ that we will use over and over again.
\end{remark}

If $\left\vert W\right\vert ^{2}\neq 1$ and $w_{0}\neq 0$, then dividing (%
\ref{l1l2mu})\ by $w_{0}$ and $1-\left\vert W\right\vert ^{2}$, and
differentiating gives us 
\begin{equation}
d\left( \ln \frac{w_{0}}{1-\left\vert W\right\vert ^{2}}\right) =-\frac{4\mu
e^{-3\Delta }}{w_{0}}\frac{1+\left\vert W\right\vert ^{2}}{1-\left\vert
W\right\vert ^{2}}w.  \label{dlnw0}
\end{equation}%
By applying $d$ to (\ref{dlnw0}) and using (\ref{dWsq}), we find that in the
case when $w_{0}\neq 0$ and $\left\vert W\right\vert ^{2}\neq 1,$ 
\begin{equation}
d\left( \frac{4\mu e^{-\Delta }}{w_{0}}w\right) =0.  \label{dmuw}
\end{equation}

In the special case where $\lambda _{1}=0,$ then (\ref{dWsq}) immediately
gives 
\begin{equation}
\left\vert W\right\vert ^{2}=\frac{1-k_{1}e^{-3\Delta }}{1+k_{1}e^{-3\Delta }%
}  \label{normWsq}
\end{equation}%
for some constant $k_{1}$. Note that for non-zero $\lambda _{2}$ and $\mu $,
this is equivalent to (\ref{W2exp}) with $k_{1}=-\frac{\mu }{\lambda _{2}}$.

If on the other hand, $\lambda _{2}=0$, then from (\ref{dlnw01}) and (\ref%
{dWsq}), we find that for some constant $k_{2}$, 
\begin{equation}
w_{0}=k_{2}\left( 1+\left\vert W\right\vert ^{2}\right) e^{-3\Delta }
\label{normw0}
\end{equation}%
and then if $\lambda _{1}$ and $\mu $ are non-zero, then this is equivalent
to (\ref{W2exp}) with $k_{2}=\frac{\mu }{2\lambda _{1}}$.\ 

In the Minkowski case, when $\lambda _{1}=\lambda _{2}=0,$ combining (\ref%
{normWsq}) and (\ref{normw0}) gives us 
\begin{equation}
w_{0}=\frac{2k_{1}k_{2}e^{-3\Delta }}{1+k_{1}e^{-3\Delta }}.
\end{equation}

\subsection{External equation}

\setcounter{equation}{0}\label{secExteq}Recall that we have 
\begin{equation}
A=1+iW.
\end{equation}%
So that taking the $\mathbb{C}$-real and $\mathbb{C}$-imaginary parts of the
external equation (\ref{susyeqoct1b}), we obtain two equations 
\begin{subequations}%
\label{exteqs2} 
\begin{eqnarray}
G^{\left( 4\right) }\left( 1\right) &=&-MW+e^{-\Delta }\left( \lambda
_{1}+\lambda _{2}W\right) \\
G^{\left( 4\right) }\left( W\right) &=&M+e^{-\Delta }\left( \lambda
_{2}-\lambda _{1}W\right)
\end{eqnarray}%
\end{subequations}%
where $G^{\left( 4\right) }\left( 1\right) $ and $G^{\left( 4\right) }\left(
W\right) $ are now defined with respect to the $G_{2}$-structure $\varphi
_{X}.$ Let us now decompose the $4$-form $G$ with respect to $\varphi _{X}$
as:%
\begin{equation}
G=Q_{0}\psi _{X}+Q_{1}\wedge \varphi _{X}+\iota _{\psi _{X}}\left(
Q_{2}\right) .  \label{GtotalX}
\end{equation}

Taking $A=W=w_{0}+\hat{w}$ and then $A=1,$ the equations (\ref{exteqs2})
become 
\begin{subequations}%
\label{exteqs3}

\begin{eqnarray}
\left( 
\begin{array}{c}
\frac{7}{3}Q_{0} \\ 
\frac{4}{3}Q_{1}%
\end{array}%
\right) &=&\left( 
\begin{array}{c}
2e^{-\Delta }\left( \lambda _{2}-\mu e^{-3\Delta }\right) \\ 
-\partial \Delta%
\end{array}%
\right) W+2\lambda _{1}e^{-\Delta }  \label{exteqs3a} \\
&=&\left( 
\begin{array}{c}
2e^{-\Delta }\left( \lambda _{2}w_{0}+\lambda _{1}-\mu w_{0}e^{-3\Delta
}\right) +\left\langle \partial \Delta ,w\right\rangle \\ 
2e^{-\Delta }\left( \lambda _{2}-\mu e^{-3\Delta }\right) w-w_{0}\partial
\Delta -\partial \Delta \times w%
\end{array}%
\right)  \notag \\
\left( 
\begin{array}{c}
\frac{7}{3}w_{0}Q_{0}+\frac{4}{3}\left\langle Q_{1},w\right\rangle \\ 
\frac{4}{3}w_{0}Q_{1}-\frac{1}{3}Q_{0}w-\frac{4}{3}Q_{2}\left( w\right)%
\end{array}%
\right) &=&\left( 
\begin{array}{c}
2e^{-\Delta }\left( \lambda _{2}+\mu e^{-3\Delta }\right) \\ 
\partial \Delta%
\end{array}%
\right) -2\lambda _{1}e^{-\Delta }W \\
&=&\left( 
\begin{array}{c}
2e^{-\Delta }\left( \lambda _{2}-\lambda _{1}w_{0}+\mu e^{-3\Delta }\right)
\\ 
\partial \Delta -2\lambda _{1}e^{-\Delta }w%
\end{array}%
\right) .  \notag
\end{eqnarray}

\end{subequations}%
From this, we find the following relationships%
\begin{subequations}%
\label{Q0Q1rel2} 
\begin{eqnarray}
Q_{0} &=&\frac{3}{7}\left\langle \partial \Delta ,w\right\rangle +\frac{6}{7}%
e^{-\Delta }\left( \lambda _{2}w_{0}+\lambda _{1}-\mu w_{0}e^{-3\Delta
}\right)  \label{Q0rel} \\
\left\langle Q_{1},w\right\rangle &=&-\frac{7}{4}w_{0}Q_{0}+\frac{3}{2}%
e^{-\Delta }\left( \lambda _{2}+\mu e^{-3\Delta }-\lambda _{1}w_{0}\right)
\label{Q1wrel} \\
Q_{1} &=&-\frac{3}{4}w_{0}\partial \Delta -\frac{3}{4}\partial \Delta \times
w+\frac{3}{2}e^{-\Delta }\left( \lambda _{2}-\mu e^{-3\Delta }\right) w
\label{Q1rel} \\
Q_{2}\left( w\right) &=&w_{0}Q_{1}-\frac{1}{4}\left( Q_{0}-6\lambda
_{1}e^{-\Delta }\right) w-\frac{3}{4}\partial \Delta .  \label{Q2wrel}
\end{eqnarray}%
\end{subequations}%
Note that the expressions (\ref{Q1rel}) and (\ref{Q1wrel}) are in fact
compatible thanks to the relation (\ref{l1l2mu}).

In the Minkowski case, together with $w_{0}=0,$ these simplify to the same
relations as equation (3.12) in \cite{KasteMinasianFlux} once the slightly
different definitions of $Q_{0},Q_{1},Q_{2}$ are taken into account. In the $%
\func{AdS}_{4}$ case, if $\mu =\lambda _{2}=w_{0}=0,$ then these expressions
simplify: 
\begin{subequations}%
\label{Q0Q1rel2 copy(1)} 
\begin{eqnarray}
Q_{0} &=&\frac{3}{7}\left\langle \partial \Delta ,w\right\rangle +\frac{6}{7}%
\lambda _{1}e^{-\Delta } \\
\left\langle Q_{1},w\right\rangle  &=&0 \\
Q_{1} &=&-\frac{3}{4}\partial \Delta \times w  \label{Q1mu0} \\
Q_{2}\left( w\right)  &=&-\frac{1}{4}\left( Q_{0}-6\lambda _{1}e^{-\Delta
}\right) w-\frac{3}{4}\partial \Delta .
\end{eqnarray}%
\end{subequations}%

\section{Torsion}

\setcounter{equation}{0}\label{secInteq}Now consider the internal equation (%
\ref{susyeqoct2c}). Expanding $DA\ $in (\ref{susyeqoct2c}), we can
equivalently rewrite it as 
\begin{equation}
AT^{\left( X\right) }=\nabla A+\frac{1}{2}\mu e^{-4\Delta }\delta A-i\delta
G^{\left( 4\right) }\left( A\right) -12iG^{\left( 3\right) }\left( A\right)
+\left( d\ln \left\vert X\right\vert \right) A.  \label{susyoctnonzero4}
\end{equation}%
So that now taking the $\mathbb{C}$-real and $\mathbb{C}$-imaginary parts of
(\ref{susyoctnonzero4}) gives us two $\mathbb{C}$-real equations 
\begin{subequations}%
\label{MainEq1}%
\begin{eqnarray}
T^{\left( X\right) } &=&\frac{1}{2}\mu e^{-4\Delta }\delta +\delta G^{\left(
4\right) }\left( W\right) +12G^{\left( 3\right) }\left( W\right) +\left(
d\ln \left\vert X\right\vert \right) \hat{1}  \label{MainEq1a} \\
WT^{\left( X\right) } &=&\frac{1}{2}\mu e^{-4\Delta }\delta W+\nabla
W-\delta G^{\left( 4\right) }\left( 1\right) -12G^{\left( 3\right) }\left(
1\right) +\left( d\ln \left\vert X\right\vert \right) W.  \label{MainEq1b}
\end{eqnarray}%
\end{subequations}%
As shown Lemma \ref{LemG4G3FA} in the Appendix, for an octonion $A=a+\hat{%
\alpha},$ we have 
\begin{equation}
\delta \left( G^{\left( 4\right) }\left( A\right) \right) +12G^{\left(
3\right) }\left( A\right) =F\left( A\right) -\left( \hat{Q}_{2}\bar{A}%
\right) ^{t}  \label{QA1}
\end{equation}%
where 
\begin{equation}
F\left( A\right) =\left( 
\begin{array}{c}
\frac{4}{3}a_{0}Q_{1}-\frac{11}{6}\alpha Q_{0}+\frac{2}{3}Q_{2}\left( \alpha
\right) -Q_{1}\times \alpha  \\ 
-\left( \frac{1}{3}\left\langle Q_{1},\alpha \right\rangle +\frac{5}{6}%
a_{0}Q_{0}\right) \delta +\left( \frac{1}{6}Q_{0}\alpha +\frac{1}{3}%
a_{0}Q_{1}-\frac{1}{3}Q_{2}\left( \alpha \right) \right) \lrcorner \varphi
_{X}-\alpha Q_{1}%
\end{array}%
\right)   \label{FAexp}
\end{equation}%
and in particular, 
\begin{equation}
F\left( 1\right) =\left( 
\begin{array}{c}
\frac{4}{3}Q_{1} \\ 
-\frac{5}{6}Q_{0}\delta +\frac{1}{3}Q_{1}\lrcorner \varphi _{X}%
\end{array}%
\right) .  \label{F1exp}
\end{equation}%
Note that in (\ref{FAexp}), and in similar expressions hereafter, we are
implicitly using tensor products - so, $\alpha Q_{1}$ means $\alpha \otimes
Q_{1}\in \Omega ^{1}\left( M\right) \otimes \Gamma \left( \func{Im}\mathbb{O}%
M\right) $ and the order therefore matters. We also suppress indices for
brevity.

We can then use the relationships (\ref{Q0Q1rel2}) that we obtained from the
external equation, to simplify the expression for $F\left( W\right) $. In
particular, we also obtain 
\begin{eqnarray}
Q_{1}\times w &=&\frac{3}{4}w_{0}\partial \Delta \times w+\frac{3}{4}\left(
\partial \Delta \times w\right) \times w  \notag \\
&=&\frac{3}{4}w_{0}\partial \Delta \times w+\frac{3}{4}\left\langle \partial
\Delta ,w\right\rangle w-\frac{3}{4}\left\vert w\right\vert ^{2}\partial
\Delta  \notag \\
&=&w_{0}Q_{1}-\frac{1}{4}\left( 7Q_{0}-6\lambda _{1}e^{-\Delta }\right) w+%
\frac{3}{4}\left\vert W\right\vert ^{2}\partial \Delta  \label{Q1crossw}
\end{eqnarray}%
and hence, 
\begin{eqnarray}
F\left( W\right) &=&-\frac{1}{4}\left( w_{0}Q_{0}e^{\Delta }+2\left\vert
W\right\vert ^{2}\lambda _{2}-2\left\vert W\right\vert ^{2}\mu e^{-3\Delta
}+2w_{0}\lambda _{1}\right) e^{-\Delta }\delta  \notag \\
&&+\left( 
\begin{array}{c}
-\frac{1}{4}\left( 3\left\vert W\right\vert ^{2}+2\right) d\Delta
+w_{0}Q_{1}-\frac{1}{4}\left( Q_{0}+2\lambda _{1}e^{-\Delta }\right) w \\ 
\frac{1}{4}\left( \partial \Delta +Q_{0}w-2\lambda _{1}e^{-\Delta }w\right)
\lrcorner \varphi _{X}-wQ_{1}%
\end{array}%
\right)  \notag \\
&=&\left( \left( 
\begin{array}{c}
Q_{1} \\ 
-\frac{1}{4}\left( Q_{0}-2\lambda _{1}e^{-\Delta }\right) \delta%
\end{array}%
\right) \bar{W}\right) ^{t}+\left( 
\begin{array}{c}
-\frac{1}{4}\left( 3\left\vert W\right\vert ^{2}+2\right) d\Delta -\lambda
_{1}e^{-\Delta }w \\ 
\frac{1}{4}\partial \Delta \lrcorner \varphi _{X}-\frac{1}{2}\left( \lambda
_{2}+\mu e^{-3\Delta }\right) e^{-\Delta }\delta%
\end{array}%
\right)  \label{FWexp}
\end{eqnarray}%
where we have also used the expression for $Q_{2}\left( w\right) $ from (\ref%
{Q0Q1rel2}) the relationship (\ref{l1l2mu}) between $\lambda _{1},\lambda
_{2}$ and $\mu e^{-4\Delta }$. It is important to note that $F\left(
1\right) $ and $F\left( W\right) $ do not contain any dependence on $Q_{2}$.
Hence, overall, also using (\ref{FWexp}) and (\ref{dlnX}), the expression (%
\ref{MainEq1a}) for the torsion $T^{\left( X\right) }$ can be rewritten
succinctly as 
\begin{subequations}
\begin{eqnarray}
T^{\left( X\right) } &=&\frac{1}{2}\mu e^{-4\Delta }\delta +F\left( W\right)
-\left( \hat{Q}_{2}\bar{W}\right) ^{t}+\left( d\ln \left\vert X\right\vert
\right) \hat{1}  \notag \\
&=&-\left( \hat{Q}\bar{W}\right) ^{t}+\left( 
\begin{array}{c}
-\frac{3}{4}\partial \Delta \\ 
\frac{1}{4}\partial \Delta \lrcorner \varphi _{X}%
\end{array}%
\right) +\frac{1}{2}e^{-\Delta }\left( 
\begin{array}{c}
\lambda _{1}w \\ 
-\lambda _{2}\delta%
\end{array}%
\right)  \notag \\
&=&-\left( \hat{Q}\bar{W}\right) ^{t}-\frac{1}{4}\delta \left( \widehat{%
\partial \Delta }\right) +\left( 
\begin{array}{c}
\frac{1}{2}\lambda _{1}e^{-\Delta }w-\partial \Delta \\ 
-\frac{1}{2}\lambda _{2}e^{-\Delta }\delta%
\end{array}%
\right)  \label{TXexp}
\end{eqnarray}%
\end{subequations}%
where we set 
\begin{equation}
\hat{Q}=\left( 
\begin{array}{c}
-Q_{1} \\ 
\frac{1}{4}\left( Q_{0}-2\lambda _{1}e^{-\Delta }\right) \delta +Q_{2}%
\end{array}%
\right) .  \label{Qhat1}
\end{equation}%
From (\ref{TXexp}), we get 
\begin{equation}
\left( T^{\left( X\right) }\right) ^{t}=-\hat{Q}\bar{W}+\frac{1}{4}\delta
\left( \widehat{\partial \Delta }\right) +\left( 
\begin{array}{c}
\frac{1}{2}\lambda _{1}e^{-\Delta }w-\frac{1}{2}\partial \Delta \\ 
-\frac{1}{2}\lambda _{2}e^{-\Delta }\delta%
\end{array}%
\right) .  \label{TXtr}
\end{equation}

The equation (\ref{MainEq1b}) can be rewritten as 
\begin{subequations}
\begin{eqnarray}
WT^{\left( X\right) } &=&\frac{1}{2}\mu e^{-4\Delta }\delta W+\nabla
W-F\left( 1\right) +\hat{Q}_{2}+\left( d\ln \left\vert X\right\vert \right) W
\label{MainEq2b} \\
&=&\nabla W+\left( \frac{1}{2}\mu e^{-4\Delta }\delta +d\ln \left\vert
X\right\vert \right) W+\left( 
\begin{array}{c}
-\frac{4}{3}Q_{1} \\ 
\frac{5}{6}Q_{0}\delta -\frac{1}{3}Q_{1}\lrcorner \varphi _{X}+Q_{2}%
\end{array}%
\right) .  \label{WTXexp}
\end{eqnarray}%
\end{subequations}%

We can also write out (\ref{TXexp}) explicitly as 
\begin{eqnarray}
T^{\left( X\right) } &=&-\frac{1}{4}\left( w_{0}Q_{0}-2w_{0}\lambda
_{1}e^{-\Delta }+2\lambda _{2}e^{-\Delta }\right) \delta +\frac{1}{4}\left(
\partial \Delta +Q_{0}w-2\lambda _{1}e^{-\Delta }w\right) \lrcorner \varphi
_{X}  \label{TorsionMain} \\
&&-wQ_{1}-w_{0}Q_{2}+\left( Q_{2}\times w\right) ^{t}.  \notag
\end{eqnarray}%
This expression can be simplified if $\mu =\lambda _{2}=w_{0}=0$: 
\begin{equation}
T^{\left( X\right) }=\frac{1}{4}\left( \partial \Delta +Q_{0}w-2\lambda
_{1}e^{-\Delta }w\right) \lrcorner \varphi _{X}-wQ_{1}+\left( Q_{2}\times
w\right) ^{t}.
\end{equation}

From the expression for the torsion (\ref{TXexp}) we can immediately work
out the $1$-dimensional and $7$-dimensional components of $T^{\left(
X\right) }.$ From (\ref{Dir1}), we find that 
\begin{equation*}
\delta ^{a}\left( T^{\left( X\right) }\right) _{a}^{t}=\left( 
\begin{array}{c}
-7\tau _{1} \\ 
-6\tau _{7}%
\end{array}%
\right) 
\end{equation*}%
and then, using the property $\delta ^{a}\delta _{a}=-7,$ we get 
\begin{eqnarray*}
\delta ^{a}\left( T_{a}^{\left( X\right) }\right) ^{t} &=&\delta ^{a}\left( 
\hat{Q}_{a}\bar{W}\right) -\frac{9}{4}\widehat{\partial \Delta }+\lambda
_{1}e^{-\Delta }\hat{w}+\frac{7}{2}e^{-\Delta }\lambda _{2} \\
&=&\left( \delta ^{a}\hat{Q}_{a}\right) \bar{W}+\left[ \delta ^{a},\hat{Q}%
_{a},\bar{W}\right] -\frac{9}{4}\widehat{\partial \Delta }+\lambda
_{1}e^{-\Delta }\hat{w}+\frac{7}{2}\lambda _{2}e^{-\Delta }.
\end{eqnarray*}%
In components, $\left[ \delta ^{a},\hat{Q}_{a},\bar{W}\right] ^{c}=-2\psi
_{\ mnp}^{c}Q^{nm}w^{p}=0,$ since $Q^{nm}$ is symmetric. Moreover, using (%
\ref{exteqs3a}), 
\begin{eqnarray*}
\delta ^{a}\hat{Q}_{a} &=&\left( 
\begin{array}{c}
\frac{7}{4}Q_{0} \\ 
Q_{1}%
\end{array}%
\right) -\frac{7}{2}\lambda _{1}e^{-\Delta } \\
&=&\frac{3}{4}\left( 
\begin{array}{c}
2e^{-\Delta }\left( \lambda _{2}-\mu e^{-3\Delta }\right)  \\ 
-\partial \Delta 
\end{array}%
\right) W-2\lambda _{1}e^{-\Delta }.
\end{eqnarray*}%
Overall, 
\begin{equation*}
\delta ^{a}\left( T_{a}^{\left( X\right) }\right) ^{t}=\left( 
\begin{array}{c}
\frac{3}{2}\left( \lambda _{2}-\mu e^{-3\Delta }\right) e^{-\Delta
}\left\vert W\right\vert ^{2}-2\lambda _{1}e^{-\Delta }w_{0}+\frac{7}{2}%
e^{-\Delta }\lambda _{2} \\ 
-\frac{3}{4}\left( \left\vert W\right\vert ^{2}+3\right) \partial \Delta
+3\lambda _{1}e^{-\Delta }w%
\end{array}%
\right) 
\end{equation*}%
Using (\ref{l1l2mu}) to simplify the $\mathbb{O}$-real part, we obtain 
\begin{subequations}%
\label{t17X} 
\begin{eqnarray}
\tau _{1}^{\left( X\right) } &=&-\frac{1}{14}\left( 1+\left\vert
W\right\vert ^{2}\right) e^{-\Delta }\left( 5\lambda _{2}-\frac{5\left\vert
W\right\vert ^{2}+2}{\left\vert W\right\vert ^{2}+1}\mu e^{-3\Delta }\right) 
\label{TrT} \\
\tau _{7}^{\left( X\right) } &=&\frac{1}{8}\left( \left\vert W\right\vert
^{2}+3\right) d\Delta -\frac{1}{2}\lambda _{1}e^{-\Delta }w.
\label{tau7main}
\end{eqnarray}%
\end{subequations}%
If $\mu =\lambda _{2}=w_{0}=0$, we see that $\tau _{1}^{\left( X\right) }=0,$
and thus the torsion is in the class $\mathbf{7\oplus 14\oplus 27}.$ We can
use (\ref{dWsq}) to rewrite $\tau _{7}^{\left( X\right) }$ as 
\begin{equation}
\tau _{7}^{\left( X\right) }=\frac{1}{16}\left( 5\left\vert W\right\vert
^{2}+3\right) \partial \Delta +\frac{1}{8}d\left( \ln \left( 1+\left\vert
W\right\vert ^{2}\right) \right) 
\end{equation}%
Thus $\tau _{7}^{\left( X\right) }$ is closed if and only if $d\left\vert
W\right\vert ^{2}\wedge d\Delta =0$, or from (\ref{dWsq}) equivalently if 
\begin{equation*}
\lambda _{1}w\wedge d\Delta =0
\end{equation*}%
that is, if either $\lambda _{1}=0\ $or $w$ and $\partial \Delta $ are
linearly dependent. In the case when $\lambda _{1}=0,$ the above simplify to 
\begin{subequations}
\begin{eqnarray}
\tau _{1}^{\left( X\right) } &=&\frac{2}{7}\frac{\mu e^{-\Delta }}{k_{1}}%
\left( 10-3k_{1}e^{-3\Delta }\right)  \\
\tau _{7}^{\left( X\right) } &=&\frac{1}{8}\left( \left\vert W\right\vert
^{2}+3\right) d\Delta 
\end{eqnarray}%
\end{subequations}%
However recall that we have the expression (\ref{normWsq}) for $\left\vert
W\right\vert ^{2}$ in terms of $\Delta .$ Integrating that, we find that 
\begin{equation}
\tau _{7}^{\left( X\right) }=\frac{1}{12}d\ln \left( \left\vert
1+k_{1}e^{-3\Delta }\right\vert e^{6\Delta }\right)   \label{tau7main2}
\end{equation}%
Thus, $\tau _{7}^{\left( X\right) }$ is again exact. In the case of a
Minkowski background, when $\lambda _{2}$ and $\mu $ also vanish, $\tau
_{1}^{\left( X\right) }$ is then $0$. Hence in that case, $\varphi _{X}$ is
conformally equivalent to a $G_{2}$-structure in the class $\mathbf{14\oplus
27}.$

\subsection{Torsion of $\protect\varphi _{Y}$}

The equation (\ref{WTXexp}) can be reformulated to give $T^{\left( Y\right)
} $ - since this is torsion of the $G_{2}$-structure $\varphi _{Y}=\sigma
_{W}\left( \varphi _{X}\right) $. Indeed, we can rewrite it as 
\begin{equation}
DW=\left( 
\begin{array}{c}
\frac{4}{3}Q_{1} \\ 
-\frac{5}{6}Q_{0}\delta +\frac{1}{3}Q_{1}\lrcorner \varphi _{X}-Q_{2}%
\end{array}%
\right) -\left( \frac{1}{2}\mu e^{-4\Delta }\delta +d\ln \left\vert
X\right\vert \right) W  \label{DWexp}
\end{equation}%
and thus, 
\begin{eqnarray}
T^{\left( Y\right) } &=&-\left( DW\right) W^{-1}+d\ln \left\vert W\right\vert
\notag \\
&=&-\left( 
\begin{array}{c}
\frac{4}{3}Q_{1} \\ 
-\frac{5}{6}Q_{0}\delta +\frac{1}{3}Q_{1}\lrcorner \varphi _{X}-Q_{2}%
\end{array}%
\right) W^{-1}+\left( \frac{1}{2}\mu e^{-4\Delta }\delta +\left( d\ln
\left\vert Y\right\vert \right) \hat{1}\right)  \label{TYexp1}
\end{eqnarray}%
since $\left\vert W\right\vert =\left\vert Y\right\vert \left\vert
X\right\vert ^{-1}$. Now note that 
\begin{equation*}
\left( 
\begin{array}{c}
\frac{4}{3}Q_{1} \\ 
-\frac{5}{6}Q_{0}\delta +\frac{1}{3}Q_{1}\lrcorner \varphi _{X}-Q_{2}%
\end{array}%
\right) =-\hat{Q}+\left( 
\begin{array}{c}
\frac{1}{3}Q_{1} \\ 
-\frac{7}{12}Q_{0}\delta +\frac{1}{3}Q_{1}\lrcorner \varphi _{X}%
\end{array}%
\right) -\frac{1}{2}\lambda _{1}e^{-\Delta }\delta .
\end{equation*}%
But,%
\begin{eqnarray*}
\left( 
\begin{array}{c}
\frac{1}{3}Q_{1} \\ 
-\frac{7}{12}Q_{0}\delta +\frac{1}{3}Q_{1}\lrcorner \varphi _{X}%
\end{array}%
\right) &=&-\delta \left( 
\begin{array}{c}
\frac{7}{12}Q_{0} \\ 
\frac{1}{3}Q_{1}%
\end{array}%
\right) \\
&=&-\frac{1}{4}\delta \left( \left( 
\begin{array}{c}
2e^{-\Delta }\left( \lambda _{2}-\mu e^{-3\Delta }\right) \\ 
-\partial \Delta%
\end{array}%
\right) W+2e^{-\Delta }\lambda _{1}\right)
\end{eqnarray*}%
where we have also used (\ref{exteqs3}). Hence, 
\begin{equation*}
\left( 
\begin{array}{c}
\frac{4}{3}Q_{1} \\ 
-\frac{5}{6}Q_{0}\delta +\frac{1}{3}Q_{1}\lrcorner \varphi _{X}-Q_{2}%
\end{array}%
\right) =-\hat{Q}-\lambda _{1}e^{-\Delta }\delta -\frac{1}{4}\delta \left(
\left( 
\begin{array}{c}
2e^{-\Delta }\left( \lambda _{2}-\mu e^{-3\Delta }\right) \\ 
-\partial \Delta%
\end{array}%
\right) W\right) .
\end{equation*}%
Therefore, 
\begin{equation}
T^{\left( Y\right) }=\hat{Q}W^{-1}-\frac{1}{4}\left( \delta \left( \left( 
\widehat{\partial \Delta }\right) W\right) \right) W^{-1}+\lambda
_{1}e^{-\Delta }\left( \delta W^{-1}\right) +\frac{1}{2}\lambda
_{2}e^{-\Delta }\delta +\left( d\ln \left\vert Y\right\vert \right) \hat{1}.
\label{TYexp0}
\end{equation}%
Using (\ref{associd1}), note that 
\begin{equation*}
\left( \delta \left( \left( \widehat{\partial \Delta }\right) W\right)
\right) W^{-1}=\delta \circ _{Y}\left( \widehat{\partial \Delta }\right) .
\end{equation*}%
Thus, we conclude that 
\begin{equation}
T^{\left( Y\right) }=\hat{Q}W^{-1}-\frac{1}{4}\delta \circ _{Y}\left( 
\widehat{\partial \Delta }\right) +\lambda _{1}e^{-\Delta }\left( \delta
W^{-1}\right) +\frac{1}{2}\lambda _{2}e^{-\Delta }\delta +\left( d\ln
\left\vert Y\right\vert \right) \hat{1}.  \label{TYexp}
\end{equation}%
Using the same procedure as for $T^{\left( X\right) }$, we can find the
components $\tau _{1}^{\left( Y\right) }\ $and $\tau _{7}^{\left( Y\right) }$%
:%
\begin{subequations}%
\label{t17Y} 
\begin{eqnarray}
\tau _{1}^{\left( Y\right) } &=&\frac{1}{7}e^{-\Delta }\left( 5w_{0}\lambda
_{1}\left\vert W\right\vert ^{-2}+5\lambda _{2}-\frac{3}{2}\mu e^{-3\Delta
}\right)  \label{tau1Y} \\
\tau _{7}^{\left( Y\right) } &=&\frac{1}{8}\left( \left\vert W\right\vert
^{-2}+3\right) \partial \Delta +\frac{1}{2}\lambda _{1}e^{-\Delta
}w\left\vert W\right\vert ^{-2}.  \label{tau7Y}
\end{eqnarray}%
\end{subequations}%

Comparing (\ref{TYexp}) with $T^{\left( X\right) },$ we find 
\begin{equation}
\left( T^{\left( X\right) }\right) ^{t}+\left\vert W\right\vert
^{2}T^{\left( Y\right) }=\frac{1}{4}\left( \left\vert W\right\vert
^{2}\partial \Delta \lrcorner \varphi _{Y}-\partial \Delta \lrcorner \varphi
_{X}\right) +\lambda _{1}e^{-\Delta }w\lrcorner \varphi _{X}+\frac{1}{2}%
\left( \left\vert W\right\vert ^{2}+1\right) \mu e^{-4\Delta }\delta
\label{TXTYrel}
\end{equation}%
where we have also used (\ref{l1l2mu}). Since the right-hand side of this
expression has a vanishing traceless symmetric part, it follows that 
\begin{equation}
\tau _{27}^{\left( X\right) }+\left\vert W\right\vert ^{2}\tau _{27}^{\left(
Y\right) }=0.
\end{equation}%
On the other hand, recall that 
\begin{subequations}
\begin{eqnarray}
\pi _{7}T^{\left( X\right) } &=&\frac{1}{8}\left( \left\vert W\right\vert
^{2}+3\right) \partial \Delta \lrcorner \varphi _{X}-\frac{1}{2}\lambda
_{1}e^{-\Delta }w\lrcorner \varphi _{X} \\
\left\vert W\right\vert ^{2}\pi _{7}T^{\left( Y\right) } &=&\frac{1}{8}%
\left( 1+3\left\vert W\right\vert ^{2}\right) \partial \Delta \lrcorner
\varphi _{Y}+\frac{1}{2}\lambda _{1}e^{-\Delta }w\lrcorner \varphi _{Y}.
\end{eqnarray}%
\end{subequations}%
Thus, taking the skew-symmetric part of (\ref{TXTYrel}), subtracting the
appropriate $7$-dimensional components, and noting that $w\lrcorner \varphi
_{X}=w\lrcorner \varphi _{Y}$, we get 
\begin{equation*}
-\pi _{14}T^{\left( X\right) }+\left\vert W\right\vert ^{2}\pi
_{14}T^{\left( Y\right) }=\frac{1}{8}\left( \left\vert W\right\vert
^{2}+1\right) \left( \partial \Delta \lrcorner \varphi _{X}-\partial \Delta
\lrcorner \varphi _{Y}\right) .
\end{equation*}%
Going back to the octonion description again, we see that 
\begin{eqnarray}
\partial \Delta \lrcorner \varphi _{X}-\partial \Delta \lrcorner \varphi
_{Y} &=&\delta \circ _{Y}\left( \widehat{\partial \Delta }\right) -\delta
\circ _{X}\left( \widehat{\partial \Delta }\right)  \notag \\
&=&\left[ \delta ,\widehat{\partial \Delta },W\right] W^{-1}  \label{ddphixy}
\end{eqnarray}%
where we have used (\ref{OctoVAB}), (\ref{exteqs3}), as well as properties
of the associator. Hence, overall, 
\begin{equation}
\left\vert W\right\vert ^{2}\pi _{14}T^{\left( Y\right) }-\pi _{14}T^{\left(
X\right) }=\frac{1}{8}\left( \left\vert W\right\vert ^{2}+1\right) \left[
\delta ,\widehat{\partial \Delta },W\right] W^{-1}.  \label{tau14XYrel}
\end{equation}%
Using (\ref{ddphixy}) we can then rewrite the skew-symmetric part of (\ref%
{TXTYrel}) as 
\begin{eqnarray}
\func{Skew}\left( \left( T^{\left( X\right) }\right) ^{t}+\left\vert
W\right\vert ^{2}T^{\left( Y\right) }\right) &=&\frac{1}{4}\left\vert
W\right\vert ^{2}\left( \partial \Delta \lrcorner \varphi _{Y}-\partial
\Delta \lrcorner \varphi _{X}\right) +\frac{1}{4}\left( \left\vert
W\right\vert ^{2}-1\right) \partial \Delta \lrcorner \varphi _{X}+\lambda
_{1}e^{-\Delta }w\lrcorner \varphi _{X}  \notag \\
&=&-\frac{1}{4}\left[ \delta ,\widehat{\partial \Delta },W\right] \bar{W}%
+\left( \frac{1}{4}\left( \left\vert W\right\vert ^{2}-1\right) \partial
\Delta +\lambda _{1}e^{-\Delta }w\right) \lrcorner \varphi _{X}.
\end{eqnarray}%
We can summarize our findings in a theorem.

\begin{theorem}
Let $M$ be a $7$-dimensional manifold that admits $G_{2}$-structures. Then,
for a given metric $g,$ an arbitrary $G_{2}$-structure $\varphi $ that is
compatible with $g$, a $4$-form $G$, a real function $\Delta ,$and real
constants $\lambda _{1},\lambda _{2},\mu $, there exists a solution $Z=X+iY$
to the equations (\ref{susyeqoct}) if and only if all the following
conditions are satisfied:

\begin{enumerate}
\item Given $W=\left( w_{0},w\right) =YX^{-1}$, the quantities $w_{0}$ and $%
\left\vert W\right\vert ^{2}$ satisfy the following equations 
\begin{equation*}
\left\{ 
\begin{array}{c}
2w_{0}\lambda _{1}-\lambda _{2}\left( 1-\left\vert W\right\vert ^{2}\right)
=\mu e^{-3\Delta }\left( 1+\left\vert W\right\vert ^{2}\right) \\ 
dw_{0}=-\frac{3}{2}\left( 1+\left\vert W\right\vert ^{2}\right) w_{0}d\Delta
-4\left( w_{0}\lambda _{1}-\lambda _{2}\right) e^{-\Delta }w \\ 
\frac{d\left\vert W\right\vert ^{2}}{1+\left\vert W\right\vert ^{2}}=\frac{3%
}{2}\left( 1-\left\vert W\right\vert ^{2}\right) d\Delta -4\lambda
_{1}e^{-\Delta }w.%
\end{array}%
\right.
\end{equation*}

\item The components $Q_{0},Q_{1},Q_{2}$ (\ref{Gdecomp}) of $G$ with respect
to $\varphi _{X}$ satisfy%
\begin{equation*}
\left\{ 
\begin{array}{c}
Q_{2}\left( w\right) =w_{0}Q_{1}-\frac{1}{4}\left( Q_{0}-6\lambda
_{1}e^{-\Delta }\right) w-\frac{3}{4}\partial \Delta \\ 
\left( 
\begin{array}{c}
\frac{7}{3}Q_{0} \\ 
\frac{4}{3}Q_{1}%
\end{array}%
\right) =\left( 
\begin{array}{c}
2e^{-\Delta }\left( \lambda _{2}-\mu e^{-3\Delta }\right) \\ 
-\partial \Delta%
\end{array}%
\right) W+2\lambda _{1}e^{-\Delta }.%
\end{array}%
\right.
\end{equation*}

\item The torsion $T^{\left( X\right) }$ of $\varphi _{X}$ is given by%
\begin{equation*}
\left( T^{\left( X\right) }\right) ^{t}=\hat{Q}\bar{W}+\frac{1}{4}\delta
\left( \widehat{\partial \Delta }\right) +\left( 
\begin{array}{c}
\lambda _{1}e^{-\Delta }w-\frac{1}{2}\partial \Delta \\ 
-\frac{1}{2}\lambda _{2}e^{-\Delta }\delta%
\end{array}%
\right) \ \text{where }\hat{Q}=\left( 
\begin{array}{c}
Q_{1} \\ 
-\frac{1}{4}\left( Q_{0}-2\lambda _{1}e^{-\Delta }\right) \delta -Q_{2}%
\end{array}%
\right) .
\end{equation*}

\item The components of the torsion tensor $T^{\left( Y\right) }$ of $%
\varphi _{Y}$ are related to the components of $T^{\left( X\right) }$ via
the following relations 
\begin{equation}
\left\{ 
\begin{array}{c}
\tau _{1}^{\left( X\right) }+\left\vert W\right\vert ^{2}\tau _{1}^{\left(
Y\right) }=\frac{1}{2}\mu e^{-4\Delta }\left( 1+\left\vert W\right\vert
^{2}\right) \\ 
\tau _{7}^{\left( X\right) }+\left\vert W\right\vert ^{2}\tau _{7}^{\left(
Y\right) }=\frac{1}{2}\left( 1+\left\vert W\right\vert ^{2}\right) \partial
\Delta \\ 
\tau _{14}^{\left( X\right) }-\left\vert W\right\vert ^{2}\tau _{14}^{\left(
Y\right) }=-\frac{1}{8}\left( \left\vert W\right\vert ^{2}+1\right) \left[
\delta ,\widehat{\partial \Delta },W\right] W^{-1} \\ 
\tau _{27}^{\left( X\right) }+\left\vert W\right\vert ^{2}\tau _{27}^{\left(
Y\right) }=0%
\end{array}%
\right. .  \label{TXTYrel2}
\end{equation}
\end{enumerate}
\end{theorem}

Moreover, we can see that if $\partial \Delta $ and $w$ are linearly
independent, then both $T^{\left( X\right) }$ and $T^{\left( Y\right) }$
have all torsion components nonzero except the $1$-dimensional components.
If $\mu =0$, from (\ref{Q1mu0}), we see that the condition of $\partial
\Delta $ and $w$ being not multiples of one another is equivalent to $Q_{1}$
being nonzero.

\begin{theorem}
Let  $\left\vert \lambda \right\vert \neq 0\ $\ and suppose $\partial \Delta 
$ and $w$ are not linearly dependent everywhere. Then, if $\mu =0$, $%
T^{\left( X\right) }$ and $T^{\left( Y\right) }$ have the same torsion type $%
\mathbf{7\oplus 14\oplus 27}$ with all components non-zero. If $\mu \neq 0,$
then moreover the $\mathbf{1}$ component may also be non-zero.
\end{theorem}

\begin{proof}
Since $\partial \Delta $ and $w$ are not linearly dependent everywhere, we
may assume that at some point, $\partial \Delta $ is not proportional to $w$
and in particular, $\partial \Delta \neq 0$. If $\mu =0$, then we know $\tau
_{1}^{\left( X\right) }=\tau _{1}^{\left( Y\right) }=0$. To show that $%
T^{\left( X\right) }$ and $T^{\left( Y\right) }$ have the same torsion type
it is now sufficient to show using (\ref{TXTYrel2}) that if $\tau
_{i}^{\left( Y\right) }=0$ then $\tau _{i}^{\left( X\right) }=0$ for any $%
i=7,14,27$. The converse will follow by a similar argument. From (\ref%
{TXTYrel2}) this is immediately obvious for the $27$-dimensional component.
From (\ref{tau7main}) and (\ref{tau7Y}) it also follows immediately that
each of $\tau _{7}^{\left( X\right) }$ and $\tau _{7}^{\left( Y\right) }$
vanish only if $\partial \Delta $ is proportional to $w$. Thus, if $\partial
\Delta $ is not proportional to $w$, then both $\tau _{7}^{\left( X\right) }$
and $\tau _{7}^{\left( Y\right) }$ are non-zero.

Suppose now $\tau _{14}^{\left( Y\right) }=0$, then from (\ref{TXTYrel2}),%
\begin{equation*}
\tau _{14}^{\left( X\right) }=-\frac{1}{8}\left( \left\vert W\right\vert
^{2}+1\right) \left[ \delta ,\widehat{\partial \Delta },W\right] W^{-1}.
\end{equation*}%
If $\partial \Delta $ is not a multiple of $w$, then $\left[ \delta ,%
\widehat{\partial \Delta },W\right] \neq 0.$ However, we can rewrite this
expression for $\tau _{14}^{\left( X\right) }$ as%
\begin{equation}
\left( \tau _{14}^{\left( X\right) }\right) _{ab}=-\frac{1}{4}\left(
\left\vert W\right\vert ^{-2}+1\right) \psi _{abcd}\left( \partial \Delta
\times w\right) ^{c}w^{d}.  \label{tau14xtau14Y0}
\end{equation}%
With respect to $\varphi _{X}$ the right-hand side will always have a $%
\Lambda _{7}^{2}$ component that is proportional to $\pi _{7}\left( \left(
\partial \Delta \times w\right) \wedge w\right) .$ Since $w$ is not
identically zero, this component is non-zero if and only if $\partial \Delta 
$ is not a multiple of $w$. Hence, we get a contradiction, since the
left-hand side of (\ref{tau14xtau14Y0}) is in $\Lambda _{14}^{2}$ and the
right-hand side has a nonvanishing component in $\Lambda _{7}^{2}.$ So, $%
\tau _{14}^{\left( Y\right) }$ must be non-zero and a similar argument will
show that $\tau _{14}^{\left( X\right) }$ must also be non-zero.

To show that the $27$-dimensional component must be non-zero, from (\ref%
{TorsionMain}) we can write down $\tau _{27}^{\left( X\right) }$ in the case 
$\mu =\lambda _{2}=w_{0}=0$:%
\begin{equation}
\tau _{27}^{\left( X\right) }=\func{Sym}\left( Q_{2}\times
w-w_{0}Q_{2}-wQ_{1}\right) .  \label{t27X}
\end{equation}%
From (\ref{t27X}), we can work out $w\lrcorner \tau _{27}^{\left( X\right) },
$ which we simplify using (\ref{Q0Q1rel2}). This will give us a linear
combination of $w,\partial \Delta ,w\times \partial \Delta $ which will
vanish only if $w$ and $\partial \Delta $ are proportional. Hence, under our
hypothesis, hence $\tau _{27}^{\left( X\right) }$ is non-zero.
\end{proof}

\subsection{Derivative of $W$}

Instead of using (\ref{dwexp}) to calculate the torsion of $T^{\left(
Y\right) }$, we may instead use it to give a differential equation for $W$:%
\begin{eqnarray}
\nabla W &=&WT^{\left( X\right) }+\left( 
\begin{array}{c}
\frac{4}{3}Q_{1} \\ 
-\frac{5}{6}Q_{0}\delta +\frac{1}{3}Q_{1}\lrcorner \varphi _{X}-Q_{2}%
\end{array}%
\right) -\left( \frac{1}{2}\mu e^{-4\Delta }\delta +d\ln \left\vert
X\right\vert \right) W  \notag \\
&=&-W\left( \hat{Q}\bar{W}\right) ^{t}-\frac{1}{4}W\left( \delta \left( 
\widehat{\partial \Delta }\right) \right) +W\left( 
\begin{array}{c}
\frac{1}{2}\lambda _{1}e^{-\Delta }w-\partial \Delta  \\ 
-\frac{1}{2}\lambda _{2}e^{-\Delta }\delta 
\end{array}%
\right)   \notag \\
&&+\left( 
\begin{array}{c}
\frac{4}{3}Q_{1} \\ 
-\frac{5}{6}Q_{0}\delta +\frac{1}{3}Q_{1}\lrcorner \varphi _{X}-Q_{2}%
\end{array}%
\right) -\left( \frac{1}{2}\mu e^{-4\Delta }\delta +d\ln \left\vert
X\right\vert \right) W.  \label{dWeq2}
\end{eqnarray}%
The $\mathbb{O}$-real part of (\ref{dWeq2}) gives a differential equation
for $w_{0},$ however it can be shown to recover (\ref{dw0}). However, the $%
\mathbb{O}$-imaginary part gives a differential equation for the vector
field $w$ (and equivalently, its dual $1$-form, which we will also denote by 
$w$). Overall, 
\begin{eqnarray}
\nabla w &=&-\frac{1}{12}\left( \left( 3\left\vert W\right\vert
^{2}+17\right) Q_{0}-6\left( 1+\left\vert W\right\vert ^{2}\right) \lambda
_{1}e^{-\Delta }+12\mu w_{0}e^{-4\Delta }\right) g  \notag \\
&&+\frac{1}{4}\left( 2+3\left\vert W\right\vert ^{2}\right) w\left( \partial
\Delta \right) -\frac{3}{4}\left( \left\vert W\right\vert ^{2}+1\right)
\left( \partial \Delta \right) w  \notag \\
&&+w_{0}Q_{1}w-\frac{1}{4}\left( 7Q_{0}-2\lambda _{1}e^{-\Delta }\right)
ww-Q_{2}-\func{Im}_{\mathbb{O}}\left( W\left( \hat{Q}_{2}\bar{W}\right)
^{t}\right) .  \label{nablawexp}
\end{eqnarray}%
Using Lemma \ref{LemSkewWPWt} from the Appendix, we see that 
\begin{eqnarray*}
\func{Skew}\left( W\left( \hat{Q}_{2}\bar{W}\right) ^{t}\right)  &=&\left(
Q_{2}\left( w\right) \right) w-w\left( Q_{2}\left( w\right) \right)  \\
&=&w_{0}\left( Q_{1}w-wQ_{1}\right) +\frac{3}{4}w\left( \partial \Delta
\right) -\frac{3}{4}\left( \partial \Delta \right) w.
\end{eqnarray*}%
Hence, taking the skew-symmetric part of (\ref{nablawexp}), we find that 
\begin{equation}
dw=-\frac{1}{2}\left( 3\left\vert W\right\vert ^{2}+1\right) d\Delta \wedge
w.  \label{dwexp1}
\end{equation}%
Note that this is equivalent to (\ref{ddyx0}). If $\mu =0$, then for $%
\left\vert W\right\vert ^{2}=1,$ then we can rewrite it as 
\begin{equation}
d\left( e^{2\Delta }w\right) =0.  \label{dwexp}
\end{equation}%
This exactly recovers equation (3.28) from \cite{KasteMinasianFlux}.
However, whenever $\left\vert W\right\vert ^{2}\neq 1$, using (\ref{dWsq}),
we obtain 
\begin{equation}
d\left( \left( \left\vert W\right\vert ^{2}-1\right) ^{-\frac{2}{3}}\left(
\left\vert W\right\vert ^{2}+1\right) ^{-\frac{1}{3}}w\right) =0.
\label{dwexp2}
\end{equation}

For the symmetric part of $\nabla w,$ we obtain 
\begin{eqnarray}
\func{Sym}\nabla w &=&-\frac{1}{12}\left( \left( 3\left\vert W\right\vert
^{2}+17\right) Q_{0}-6\left( 1+\left\vert W\right\vert ^{2}\right) \lambda
_{1}e^{-\Delta }+12\mu w_{0}e^{-4\Delta }\right) g  \notag \\
&&-\frac{1}{4}\left( 7Q_{0}-2\lambda _{1}e^{-\Delta }\right) ww-Q_{2}  \notag
\\
&&+\func{Sym}\left( w_{0}Q_{1}w-\frac{1}{4}w\left( \partial \Delta \right) -%
\func{Im}_{\mathbb{O}}\left( W\left( \hat{Q}_{2}\bar{W}\right) ^{T}\right)
\right) .  \label{symdwexp}
\end{eqnarray}%
Explicitly, in coordinates, this gives 
\begin{eqnarray}
\nabla _{(a}w_{b)} &=&-\frac{1}{12}\left( \left( 3\left\vert W\right\vert
^{2}+17\right) Q_{0}-6\left( 1+\left\vert W\right\vert ^{2}\right) \lambda
_{1}e^{-\Delta }+12\mu w_{0}e^{-4\Delta }\right) g_{ab}  \notag \\
&&-\frac{1}{2}\left( 3Q_{0}+2\lambda _{1}e^{-\Delta }\right)
w_{a}w_{b}-\left( Q_{2}\right) _{ab}\left( w_{0}^{2}+1\right)  \notag \\
&&+\frac{1}{2}w_{(a}\left( \partial _{b)}\Delta \right) -2w_{0}\left(
Q_{2}\right) _{(a}^{\ m}w^{n}\varphi _{b)mn}^{\ }-\varphi _{amn}\varphi
_{bpq}\left( Q_{2}\right) ^{mp}w^{n}w^{q}.  \label{symdw}
\end{eqnarray}%
In particular, if $\mu =w_{0}=\lambda _{2}=0$, this simplifies to%
\begin{eqnarray}
\nabla _{(a}w_{b)} &=&-\frac{1}{28}\left( \left( 3\left\vert W\right\vert
^{2}+17\right) \left\langle \partial \Delta ,w\right\rangle +4\left(
5-2\left\vert W\right\vert ^{2}\right) \lambda _{1}e^{-\Delta }\right) g_{ab}
\notag \\
&&-\frac{1}{14}\left( 3\left\langle \partial \Delta ,w\right\rangle
+20\lambda _{1}e^{-\Delta }\right) w_{a}w_{b}-\left( Q_{2}\right) _{ab}+%
\frac{1}{2}w_{(a}\left( \partial _{b)}\Delta \right)  \notag \\
&&-\varphi _{amn}\varphi _{bpq}\left( Q_{2}\right) ^{mp}w^{n}w^{q}.
\end{eqnarray}%
where we have also used (\ref{Q0Q1rel2 copy(1)}). Furthermore, in the
Minkowski case, together with $\left\vert W\right\vert ^{2}=1,$ we get 
\begin{equation}
\nabla _{(a}w_{b)}=-\frac{5}{3}Q_{0}g_{ab}-\frac{3}{2}Q_{0}w_{a}w_{b}-\left(
Q_{2}\right) _{ab}+\frac{1}{2}w_{(a}\left( \partial _{b)}\Delta \right)
-\varphi _{amn}\varphi _{bpq}\left( Q_{2}\right) ^{mp}w^{n}w^{q}.
\end{equation}%
This has some additional terms compared to the corresponding expression
(3.29) in \cite{KasteMinasianFlux}, however it agrees with \cite%
{Behrndt:2003zg}.

Taking the trace in (\ref{symdw}), using (\ref{l1l2mu}), and then using the
expression (\ref{Q0rel}) for $Q_{0},$ we obtain 
\begin{equation}
\func{div}w=-\frac{3}{2}\left( \left\vert W\right\vert ^{2}+3\right)
\left\langle \partial \Delta ,w\right\rangle +\left( 4w_{0}^{2}+\left\vert
W\right\vert ^{2}-5\right) \lambda _{1}e^{-\Delta }-10w_{0}\lambda
_{2}e^{-\Delta }.  \label{divw}
\end{equation}

We can also use the equation (\ref{dWeq2}) to express $Q_{2}$ in terms of $W$
and $\partial \Delta $. First note that 
\begin{equation*}
\bar{W}\hat{Q}_{2}=-\hat{Q}_{2}\bar{W}+2Q_{2}\left( w\right) +2w_{0}\hat{Q}%
_{2}
\end{equation*}%
and hence multiplying (\ref{dWeq2}) by $\bar{W}$ on the left, and
rearranging, we obtain 
\begin{equation}
\hat{Q}_{2}\bar{W}-\left\vert W\right\vert ^{2}\left( \hat{Q}_{2}\bar{W}%
\right) ^{t}-2w_{0}\hat{Q}_{2}=\bar{W}\nabla W-\left\vert W\right\vert
^{2}F\left( W\right) -\bar{W}F\left( 1\right) +2Q_{2}\left( w\right) +\frac{1%
}{2}\mu e^{-4\Delta }\left( \bar{W}\delta W-\left\vert W\right\vert
^{2}\delta \right) .  \label{Q2Weq}
\end{equation}%
Note that due to (\ref{Q2wrel}), the right-hand side of (\ref{Q2Weq}) has no
dependence on $Q_{2}$. When $W^{2}\neq -1$, using Lemma \ref{LemOctSol} we
can solve equation (\ref{Q2Weq}) for $\hat{Q}_{2}$ to obtain%
\begin{equation}
\hat{Q}_{2}=-\frac{\left( K+\left\vert W\right\vert ^{2}K^{t}\right) \left( 
\bar{W}+\left\vert W\right\vert ^{2}W\right) }{\left( 1+\left\vert
W\right\vert ^{2}\right) \left\vert W\right\vert ^{2}\left\vert
1+W^{2}\right\vert ^{2}}  \label{Q2sol}
\end{equation}%
where $K=\bar{W}\nabla W-\left\vert W\right\vert ^{2}F\left( W\right) -\bar{W%
}F\left( 1\right) +2Q_{2}\left( w\right) +\frac{1}{2}\mu e^{-4\Delta }\left( 
\bar{W}\delta W-\left\vert W\right\vert ^{2}\delta \right) $. The condition $%
W^{2}\neq -1$ is precisely equivalent to saying that the complex octonion $%
A=1+iW$ (or equivalently, $Z=X+iY$) is not a zero divisor. We know this is
true when $\mu \neq 0$ and for nonzero $\lambda $, this is true at least
locally even if $\mu =0$, so we may assume that $W^{2}\neq -1$. The
expression (\ref{Q2sol}) gives us a solution for $\hat{Q}_{2}$ that only
involves $\nabla w,$ $d\Delta ,$ $w_{0},$ $\left\vert W\right\vert ^{2}$,
and $\Delta $. 

\section{Integrability conditions}

\setcounter{equation}{0}\label{secInteg}From \cite%
{CleytonIvanovClosed,GrigorianG2Torsion1,karigiannis-2007} we know that the
torsion of a $G_{2}$-structure satisfies the following integrability
condition:%
\begin{equation}
\frac{1}{4}\func{Riem}_{ij}^{\ \ \beta \gamma }\varphi _{\ \ \beta \gamma
}^{\alpha }=\nabla _{i}T_{j}^{\ \ \alpha }-\nabla _{j}T_{i}^{\ \ \alpha
}+2T_{i}^{\ \beta }T_{j}^{\ \gamma }\varphi _{\ \beta \gamma }^{\alpha }.
\label{dtinteg}
\end{equation}%
In \cite{GrigorianOctobundle}, the right-hand side of this expression was
interpreted as the $\func{Im}\mathbb{O}M$-valued $2$-form $d_{D}T,$ where $%
d_{D}$ is the octonionic covariant exterior derivative obtained from $D$ (%
\ref{octocov}). The left-hand side expression has been denoted in \cite%
{karigiannis-2007} as $\pi _{7}\func{Riem}$ since we are effectively taking
a projection to $\Lambda _{7}^{2}$ of one of the $2$-form pieces of $\func{%
Riem}.$ Hence, (\ref{dtinteg}) is equivalently written in terms of octonions
as%
\begin{equation}
d_{D}T=\frac{1}{4}\pi _{7}\func{Riem}.  \label{dtinteg2}
\end{equation}

The quantity $d_{D}T$ is pointwise in the $\left( \mathbf{7}\oplus \mathbf{14%
}\right) \otimes \mathbf{7}$ representation of $G_{2}$, which, as shown in 
\cite{CleytonIvanovCurv}, decomposes in terms irreducible representations as 
\begin{equation}
\left( \mathbf{7}\oplus \mathbf{14}\right) \mathbf{\otimes 7}\cong \left( 
\mathbf{1}\oplus \mathbf{7}\oplus \mathbf{14}\oplus \mathbf{27}\right)
\oplus \left( \mathbf{7\oplus 27\oplus 64}\right) .  \label{torsintegdecomp}
\end{equation}%
We can consider different projections of $d_{D}T$ to project out the
component in $\mathbf{64}.$ In particular, let 
\begin{subequations}
\begin{eqnarray}
\left( \mathcal{T}_{1}\right) _{ab} &=&\left( d_{D}T\right)
_{acd}^{{}}\varphi _{\ \ b}^{cd}  \label{bianchit1proj} \\
\left( \mathcal{T}_{2}\right) _{ab} &=&\left( d_{D}T\right)
_{cda}^{{}}\varphi _{\ \ b}^{cd}.  \label{bianchit2proj}
\end{eqnarray}%
\end{subequations}%
From (\ref{dtinteg}), these conditions then give%
\begin{subequations}%
\label{bianchiphi} 
\begin{eqnarray}
\mathcal{T}_{1} &=&\frac{1}{2}\func{Ric}  \label{torsricci} \\
\mathcal{T}_{2} &=&\frac{1}{4}\func{Ric}^{\ast }  \label{starRicci}
\end{eqnarray}%
\end{subequations}%
where $\left( \func{Ric}^{\ast }\right) _{ab}=\func{Riem}_{mnpq}^{\ }\varphi
_{a}^{\ mn}\varphi _{b}^{\ pq}$ is the $\ast $-Ricci curvature as defined in 
\cite{CleytonIvanovClosed}. From (\ref{torsricci}) we can then obtain an
expression for $\func{Ric}$ in terms of $T$ and its derivatives which has
originally been shown by Bryant in \cite{bryant-2003}. In particular, both $%
\func{Ric}$ and $\func{Ric}^{\ast }$ are symmetric and their traces give the
scalar curvature $R$. Together, they generate the $\mathbf{1\oplus 27\oplus
27}$ part in (\ref{torsintegdecomp}). The skew-symmetric parts of $\mathcal{T%
}_{1}$ and $\mathcal{T}_{2}$ therefore vanish - the two $\mathbf{7}$
components generate the $\mathbf{7\oplus 7}$ part in (\ref{torsintegdecomp}%
), and the two $\mathbf{14}$ components are actually proportional and
correspondingly give the $\mathbf{14}$ component in (\ref{torsintegdecomp}).
An equivalent way to obtain the $\mathbf{7\oplus 7\oplus 14}$ components of
the integrability condition is from $d^{2}\varphi =0$ and $d^{2}\psi =0$.
The expression for the scalar curvature can be easily written out explicitly
in terms of components of the torsion:%
\begin{equation}
\frac{1}{4}R=42\tau _{1}^{2}+30\left\vert \tau _{7}\right\vert
^{2}-\left\vert \tau _{14}\right\vert ^{2}-\left\vert \tau _{27}\right\vert
^{2}+6\func{div}\tau _{7}  \label{ScalCurv}
\end{equation}%
or equivalently,%
\begin{equation}
\frac{1}{4}R=49\tau _{1}^{2}+36\left\vert \tau _{7}\right\vert
^{2}-\left\vert T\right\vert ^{2}+6\func{div}\tau _{7}.  \label{ScalCurv1}
\end{equation}

In our case, we have two $G_{2}$-structure $\varphi _{X}$ and $\varphi _{Y}$
with torsion tensors $T^{\left( X\right) }$ and $T^{\left( Y\right) }$ given
by (\ref{TorsionMain}) and (\ref{TYexp}), respectively. Each of these
torsion tensors satisfies their corresponding integrability conditions (\ref%
{dtinteg}). However, since they correspond to the same metric, and thus the
same Ricci curvature, we get $\mathcal{T}_{1}^{\left( X\right) }=\mathcal{T}%
_{1}^{\left( Y\right) }$. It also turns out that $\mathbf{7\oplus 7\oplus 14}
$ components of (\ref{dtinteg}) for $T^{\left( X\right) }$ and $T^{\left(
Y\right) }$ are equivalent, and together give the condition $dG=0$.
Moreover, whenever $w_{0}\neq 0$ and $\left\vert W\right\vert ^{2}\neq 1$
(which holds almost everywhere if $\mu \neq 0$), we will find that $\mathcal{%
T}_{1}^{\left( X\right) }-\mathcal{T}_{1}^{\left( Y\right) }=0$ gives
precisely the equation of motion (\ref{eom7D}) and the $1$-dimensional
component is moreover equivalent to the $4$-dimensional Einstein's equation %
\ref{ric4cond}.

Let us first work out $R$ using both $T^{\left( X\right) }$ and $T^{\left(
Y\right) }.$ To work out $R,$ from (\ref{ScalCurv1}) we need to know $\left(
\tau _{1}^{\left( X\right) }\right) ^{2},$ $\left\vert \tau _{7}^{\left(
X\right) }\right\vert ^{2},$ $\left\vert T^{\left( X\right) }\right\vert
^{2},$ and $\func{div}\tau _{7}^{\left( X\right) }.$ The first two are
immediately obtained from expressions (\ref{TrT}) and (\ref{tau7main}).
Then, $\left\vert T^{\left( X\right) }\right\vert ^{2}=\left\vert \left(
T^{\left( X\right) }\right) ^{t}\right\vert ^{2}$ can be worked out from the
expression (\ref{TXtr}) for $\left( T^{\left( X\right) }\right) ^{t}$. After
some manipulations, we obtain 
\begin{eqnarray}
\left\vert T^{\left( X\right) }\right\vert ^{2} &=&\left\vert W\right\vert
^{2}\left\vert Q_{2}\right\vert ^{2}-\frac{21}{8}\left\vert W\right\vert
^{2}Q_{0}\left( Q_{0}-\frac{4}{3}\lambda _{1}e^{-\Delta }\right) +\frac{3}{16%
}\left( 3\left\vert W\right\vert ^{4}+2\left\vert W\right\vert ^{2}+1\right)
\left\vert d\Delta \right\vert ^{2}  \label{TXsq} \\
&&-\frac{3}{2}\left\vert W\right\vert ^{2}\lambda _{1}^{2}e^{-2\Delta }+%
\frac{1}{2}\left( 5\left\vert W\right\vert ^{4}+4\left\vert W\right\vert
^{2}+2\right) \lambda _{2}^{2}e^{-2\Delta }-\frac{1}{2}\left( 10\left\vert
W\right\vert ^{4}+5\left\vert W\right\vert ^{2}+1\right) \mu \lambda
_{2}e^{-5\Delta }  \notag \\
&&+\frac{1}{4}\left( 10\left\vert W\right\vert ^{4}+2\left\vert W\right\vert
^{2}+1\right) \mu ^{2}e^{-8\Delta }.  \notag
\end{eqnarray}%
However, from the expansion of the $4$-form $G$ (\ref{GtotalX}) and then
using (\ref{Q0Q1rel2}) to expand $Q_{1}$ we obtain 
\begin{eqnarray}
\left\vert G\right\vert ^{2} &=&\frac{1}{24}G_{abcd}G^{abcd}=2\left\vert
Q_{2}\right\vert ^{2}+4\left\vert Q_{1}\right\vert ^{2}+7Q_{0}^{2}
\label{Gsq} \\
&=&2\left\vert Q_{2}\right\vert ^{2}-\frac{21}{4}Q_{0}\left( Q_{0}-4\lambda
_{1}e^{-\Delta }\right) +\frac{9}{4}\left\vert W\right\vert ^{2}\left\vert
d\Delta \right\vert ^{2}-9\lambda _{1}^{2}e^{-2\Delta }  \notag \\
&&+9\left\vert W\right\vert ^{2}\left( \lambda _{2}-\mu e^{-3\Delta }\right)
^{2}e^{-2\Delta }.  \notag
\end{eqnarray}%
Hence, we can use (\ref{Gsq}) to rewrite (\ref{TXsq}) as 
\begin{eqnarray}
\left\vert T^{\left( X\right) }\right\vert ^{2} &=&\frac{1}{2}\left\vert
W\right\vert ^{2}\left\vert G\right\vert ^{2}+\frac{3}{16}\left(
1-3\left\vert W\right\vert ^{2}\right) \left( 1-\left\vert W\right\vert
^{2}\right) \left\vert d\Delta \right\vert ^{2}-2\left\vert W\right\vert
^{2}\lambda _{1}e^{-\Delta }\left\langle d\Delta ,w\right\rangle
\label{TXsq2} \\
&&-3\left\vert W\right\vert ^{2}\lambda _{1}^{2}e^{-2\Delta }+\left(
\left\vert W\right\vert ^{4}-\left\vert W\right\vert ^{2}+1\right) \lambda
_{2}^{2}e^{-2\Delta }-\frac{1}{2}\left( 1+\left\vert W\right\vert
^{2}\right) \left( 1+4\left\vert W\right\vert ^{2}\right) \mu \lambda
_{2}e^{-5\Delta }  \notag \\
&&+\frac{1}{4}\left( 4\left\vert W\right\vert ^{4}+14\left\vert W\right\vert
^{2}+1\right) \mu ^{2}e^{-8\Delta }.  \notag
\end{eqnarray}%
where we have also used (\ref{Q0rel}) to eliminate $Q_{0}$ and (\ref{l1l2mu}%
) to eliminate $w_{0}$.

Now let us work out $\func{div}\tau _{7}^{\left( X\right) }$. From (\ref%
{tau7main}), we have 
\begin{equation*}
\func{div}\tau _{7}^{\left( X\right) }=\frac{1}{8}\left( 3+\left\vert
W\right\vert ^{2}\right) \nabla ^{2}\Delta +\frac{1}{8}\left\langle
d\left\vert W\right\vert ^{2},d\Delta \right\rangle ++\frac{1}{2}\lambda
_{1}e^{-\Delta }\left\langle d\Delta ,w\right\rangle -\frac{1}{2}\lambda
_{1}e^{-\Delta }\func{div}w.
\end{equation*}%
Using the expression (\ref{divw}) for $\func{div}w$ and (\ref{dWsq}) for $%
d\left\vert W\right\vert ^{2},$ as well as (\ref{Q0rel}) and (\ref{l1l2mu})
again, we obtain 
\begin{eqnarray}
\func{div}\tau _{7}^{\left( X\right) } &=&\frac{1}{8}\left( 3+\left\vert
W\right\vert ^{2}\right) \nabla ^{2}\Delta +\frac{3}{16}\left( 1-\left\vert
W\right\vert ^{4}\right) \left\vert d\Delta \right\vert ^{2}+\frac{1}{4}%
\left( 9+\left\vert W\right\vert ^{2}\right) \lambda _{1}e^{-\Delta
}\left\langle d\Delta ,w\right\rangle  \label{divt7X} \\
&&+\frac{1}{2}\left( 5-\left\vert W\right\vert ^{2}\right) \lambda
_{1}^{2}e^{-2\Delta }+\frac{1}{2}\left( 4+\left\vert W\right\vert
^{2}\right) \left( 1-\left\vert W\right\vert ^{2}\right) \lambda
_{2}^{2}e^{-2\Delta }  \notag \\
&&+\frac{1}{2}\left( 3+2\left\vert W\right\vert ^{2}\right) \left(
1+\left\vert W\right\vert ^{2}\right) \mu \lambda _{2}e^{-5\Delta }-\frac{1}{%
2}\left( 1+\left\vert W\right\vert ^{2}\right) ^{2}\mu ^{2}e^{-8\Delta }. 
\notag
\end{eqnarray}%
Overall, combining everything, we conclude that 
\begin{eqnarray}
R &=&-2\left\vert W\right\vert ^{2}\left\vert G\right\vert ^{2}+3\left(
3+\left\vert W\right\vert ^{2}\right) \nabla ^{2}\Delta +12\left(
2+\left\vert W\right\vert ^{2}\right) \left\vert d\Delta \right\vert
^{2}+12\left( 5+3\left\vert W\right\vert ^{2}\right) \left( \lambda
_{1}^{2}+\lambda _{2}^{2}\right) e^{-2\Delta }  \label{scalX} \\
&&-18\left( 1+2\left\vert W\right\vert ^{2}\right) \mu ^{2}e^{-8\Delta }. 
\notag
\end{eqnarray}%
Now recall that since the original choice of the $G_{2}$-structure $\varphi
_{X}$ or $\varphi _{Y}$ was arbitrary, our expressions are invariant under
the transformation $\left\{ X\longrightarrow Y,W\longrightarrow
-W^{-1},\lambda \longrightarrow -\lambda \right\} $. In (\ref{scalX}), this
means that it must be invariant under the transformation $\left\vert
W\right\vert ^{2}\longrightarrow \left\vert W\right\vert ^{-2}$, since $%
\left\vert G\right\vert ^{2}$ and other quantities are independent of the
choice of the $G_{2}$-structure. Hence the corresponding corresponding
expression for $R$ that that we would have obtained from $T^{\left( Y\right)
}$ is 
\begin{eqnarray}
R &=&-2\left\vert W\right\vert ^{-2}\left\vert G\right\vert ^{2}+3\left(
3+\left\vert W\right\vert ^{-2}\right) \nabla ^{2}\Delta +12\left(
2+\left\vert W\right\vert ^{-2}\right) \left\vert d\Delta \right\vert ^{2}
\label{scalY} \\
&&+12\left( 5+3\left\vert W\right\vert ^{-2}\right) \left( \lambda
_{1}^{2}+\lambda _{2}^{2}\right) e^{-2\Delta }-18\left( 1+2\left\vert
W\right\vert ^{-2}\right) \mu ^{2}e^{-8\Delta }.  \notag
\end{eqnarray}%
Subtracting (\ref{scalX}) from (\ref{scalY}), we obtain 
\begin{equation}
\left( \left\vert W\right\vert ^{-2}-\left\vert W\right\vert ^{2}\right)
\left( -2\left\vert G\right\vert ^{2}+3\nabla ^{2}\Delta +12\left\vert
d\Delta \right\vert ^{2}+36\left( \lambda _{1}^{2}+\lambda _{2}^{2}\right)
e^{-2\Delta }-36\mu ^{2}e^{-8\Delta }\right) =0.  \label{scaldiff}
\end{equation}%
Therefore, we see that if $\left\vert W\right\vert ^{2}\neq 1$, equality of (%
\ref{scalX}) and (\ref{scalY}) is precisely equivalent to equation (\ref%
{ric4cond}), which is the $4$-dimensional Einstein's equation in
supergravity. As we have argued previously, when $\mu \neq 0$, $\left\vert
W\right\vert ^{2}$ cannot be equal to $1$ in a neighborhood, therefore, if (%
\ref{ric4cond}) holds where $\left\vert W\right\vert ^{2}\neq 1$, by
continuity it will have to hold where $\left\vert W\right\vert ^{2}=1$. If
on the other hand, $\mu =0$, we already assume that $w_{0}=0$, so $%
\left\vert W\right\vert ^{2}=1$ implies that the complex octonion $Z=X+iY$
is a zero divisor. However, if $\lambda \neq 0$, again we know that $Z$
cannot be a zero divisor everywhere in a neighborhood, so again, we can
conclude by continuity that (\ref{ric4cond}) holds everywhere.

Now using (\ref{ric4cond}) to eliminate $\left\vert \lambda \right\vert ^{2}$
from (\ref{scalX}), we precisely obtain (\ref{scal7}), which is the scalar
curvature expression obtained from the $7$-dimensional Einstein's equation.
Therefore, the two integrability conditions that lie in the $\mathbf{1}$
representation give us%
\begin{subequations}%
\label{1conds}%
\begin{eqnarray}
\frac{1}{4}R &=&\nabla ^{2}\Delta +\left\vert d\Delta \right\vert ^{2}+\frac{%
5}{6}\left\vert G\right\vert ^{2}+\frac{21}{2}\mu ^{2}e^{-8\Delta }
\label{1cond1} \\
\nabla ^{2}\Delta &=&12\mu ^{2}e^{-8\Delta }+\frac{2}{3}\left\vert
G\right\vert ^{2}-4\left\vert d\Delta \right\vert ^{2}-12\left\vert \lambda
\right\vert ^{2}e^{-2\Delta }.  \label{1cond2}
\end{eqnarray}%
\end{subequations}%

We may now attempt to obtain further conditions by calculating the Ricci
tensor using $T^{\left( X\right) }$ and $T^{\left( Y\right) },$ and then
equating. These are long computations that have been completed in \emph{Maple%
}, but below we outline the general approach. To express $\func{Ric}$ in
terms of octonions, let us first extend the Dirac operator (\ref{defDiracOp}%
) to $\Omega ^{1}\left( \mathbb{O}M\right) ,$ so that if $\hat{P}=\left(
P_{0},P\right) \in \Omega ^{1}\left( \mathbb{O}M\right) $, 
\begin{equation}
\slashed{D}\hat{P}_{a}=\delta ^{i}\left( D_{i}\hat{P}_{a}\right) .
\label{defDirac1form}
\end{equation}%
Similarly, define 
\begin{eqnarray}
\slashed{\nabla}\hat{P}_{a} &=&\delta ^{i}\left( \nabla _{i}\hat{P}%
_{a}\right) =\left( 
\begin{array}{c}
-\nabla _{i}P_{a}^{\ i} \\ 
\nabla ^{\alpha }\left( P_{0}\right) _{a}+\nabla _{i}P_{aj}^{\ }\varphi
^{ij\alpha }%
\end{array}%
\right)  \notag \\
&:&=\left( 
\begin{array}{c}
-\left( \func{div}P^{t}\right) _{a} \\ 
\func{grad}\left( P_{0}\right) _{a}+\func{curl}\left( P^{t}\right)%
\end{array}%
\right) .  \label{nabslash}
\end{eqnarray}%
We can rewrite (\ref{defDirac1form}) as 
\begin{equation*}
\delta ^{i}\left( D_{i}\hat{P}_{a}\right) =\delta ^{i}\left( D_{i}\left(
\left( P_{0}\right) _{a}+\delta ^{j}P_{aj}\right) \right)
\end{equation*}%
and hence 
\begin{equation}
\slashed{D}\hat{P}=\slashed{\nabla}\hat{P}+P_{0}\left( \slashed{D}1\right)
-P_{aj}\delta ^{i}\left( \delta ^{j}T_{i}\right) .  \label{DirTdef}
\end{equation}
Using these definitions, we have the following.

\begin{lemma}
\label{lemRicOct}Given a $G_{2}$-structure $\varphi $ with torsion $T$, the
Ricci curvature of the corresponding metric is given by 
\begin{eqnarray}
\frac{1}{2}\widehat{\func{Ric}}_{j} &=&-\delta ^{i}\left( \left(
d_{D}T\right) _{ij}\right)  \notag \\
&=&-\slashed{D}T-D\slashed{D}1  \label{RicDT}
\end{eqnarray}%
where $\widehat{\func{Ric}}_{a}=\func{Ric}_{ab}\delta ^{b}$ is an $\func{Im}%
\mathbb{O}$-valued $1$-form.
\end{lemma}

\begin{proof}
Using (\ref{dtinteg2}), we have%
\begin{eqnarray*}
\delta ^{i}\left( \left( d_{D}T\right) _{ij}\right) &=&\frac{1}{4}\delta
^{i}\left( \left( 
\begin{array}{c}
0 \\ 
\left( \pi _{7}\func{Riem}\right) _{ij}^{\ \ }%
\end{array}%
\right) \right) \\
&=&\frac{1}{4}\left( 
\begin{array}{c}
\left\langle \delta ^{i},\left( \pi _{7}\func{Riem}\right) _{ij}^{\ \
}\right\rangle \\ 
\delta ^{i}\times \left( \pi _{7}\func{Riem}\right) _{ij}^{\ \ }%
\end{array}%
\right) .
\end{eqnarray*}%
However, 
\begin{eqnarray*}
\left\langle \delta ^{i},\left( \pi _{7}\func{Riem}\right) _{ij}^{\ \
}\right\rangle &=&\delta ^{i\alpha }R_{ij\beta \gamma }\varphi _{\alpha }^{\
\beta \gamma }=R_{\alpha j\beta \gamma }\varphi ^{\alpha \beta \gamma }=0 \\
\left( \delta ^{i}\times \left( \pi _{7}\func{Riem}\right) _{ij}^{\ \
}\right) ^{\alpha } &=&\delta _{\beta }^{i}R_{ijmn}\varphi _{\gamma }^{\
mn}\varphi ^{\alpha \beta \gamma } \\
&=&R_{ijmn}\varphi _{\gamma }^{\ mn}\varphi ^{\alpha i\gamma
}=R_{ijmn}\left( \psi ^{mn\alpha i}+g^{m\alpha }g^{ni}-g^{mi}g^{n\alpha
}\right) \\
&=&-2\func{Ric}_{\ j}^{\alpha }.
\end{eqnarray*}%
Hence, we get the expression (\ref{RicDT}).

On the other hand, 
\begin{eqnarray*}
\delta ^{i}\left( \left( d_{D}T\right) _{ij}\right) &=&\delta ^{i}\left(
D_{i}T_{j}-D_{j}T_{i}\right) =\slashed{D}T_{j}-\delta ^{i}\left(
D_{j}T_{i}\right) \\
&=&\slashed{D}T_{j}-D_{j}\left( \delta ^{i}T_{i}\right) =\slashed{D}%
T_{j}-D_{j}\NEG{D}1.
\end{eqnarray*}
\end{proof}

Recall from (\ref{Dir1}) that $\slashed{D}1=\left( 
\begin{array}{c}
7\tau _{1} \\ 
-6\tau _{7}%
\end{array}%
\right) ,$ so only the $\slashed{D}T$ term in (\ref{RicDT}) involves
derivatives of $\tau _{14}$ and $\tau _{27}$. The expression (\ref{RicDT})
in fact contains all integrability conditions except one of the $\mathbf{27}$
components (related to $\func{Ric}^{\ast }$) and the $\mathbf{64}$
component. In fact, taking the $\mathbb{O}$-real part of (\ref{RicDT}) gives
a $\mathbf{7}$ component, taking the projections to $\Lambda _{7}^{2}$ and $%
\Lambda _{14}^{2}$ of the $\mathbb{O}$-imaginary part gives another $\mathbf{%
7}$ component and the $\mathbf{14}$ component, respectively. Of course, the
symmetric part of the $\mathbb{O}$-imaginary part of (\ref{RicDT}) gives the 
$\mathbf{1\oplus 27}$ components of the integrability conditions (scalar
curvature and the traceless Ricci tensor).

Now consider a $G_{2}$-structure $\varphi _{V}$ for some nowhere-vanishing
octonion section $V$. Then, (\ref{dtinteg2}) becomes 
\begin{eqnarray*}
\frac{1}{4}\pi _{7}^{\left( V\right) }\func{Riem} &=&d_{D^{\left( V\right)
}}T^{\left( V\right) } \\
&=&d_{D^{\left( V\right) }}\left( -\left( DV\right) V^{-1}+d\ln \left\vert
V\right\vert \right) \\
&=&-\left( d_{D}^{2}V\right) V^{-1}+\left( d\ln \left\vert V\right\vert
\right) \wedge \left( DV\right) V^{-1}+d\ln \left\vert V\right\vert \wedge
T^{\left( V\right) } \\
&=&-\left( d_{D}^{2}V\right) V^{-1}
\end{eqnarray*}%
Thus, 
\begin{equation*}
\left( d_{D}^{2}V\right) =-\frac{1}{4}\left( \pi _{7}^{\left( V\right) }%
\func{Riem}\right) V
\end{equation*}%
where $T^{\left( V\right) }$ is the torsion of $\varphi _{V}$ and we have
used (\ref{TorsV2}). Also, using (\ref{RicDT}) for $T^{\left( V\right) }$,
we have 
\begin{eqnarray*}
\frac{1}{2}\widehat{\func{Ric}}_{j} &=&-\delta ^{i}\circ _{V}\left( \left(
d_{D^{\left( V\right) }}T^{\left( V\right) }\right) _{ij}\right) \\
&=&-\delta ^{i}\circ _{V}\left( d_{D}\left( T^{\left( V\right) }V\right)
_{ij}V^{-1}-\left( \partial _{i}\ln \left\vert V\right\vert \right)
T_{j}^{\left( V\right) }+\left( \partial _{j}\ln \left\vert V\right\vert
\right) T_{i}^{\left( V\right) }\right) \\
&=&-\left( \delta ^{i}\left( d_{D}\left( T^{\left( V\right) }V\right)
_{ij}\right) \right) V^{-1}+\left( \left( \hat{\partial}\ln \left\vert
V\right\vert \right) \left( T_{j}^{\left( V\right) }V\right) \right) V^{-1}
\\
&&+\left( \partial _{j}\ln \left\vert V\right\vert \right) \slashed{D}%
^{\left( V\right) }1
\end{eqnarray*}%
where we have also used (\ref{associd1}) and (\ref{DtildeAV3}). Now
moreover, 
\begin{eqnarray*}
\delta ^{i}\left( d_{D}\left( T^{\left( V\right) }V\right) _{ij}\right)
&=&\delta ^{i}\left( D_{i}\left( T_{j}^{\left( V\right) }V\right)
-D_{j}\left( T_{i}^{\left( V\right) }V\right) \right) \\
&=&\slashed{D}\left( T_{j}^{\left( V\right) }V\right) -D_{j}\left( \delta
^{i}\left( T_{i}^{\left( V\right) }V\right) \right) \\
&=&\slashed{D}\left( T_{j}^{\left( V\right) }V\right) +D_{j}\left( \left( 
\NEG{D}^{\left( V\right) }1\right) V\right) .
\end{eqnarray*}%
Hence, we find that 
\begin{eqnarray}
\frac{1}{2}\widehat{\func{Ric}} &=&-\left( \slashed{D}\left( T^{\left(
V\right) }V\right) \right) V^{-1}-\left( D\left( \left( \slashed{D}^{\left(
V\right) }1\right) V\right) \right) V^{-1}  \label{RicTV} \\
&&+\left( \left( \hat{\partial}\ln \left\vert V\right\vert \right) \left(
T_{j}^{\left( V\right) }V\right) \right) V^{-1}+\left( \partial _{j}\ln
\left\vert V\right\vert \right) \slashed{D}^{\left( V\right) }1  \notag
\end{eqnarray}%
and then by comparing with (\ref{RicDT}), which we can also transpose since $%
\func{Ric}$ is symmetric, we conclude the following.

\begin{lemma}
\label{LemRicSame}Suppose $\varphi $ and $\varphi _{V}=\sigma _{V}\left(
\varphi \right) $ are isometric $G_{2}$-structures with torsion $T$ and $%
T^{\left( V\right) }$, respectively. Then, 
\begin{eqnarray}
\left( \slashed{D}T\right) V-\slashed{D}\left( T^{\left( V\right) }V\right)
&=&D\left( \left( \slashed{D}^{\left( V\right) }1\right) V\right) -\left( D%
\NEG{D}1\right) V  \label{TVcomp} \\
&&-\left( \left( \hat{\partial}\ln \left\vert V\right\vert \right) \left(
T_{j}^{\left( V\right) }V\right) \right) V^{-1}-\left( \partial _{j}\ln
\left\vert V\right\vert \right) \slashed{D}^{\left( V\right) }1  \notag
\end{eqnarray}%
and equivalently, 
\begin{eqnarray}
\left( \slashed{D}T\right) ^{t}V-\slashed{D}\left( T^{\left( V\right)
}V\right) &=&D\left( \left( \slashed{D}^{\left( V\right) }1\right) V\right)
-\left( D\NEG{D}1\right) ^{t}V  \label{TVcomp2} \\
&&-\left( \left( \hat{\partial}\ln \left\vert V\right\vert \right) \left(
T_{j}^{\left( V\right) }V\right) \right) V^{-1}-\left( \partial _{j}\ln
\left\vert V\right\vert \right) \slashed{D}^{\left( V\right) }1  \notag
\end{eqnarray}
\end{lemma}

We will use Lemma \ref{LemRicSame} to our torsions $T^{\left( X\right) }$
and $T^{\left( Y\right) }$, however, recall from (\ref{TXexp}) that our
expression for $T^{\left( X\right) }$ involves the transpose of an octonion
product, so in Lemma \ref{lemDtrans} below, we find how to work out $%
\slashed{D}$ of such expressions.

\begin{lemma}
\label{lemDtrans}Suppose $\hat{P}$ is an $\mathbb{O}M$-valued $1$-form.
Then, for any $V\in \Gamma \left( \mathbb{O}M\right) $, 
\begin{eqnarray}
\slashed{D}\left( \left( \hat{P}V\right) _{a}^{t}\right) &=&\left( \left( 
\func{curl}\hat{P}\right) ^{t}V\right) _{a}^{t}-\left\langle \left( \func{%
curl}\hat{P}\right) _{a}^{t}+\delta _{a}\left( \func{div}\hat{P}\right) ,%
\bar{V}\right\rangle \hat{1}  \label{DirPWt} \\
&&+\slashed{D}\left\langle \hat{P}_{a},\bar{V}\right\rangle +\left\langle 
\hat{P}_{j},\delta _{a}\left( \nabla _{i}\bar{V}\right) \right\rangle \delta
^{i}\delta ^{j}  \notag \\
&&-\left\langle \hat{P}_{j},\delta _{a}\bar{V}\right\rangle \delta
^{i}\left( \delta ^{j}T_{i}\right) -\left\langle \hat{P}_{j},\left[
T_{i},\delta _{a},\bar{V}\right] \right\rangle \delta ^{i}\delta ^{j}. 
\notag
\end{eqnarray}
\end{lemma}

\begin{proof}
We have 
\begin{eqnarray*}
\slashed{D}\left( \left( \hat{P}V\right) _{a}^{t}\right) &=&\delta
^{i}\left( D_{i}\left( \left( \hat{P}V\right) _{a}^{t}\right) \right) \\
&=&D_{i}\left( \delta ^{i}\left( \hat{P}V\right) _{a}^{t}\right) .
\end{eqnarray*}%
However, 
\begin{equation*}
\delta ^{i}\left( \hat{P}V\right) _{a}^{t}=\left\langle \hat{P}_{a},\bar{V}%
\right\rangle \delta ^{i}+\left\langle \hat{P}_{j},\delta _{a}\bar{V}%
\right\rangle \delta ^{i}\delta ^{j}.
\end{equation*}%
Then, using (\ref{octocovprod}) and (\ref{DXmet}), we get%
\begin{eqnarray*}
D_{i}\left( \delta ^{i}\left( \hat{P}V\right) _{a}^{t}\right) &=&\NEG%
{D}\left\langle \hat{P}_{a},\bar{V}\right\rangle +\left\langle \nabla _{i}%
\hat{P}_{j},\delta _{a}\bar{V}\right\rangle \delta ^{i}\delta ^{j} \\
&&+\left\langle \hat{P}_{j},\delta _{a}\left( \nabla _{i}\bar{V}\right)
\right\rangle \delta ^{i}\delta ^{j}-\left\langle \hat{P}_{j},\left[
T_{i},\delta _{a},\bar{V}\right] \right\rangle \delta ^{i}\delta ^{j} \\
&&-\left\langle \hat{P}_{j},\delta _{a}\bar{V}\right\rangle \delta
^{i}\left( \delta ^{j}T_{i}\right) \\
&=&-\left\langle \left( \func{div}\hat{P}\right) V,\delta _{a}\right\rangle 
\hat{1}+\left\langle \left( \func{curl}\hat{P}\right) _{k}^{t}V,\delta
_{a}\right\rangle \delta ^{k} \\
&&+\slashed{D}\left\langle \hat{P}_{a},\bar{V}\right\rangle +\left\langle 
\hat{P}_{j},\delta _{a}\left( \nabla _{i}\bar{V}\right) \right\rangle \delta
^{i}\delta ^{j} \\
&&-\left\langle \hat{P}_{j},\delta _{a}\bar{V}\right\rangle \delta
^{i}\left( \delta ^{j}T_{i}\right) -\left\langle \hat{P}_{j},\left[
T_{i},\delta _{a},\bar{V}\right] \right\rangle \delta ^{i}\delta ^{j}
\end{eqnarray*}%
from which we obtain (\ref{DirPWt}).
\end{proof}

\qquad Now recall from (\ref{TXexp}) and (\ref{TYexp0})that we can write $%
T^{\left( X\right) }$ and $T^{\left( Y\right) }W$ as: 
\begin{subequations}%
\begin{eqnarray}
T^{\left( X\right) } &=&-\left( \hat{Q}\bar{W}\right) ^{t}+\left(
K_{1}^{\left( X\right) }\right) ^{t} \\
T^{\left( Y\right) }W &=&\hat{Q}+K_{1}^{\left( Y\right) }
\end{eqnarray}%
\end{subequations}%
where 
\begin{subequations}
\begin{eqnarray}
\left( K_{1}^{\left( X\right) }\right) ^{t} &=&-\frac{1}{4}\delta \left( 
\begin{array}{c}
-\frac{1}{2}\lambda _{2}e^{-\Delta } \\ 
\partial \Delta 
\end{array}%
\right) +\left( \lambda _{1}e^{-\Delta }w-\partial \Delta \right) \hat{1} \\
K_{1}^{\left( Y\right) } &=&-\frac{1}{4}\left( \delta \left( \left( \widehat{%
\partial \Delta }\right) W\right) \right) +\lambda _{1}e^{-\Delta }\delta +%
\frac{1}{2}\lambda _{2}e^{-\Delta }\delta W+\left( d\ln \left\vert
Y\right\vert \right) W \\
\hat{Q} &=&\left( 
\begin{array}{c}
-Q_{1} \\ 
\frac{1}{4}\left( Q_{0}-2\lambda _{1}e^{-\Delta }\right) \delta +Q_{2}%
\end{array}%
\right) .
\end{eqnarray}%
\end{subequations}%
Now from (\ref{DirPWt}), 
\begin{eqnarray}
\slashed{D}T^{\left( X\right) } &=&-\slashed{D}\left( \left( \hat{Q}\bar{W}%
\right) ^{t}\right) +\slashed{D}\left( K_{1}^{\left( X\right) }\right) ^{t} 
\notag \\
&=&-\left( \left( \func{curl}\hat{Q}\right) ^{t}\bar{W}\right)
_{a}^{t}+\left\langle \left( \func{curl}\hat{Q}\right) _{a}^{t}+\delta
_{a}\left( \func{div}\hat{Q}\right) ,W\right\rangle \hat{1}  \label{DirTX} \\
&&+\slashed{D}\left\langle \hat{Q}_{a},W\right\rangle -\left\langle \hat{Q}%
_{j},\delta _{a}\left( \nabla _{i}W\right) \right\rangle \delta ^{i}\delta
^{j}  \notag \\
&&+\left\langle \hat{Q}_{j},\delta _{a}W\right\rangle \delta ^{i}\left(
\delta ^{j}T_{i}^{\left( X\right) }\right) +\left\langle \hat{Q}_{j},\left[
T_{i}^{\left( X\right) },\delta _{a},W\right] \right\rangle \delta
^{i}\delta ^{j}.  \notag \\
&&+\slashed{D}\left\langle \left( K_{1}^{\left( X\right) }\right)
_{a},1\right\rangle +\left( \func{curl}K_{1}^{\left( X\right) }\right) _{a} 
\notag \\
&&-\left\langle \left( \func{curl}K_{1}^{\left( X\right) }\right)
_{a}+\delta _{a}\left( \func{div}K_{1}^{\left( X\right) }\right)
,1\right\rangle \hat{1}  \notag \\
&&-\left\langle \left( K_{1}^{\left( X\right) }\right) _{j},\delta
_{a}T_{i}\right\rangle \delta ^{i}\delta ^{j}-\left\langle \left(
K_{1}^{\left( X\right) }\right) _{j},\delta _{a}\right\rangle \delta
^{i}\left( \delta ^{j}T_{i}^{\left( X\right) }\right) .  \notag
\end{eqnarray}%
For brevity, we focus on terms that include derivatives of $Q_{2}$. From (%
\ref{Q0Q1rel2}) we know that the other components of $G$ - $Q_{0}$ and $Q_{1}
$ can be expressed in terms of $e^{-\Delta },d\Delta ,w,w_{0},\left\vert
W\right\vert ^{2}.$ Similarly, $\tau _{1}^{\left( X\right) }$ and $\tau
_{7}^{\left( X\right) }$, and hence $\slashed{D}1$, also depend only on
these variables. Also, from (\ref{dw0}) and (\ref{dWsq}) we know derivatives
of $w_{0}$ and $\left\vert W\right\vert ^{2}$ are also expressed in terms of
these variables and we know that $\nabla w$ and $T^{\left( X\right) }$ are
also expressed in terms of these variables and are linear $Q_{2}$. Also, $%
\nabla _{i}\left\langle \hat{Q}_{2},W\right\rangle =\nabla _{i}\left(
Q_{2}\left( w\right) \right) $ and 
\begin{eqnarray*}
\left\langle \left( \func{curl}\hat{Q}_{2}\right) _{a}^{t},W\right\rangle 
&=&\left( \nabla _{i}\left( Q_{2}\right) _{mj}\right) \varphi _{a}^{\
ij}w^{m} \\
&=&\nabla _{i}Q_{2}\left( w\right) _{j}\varphi _{a}^{\ ij}-\varphi _{a}^{\
ij}\left( Q_{2}\right) _{mj}\nabla _{i}w^{m},
\end{eqnarray*}%
but we know from (\ref{Q2wrel}) and (\ref{dwexp}) that $Q_{2}\left( w\right) 
$ and $\nabla w$ are also given as polynomial expressions in our set of
basic variables. Therefore, let us rewrite (\ref{DirTX}) as 
\begin{equation}
\left( \slashed{D}T^{\left( X\right) }\right) ^{t}=-\left( \func{curl}\hat{Q}%
_{2}\right) ^{t}\bar{W}+\left\langle \delta \left( \func{div}\hat{Q}%
_{2}\right) ,W\right\rangle \hat{1}+K_{2}^{\left( X\right) }  \label{DirTX2}
\end{equation}%
where $K_{2}^{\left( X\right) }$ contains no derivatives of $Q_{2}$. On the
other hand, using (\ref{nabslash}), 
\begin{eqnarray}
\slashed{D}\left( T^{\left( Y\right) }W\right)  &=&\slashed{D}\hat{Q}+\NEG%
{D}K_{1}^{\left( Y\right) }  \notag \\
&=&\slashed{\nabla}\hat{Q}-Q_{1}\left( \slashed{D}1\right) -\left( \func{Im}%
_{\mathbb{O}}\hat{Q}\right) _{aj}\delta ^{i}\left( \delta ^{j}T_{i}^{\left(
X\right) }\right)   \notag \\
&=&\func{curl}\hat{Q}_{2}-\left( \func{div}\hat{Q}_{2}\right) \hat{1}%
+K_{2}^{\left( Y\right) }  \label{DirTY2}
\end{eqnarray}%
where $K_{2}^{\left( X\right) }$ also contains no derivatives of $Q_{2}$.
Overall, from (\ref{TVcomp2}) we have, 
\begin{eqnarray}
\slashed{D}\left( T^{\left( Y\right) }W\right) -\left( \slashed{D}T^{\left(
X\right) }\right) ^{t}W &=&\left( D\slashed{D}1\right) ^{t}W-D\left( \left( 
\NEG{D}^{\left( W\right) }1\right) W\right)  \\
&&+\left( \left( \hat{\partial}\ln \left\vert W\right\vert \right) \left(
T_{j}^{\left( Y\right) }W\right) \right) W^{-1}+\left( \partial _{j}\ln
\left\vert W\right\vert \right) \slashed{D}^{\left( W\right) }1  \notag
\end{eqnarray}%
and hence, 
\begin{equation}
\func{curl}\hat{Q}_{2}+\left( \func{curl}\hat{Q}_{2}\right) ^{t}\left\vert
W\right\vert ^{2}-\left( \func{div}\hat{Q}_{2}\right) \hat{1}+\left\langle
\left( \func{div}\hat{Q}_{2}\right) \bar{W},\delta \right\rangle W=K_{3}
\label{curleq1}
\end{equation}%
where 
\begin{eqnarray}
K_{3} &=&\left( D\slashed{D}1\right) ^{t}W-D\left( \left( \slashed{D}%
^{\left( W\right) }1\right) W\right)  \\
&&+\left( \left( \hat{\partial}\ln \left\vert W\right\vert \right) \left(
T_{j}^{\left( Y\right) }W\right) \right) W^{-1}+\left( \partial _{j}\ln
\left\vert W\right\vert \right) \slashed{D}^{\left( W\right)
}1-K_{2}^{\left( Y\right) }+K_{2}^{\left( X\right) }W  \notag
\end{eqnarray}%
and also does not depend on derivatives of $Q_{2}.$ Taking the real part of (%
\ref{curleq1}), we find an expression for $\func{div}Q_{2}$:%
\begin{eqnarray}
\func{div}Q_{2} &=&-\frac{1}{6}GG\left( w\right) -\frac{3}{2}\left(
3+\left\vert W\right\vert ^{2}\right) Q_{2}\left( d\Delta \right) -2\left(
2w_{0}\lambda _{1}-\lambda _{2}\right) e^{-\Delta }Q_{1}-\frac{3}{28}%
d\left\langle w,d\Delta \right\rangle   \label{divQ2} \\
&&-\frac{3}{56}d\Delta \left( \left( 3\left\vert W\right\vert ^{2}+65\right)
\left\langle w,d\Delta \right\rangle -108w_{0}\mu e^{-4\Delta }+\left(
154+120w_{0}\lambda _{2}-50\left\vert W\right\vert ^{2}\lambda _{1}\right)
e^{-\Delta }\right)   \notag \\
&&+\frac{1}{14}w\left( 7\left\vert G\right\vert ^{2}+12w_{0}\lambda
_{1}\left( \lambda _{2}-\mu e^{-3\Delta }\right) e^{-2\Delta }+\left(
126\lambda _{1}^{2}-6\lambda _{2}^{2}\right) e^{-2\Delta }+12\lambda _{2}\mu
e^{-5\Delta }\right) .  \notag
\end{eqnarray}%
Using (\ref{divQ2}), we rewrite (\ref{curleq1}) as 
\begin{equation}
\left\vert W\right\vert ^{2}\left( \func{curl}\hat{Q}_{2}\right) ^{t}+\func{%
curl}\hat{Q}_{2}=K_{4}  \label{RicCond}
\end{equation}%
where $K_{4}$ does not contain derivatives of $Q_{2}$. Taking the transpose
of (\ref{RicCond}), and solving for $\func{curl}\hat{Q}_{2}$ we get%
\begin{equation}
\left( 1-\left\vert W\right\vert ^{4}\right) \left( \func{curl}\hat{Q}%
_{2}\right) =K_{4}-\left\vert W\right\vert ^{2}K_{4}^{t}.  \label{nabslashQ2}
\end{equation}%
Thus we see that if $\left\vert W\right\vert ^{2}\neq 1,$ we can obtain a
solution for $\func{curl}\hat{Q}_{2}$. In general, it's a very long
expression, which we will not write down.

Now using the expressions for $\func{div}Q_{2}$ and $\func{curl}\hat{Q}_{2}$
in (\ref{RicDT}) or (\ref{RicTV}), we can obtain (after some manipulations)
the expression for $\func{Ric}$ (\ref{ric7cond}) that comes from Einstein's
equations. Moreover, it can be seen that the expression for $\func{div}Q_{2}$
and $\func{curl}\hat{Q}_{2}$, together with (\ref{ric7cond}) and (\ref%
{1cond2}), do in fact imply (\ref{RicDT}) and (\ref{RicTV}). Hence we do not
obtain any other conditions.

Also, schematically, we can write 
\begin{subequations}%
\label{schemeom} 
\begin{eqnarray}
\ast dG &\sim &\func{Skew}\left( \func{curl}Q_{2}\right) +\text{terms
without derivatives of }Q_{2} \\
d\ast G &\sim &\iota _{\psi }\left( \func{Sym}\left( \func{curl}Q_{2}\right)
\right) +\left( \func{div}Q_{2}\right) \wedge \varphi \ +\text{terms without
derivatives of }Q_{2}.
\end{eqnarray}%
\end{subequations}%
It can be shown that from the expressions for $\func{curl}Q_{2}$ and $\func{%
div}Q_{2},$ together with (\ref{1cond2}), we do get $dG=0$ and $d\ast G$
satisfies the equation of motion (\ref{eom7D}). Moreover, since (\ref%
{schemeom}) also uniquely determine $\func{curl}Q_{2}$ and $\func{div}Q_{2}$%
, together with Einstein's equations (\ref{ric7cond}) and (\ref{1cond2}),
they are equivalent to the integrability conditions (\ref{RicDT}) and (\ref%
{RicTV}).

\section{Concluding remarks}

The main results in this paper were to reformulate $N=1$ Killing spinor
equations for compactifications of $11$-dimensional supergravity on $\func{%
AdS}_{4}$ in terms of the torsion of $G_{2}$-structures on $7$-manifolds. We
obtain either a single complexified $G_{2}$-structure that corresponds to
complex spinor that solves the $7$-dimensional Killing spinor equation, or
correspondingly 2 real $G_{2}$-structures. To obtain these results we used
the octonion bundle structure that was originally developed in \cite%
{GrigorianOctobundle}, which is a structure that is specific to $7$
dimensions. In the process, we also found that on a compact $7$-manifold,
the defining vector field $w$ that relates the two $\mathbb{C}$-real
octonions (or equivalently, spinors) is necessarily vanishing at some point.
This shows that the $SU\left( 3\right) $-structure defined by $w$ is not
globally defined, and is hence only local. This is in sharp contrast with
the $N=2$ case studied in \cite{SparksFlux}, where there is always a
nowhere-vanishing vector field on the $7$-manifold. 

In order to obtain the torsion of the complexified $G_{2}$-structure we
expressed the Killing spinor equation solutions as a complexified octonion
section $Z$. This had very different properties depending on whether it
could be a zero divisor or not. It is interesting to understand the physical
significance of the submanifold where $Z$ is a zero divisor. In the $N=2$
supersymmetry case, we would need to have $4$ $\mathbb{C}$-real octonion
sections, which could be combined together to give a \emph{quaternionic
octonion}, i.e. objects that pointwise lie in $\mathbb{H}\otimes \mathbb{O}.$
This should then define some kind of a quaternionic $G_{2}$-structure.
However, even from a pure algebraic point of view it would be challenging to
understand this, because this is not as simple as just changing the
underlying field from $\mathbb{R}$ to $\mathbb{C}$, as is the case for
complexified $G_{2}$-structures. However such an approach could give an
alternative point of view to the results in \cite{SparksFlux}.

\appendix

\section{Proofs of identities}

\begin{lemma}
\label{LemNablaZZst}\setcounter{equation}{0}Suppose $Z=X+iY$ satisfies the
equations (\ref{susyeqoct1a2a}), then 
\begin{equation}
\left\vert X\right\vert ^{2}+\left\vert Y\right\vert ^{2}=k_{1}e^{\Delta }
\label{xysol}
\end{equation}%
for an arbitrary real constant $k$.
\end{lemma}

\begin{proof}
Let us consider the derivative of $\left\langle Z,Z^{\ast }\right\rangle $: 
\begin{eqnarray}
\nabla \left\langle Z,Z^{\ast }\right\rangle &=&\left\langle DZ,Z^{\ast
}\right\rangle +\left\langle Z,DZ^{\ast }\right\rangle  \notag \\
&=&\left\langle -\frac{1}{2}\mu e^{-4\Delta }\delta Z+i\delta G^{\left(
4\right) }\left( Z\right) +12iG^{\left( 3\right) }\left( Z\right) ,Z^{\ast
}\right\rangle  \label{nabzzst1} \\
&&+\left\langle Z,-\frac{1}{2}\mu e^{-4\Delta }\delta Z^{\ast }-i\delta
G^{\left( 4\right) }\left( Z^{\ast }\right) -12iG^{\left( 3\right) }\left(
Z^{\ast }\right) \right\rangle  \notag
\end{eqnarray}%
However from the properties of $G^{\left( 3\right) }$ and $G^{\left(
4\right) }$ from Lemma \ref{LemG3G4}, we get 
\begin{eqnarray*}
\left\langle Z,i\delta G^{\left( 4\right) }\left( Z^{\ast }\right)
\right\rangle &=&-i\left\langle \delta G^{\left( 4\right) }\left( Z\right)
,Z^{\ast }\right\rangle -8i\left\langle G^{\left( 3\right) }\left( Z\right)
,Z^{\ast }\right\rangle \\
\left\langle Z,12iG^{\left( 3\right) }\left( Z^{\ast }\right) \right\rangle
&=&\left\langle 12iG^{\left( 3\right) }\left( Z\right) ,Z^{\ast
}\right\rangle .
\end{eqnarray*}%
Also, 
\begin{eqnarray*}
\left\langle \delta Z,Z^{\ast }\right\rangle +\left\langle Z,\delta Z^{\ast
}\right\rangle &=&\left\langle \delta ,Z^{\ast }\bar{Z}+Z\overline{Z^{\ast }}%
\right\rangle \\
&=&\left\langle \delta ,2\func{Re}_{\mathbb{C}}\left( Z^{\ast }\bar{Z}%
\right) \right\rangle =0
\end{eqnarray*}%
where we have used the property $\func{Re}_{\mathbb{C}}\left( Z^{\ast }\bar{Z%
}\right) =\func{Re}_{\mathbb{O}}\left( Z^{\ast }\bar{Z}\right) $ from (\ref%
{ZZstbar}).Thus, (\ref{nabzzst1}) becomes 
\begin{eqnarray}
\nabla \left\langle Z,Z^{\ast }\right\rangle &=&2i\left\langle \delta
G^{\left( 4\right) }\left( Z\right) ,Z^{\ast }\right\rangle +8i\left\langle
G^{\left( 3\right) }\left( Z\right) ,Z^{\ast }\right\rangle  \notag \\
&=&-2\left\langle \delta \left( MZ-i\lambda Z^{\ast }\right) ,Z^{\ast
}\right\rangle +8i\left\langle G^{\left( 3\right) }\left( Z\right) ,Z^{\ast
}\right\rangle  \notag \\
&=&2\left\langle M,\left( \delta Z^{\ast }\right) \bar{Z}\right\rangle
+8i\left\langle G^{\left( 3\right) }\left( Z\right) ,Z^{\ast }\right\rangle
\label{nabzzst2}
\end{eqnarray}%
where we used the property $\left\langle \delta Z^{\ast },Z^{\ast
}\right\rangle =\left\langle \delta ,Z^{\ast }\overline{Z^{\ast }}%
\right\rangle =0$. Decomposing (\ref{susyeqoct1a}) into $\mathbb{C}$-real
and $\mathbb{C}$-imaginary parts gives us 
\begin{subequations}%
\label{mxmy} 
\begin{eqnarray}
G^{4}\left( X\right) &=&-MY+\lambda _{1}e^{-\Delta }X+\lambda _{2}e^{-\Delta
}Y \\
G^{4}\left( Y\right) &=&MX+\lambda _{2}e^{-\Delta }X-\lambda _{1}e^{-\Delta
}Y.
\end{eqnarray}%
\end{subequations}
Using this decomposition, we get 
\begin{eqnarray}
8\left\langle G^{\left( 3\right) }\left( Z\right) ,Z^{\ast }\right\rangle
&=&8\left\langle G^{\left( 3\right) }\left( X\right) ,X\right\rangle
+8\left\langle G^{\left( 3\right) }\left( Y\right) ,Y\right\rangle  \notag \\
&=&-2\left\langle \delta G^{\left( 4\right) }\left( X\right) ,X\right\rangle
-2\left\langle \delta G^{\left( 4\right) }\left( Y\right) ,Y\right\rangle 
\notag \\
&=&2\left\langle \delta \left( MY-\lambda _{1}X-\lambda _{2}Y\right)
,X\right\rangle -2\left\langle \delta \left( MX+\lambda _{2}X-\lambda
_{1}Y\right) ,Y\right\rangle  \label{G3zzst} \\
&=&2\left\langle \delta \left( MY\right) ,X\right\rangle -2\left\langle
\delta \left( MX\right) ,Y\right\rangle -2\lambda _{2}\left\langle \delta
Y,X\right\rangle -2\lambda _{2}\left\langle \delta X,Y\right\rangle \\
&=&-2\left\langle M,\left( \delta X\right) \bar{Y}\right\rangle
+2\left\langle M,\left( \delta Y\right) \bar{X}\right\rangle
\end{eqnarray}%
Also note that 
\begin{eqnarray}
\left( \delta Z^{\ast }\right) \bar{Z} &=&\left( \delta \left( X-iY\right)
\right) \left( \bar{X}+i\bar{Y}\right)  \notag \\
&=&\delta \left( \left\vert X\right\vert ^{2}+\left\vert Y\right\vert
^{2}\right) +i\left( \left( \left( \delta X\right) \bar{Y}\right) -\left(
\delta Y\right) \bar{X}\right)  \label{Mdelzstz}
\end{eqnarray}%
Thus, substituting (\ref{G3zzst}) and (\ref{Mdelzstz}) into (\ref{nabzzst2}%
), we may conclude that 
\begin{equation}
\nabla \left\langle Z,Z^{\ast }\right\rangle =2\left\langle M,\delta
\right\rangle \left( \left\vert X\right\vert ^{2}+\left\vert Y\right\vert
^{2}\right) =\partial \Delta \left( \left\vert X\right\vert ^{2}+\left\vert
Y\right\vert ^{2}\right)  \label{nablazzst}
\end{equation}%
and hence, (\ref{xysol}).
\end{proof}

\begin{lemma}
Let $A=$ $a_{0}+\hat{\alpha}$ be a (bi)octonion section and let $G$ be a $4$%
-form with a decomposition with respect to a $G_{2}$-structure $\varphi $
given by (\ref{Gdecomp}). Then, $G^{\left( 4\right) }\left( A\right) $
satisfies the following relation 
\begin{equation}
G^{\left( 4\right) }\left( A\right) =-\left( 
\begin{array}{c}
\frac{7}{6}a_{0}Q_{0}+\frac{2}{3}\left\langle Q_{1},\alpha \right\rangle \\ 
\frac{2}{3}a_{0}Q_{1}-\frac{1}{6}Q_{0}\alpha -\frac{2}{3}Q_{2}\left( \alpha
\right)%
\end{array}%
\right)  \label{Q0Q1rel}
\end{equation}
\end{lemma}

\begin{proof}
We have 
\begin{eqnarray}
G^{\left( 4\right) }\left( A\right) &=&a_{0}G^{\left( 4\right) }\left(
1\right) +\alpha ^{i}G^{\left( 4\right) }\left( \delta _{i}\right)  \notag \\
&=&a_{0}G^{\left( 4\right) }\left( 1\right) +\alpha ^{i}\left( \delta
^{i}G^{\left( 4\right) }\left( 1\right) +8G_{i}^{3}\left( 1\right) \right) 
\notag \\
&=&AG^{\left( 4\right) }\left( 1\right) +8\alpha ^{i}G_{i}^{\left( 3\right)
}\left( 1\right)  \label{G4G3}
\end{eqnarray}%
where we have used properties of $G^{\left( 4\right) }$ from Lemma \ref%
{LemG3G4}. Now, using Lemma \ref{LemG4G3} gives us (\ref{Q0Q1rel}).
\end{proof}

\begin{lemma}
\label{LemG4G3FA}Let $A=$ $a_{0}+\hat{\alpha}$ be a (bi)octonion section,
then the quantities $G^{4}\left( A\right) $ and $G^{\left( 3\right) }\left(
A\right) $ satisfy the following relation 
\begin{equation}
\delta \left( G^{\left( 4\right) }\left( A\right) \right) +12G^{\left(
3\right) }\left( A\right) =F\left( A\right) +\left( \hat{Q}_{2}\bar{A}%
\right) ^{T}  \label{G4G3FA}
\end{equation}%
where 
\begin{equation}
F\left( A\right) =\left( 
\begin{array}{c}
-\frac{4}{3}a_{0}Q_{1}+\frac{11}{6}\alpha Q_{0}-\frac{2}{3}Q_{2}\left(
\alpha \right) +Q_{1}\times \alpha \\ 
\left( \frac{1}{3}\left\langle Q_{1},\alpha \right\rangle +\frac{5}{6}%
a_{0}Q_{0}\right) \delta +\left( -\frac{1}{6}Q_{0}\alpha -\frac{1}{3}%
a_{0}Q_{1}+\frac{1}{3}Q_{2}\left( \alpha \right) \right) \lrcorner \varphi
+\alpha Q_{1}%
\end{array}%
\right)  \label{FA}
\end{equation}
\end{lemma}

\begin{proof}
The expression (\ref{Q0Q1rel}) gives us a relation that expresses $G^{\left(
4\right) }\left( A\right) $ in terms of $A$ and components of $G$ with
respect to $\varphi $. Similarly, let us re-express $G^{\left( 3\right)
}\left( A\right) .$ Using (\ref{octoenvelopCliff}), we can show that 
\begin{equation}
G_{a}^{\left( 3\right) }\left( \delta _{b}\right) =-\delta _{b}G_{a}^{\left(
3\right) }-6G_{ab}^{\left( 2\right) }  \label{G3G2}
\end{equation}%
where $G_{ab}^{\left( 2\right) }:=\frac{1}{144}G_{abcd}\delta ^{c}\delta
^{d}.$ Therefore, 
\begin{eqnarray}
G_{a}^{\left( 3\right) }\left( A\right) &=&a_{0}G_{a}^{\left( 3\right)
}+\alpha ^{i}G_{a}^{\left( 3\right) }\left( \delta _{i}\right)  \notag \\
&=&a_{0}G_{a}^{\left( 3\right) }-\alpha G_{a}^{\left( 3\right) }+6\left(
\alpha \lrcorner G^{\left( 2\right) }\right) _{a}  \notag \\
&=&\bar{A}G_{a}^{\left( 3\right) }+6\left( \alpha \lrcorner G^{\left(
2\right) }\right) _{a}
\end{eqnarray}%
Note however, that if we compare $G^{\left( 3\right) }$ and $\delta
G^{\left( 4\right) }$ in (\ref{G4G3Q}), we find 
\begin{equation}
G_{a}^{\left( 3\right) }=-\frac{1}{8}\left( \delta _{a}G^{\left( 4\right) }+%
\frac{2}{3}\left( Q_{1}\right) _{a}-\frac{1}{6}Q_{0}\delta _{a}-\frac{2}{3}%
\left( \hat{Q}_{2}\right) _{a}\right)
\end{equation}%
where $\hat{Q}_{a}=\left( Q_{2}\right) _{ab}\delta ^{b}$. So overall, 
\begin{equation*}
G_{a}^{\left( 3\right) }\left( A\right) =-\frac{1}{8}\bar{A}\left( \delta
_{a}G^{\left( 4\right) }+\frac{2}{3}\left( Q_{1}\right) _{a}-\frac{1}{6}%
Q_{0}\delta _{a}-\frac{2}{3}\left( \hat{Q}_{2}\right) _{a}\right) +6\left(
\alpha \lrcorner G^{\left( 2\right) }\right) _{a}
\end{equation*}%
It's not difficult to work out $\left( \alpha \lrcorner G^{\left( 2\right)
}\right) $ explicitly. In fact, note that $G^{\left( 2\right) }$ is $\mathbb{%
O}$-imaginary, since $\delta ^{c}\delta ^{d}$ has a symmetric $\mathbb{O}$%
-real part. So, 
\begin{equation*}
\left( \alpha \lrcorner G^{\left( 2\right) }\right) _{a}^{\ m}=\frac{1}{144}%
\alpha ^{b}G_{bacd}\varphi ^{cdm}
\end{equation*}%
Just plugging in the expression (\ref{Gdecomp}) for $G$, we find 
\begin{eqnarray*}
\alpha \lrcorner G^{\left( 2\right) } &=&\frac{1}{72}\left( 2Q_{0}\alpha
-Q_{1}\times \alpha +Q_{2}\left( \alpha \right) \right) \lrcorner \varphi \\
&&+\frac{1}{36}\left\langle Q_{1},\alpha \right\rangle \delta +\frac{1}{72}%
\left( Q_{2}\times \alpha +\left( Q_{2}\times \alpha \right) ^{t}+\alpha
Q_{1}+Q_{1}\alpha \right) \\
&=&\frac{1}{36}\left\langle Q_{1},\alpha \right\rangle \delta +\frac{1}{72}%
\left( 2Q_{0}\alpha +Q_{2}\left( \alpha \right) \right) \lrcorner \varphi -%
\frac{1}{72}\psi \left( Q_{1},\alpha \right) \\
&&-\frac{1}{72}\left( Q_{2}\times \alpha +\left( Q_{2}\times \alpha \right)
^{t}\right) -\frac{1}{36}Q_{1}\alpha
\end{eqnarray*}%
Now, overall, 
\begin{eqnarray*}
\delta _{a}G^{\left( 4\right) }\left( A\right) +12G_{a}^{\left( 3\right)
}\left( A\right) &=&\delta _{a}\left( G^{\left( 4\right) }\left( A\right)
\right) -\frac{3}{2}\bar{A}\left( \delta _{a}G^{\left( 4\right) }\right)
-\left( Q_{1}\right) _{a}\bar{A}+\frac{1}{4}Q_{0}\left( \bar{A}\delta
_{a}\right) +\bar{A}\left( \hat{Q}_{2}\right) _{a} \\
&&+2\left\langle Q_{1},\alpha \right\rangle \delta _{a}+\left( 2Q_{0}\alpha
+Q_{2}\left( \alpha \right) \right) \lrcorner \varphi _{a}-\psi \left(
Q_{1},\alpha \right) _{a}-2\left( Q_{1}\right) _{a}\alpha \\
&&-\left( Q_{2}\times \alpha +\left( Q_{2}\times \alpha \right) ^{t}\right)
_{a}
\end{eqnarray*}%
Note however that 
\begin{equation*}
\bar{A}\hat{Q}_{2}=\left( 
\begin{array}{c}
Q_{2}\left( \alpha \right) \\ 
a_{0}Q_{2}+Q_{2}\times \alpha%
\end{array}%
\right)
\end{equation*}%
and similarly, 
\begin{equation*}
\hat{Q}_{2}\bar{A}=\left( 
\begin{array}{c}
Q_{2}\left( \alpha \right) \\ 
a_{0}Q_{2}-Q_{2}\times \alpha%
\end{array}%
\right)
\end{equation*}%
Hence,%
\begin{eqnarray*}
\left( \hat{Q}_{2}\bar{A}\right) ^{t} &=&\left( 
\begin{array}{c}
Q_{2}\left( \alpha \right) \\ 
a_{0}Q_{2}-\left( Q_{2}\times \alpha \right) ^{t}%
\end{array}%
\right) \\
&=&\bar{A}\hat{Q}_{2}-\left( 
\begin{array}{c}
0 \\ 
Q_{2}\times \alpha +\left( Q_{2}\times \alpha \right) ^{t}%
\end{array}%
\right)
\end{eqnarray*}%
Therefore, we can write 
\begin{eqnarray*}
\delta _{a}G^{\left( 4\right) }\left( A\right) +12G_{a}^{\left( 3\right)
}\left( A\right) &=&\delta _{a}\left( G^{\left( 4\right) }\left( A\right)
\right) -\frac{3}{2}\bar{A}\left( \delta _{a}G^{\left( 4\right) }\right)
-\left( Q_{1}\right) _{a}\bar{A}+\frac{1}{4}Q_{0}\left( \bar{A}\delta
_{a}\right) \\
&&+2\left\langle Q_{1},\alpha \right\rangle \delta _{a}+\left( 2Q_{0}\alpha
+Q_{2}\left( \alpha \right) \right) \lrcorner \varphi _{a}-\psi \left(
Q_{1},\alpha \right) _{a}-2\left( Q_{1}\right) _{a}\alpha \\
&&+\left( \hat{Q}_{2}\bar{A}\right) _{a}^{t}
\end{eqnarray*}%
Now using the expression (\ref{Q0Q1rel}) for $G^{\left( A\right) }\left(
A\right) $, we find, 
\begin{equation*}
\delta _{a}\left( G^{\left( 4\right) }\left( A\right) \right) =\left( 
\begin{array}{c}
\frac{2}{3}a_{0}Q_{1}-\frac{1}{6}Q_{0}\alpha -\frac{2}{3}Q_{2}\left( \alpha
\right) \\ 
-\left( \frac{7}{6}a_{0}Q_{0}+\frac{2}{3}\left\langle Q_{1},\alpha
\right\rangle \right) \delta +\left( \frac{2}{3}a_{0}Q_{1}-\frac{1}{6}%
Q_{0}\alpha -\frac{2}{3}Q_{2}\left( \alpha \right) \right) \lrcorner \varphi%
\end{array}%
\right) .
\end{equation*}%
Using (\ref{G41}) for $G^{\left( 4\right) }\left( 1\right) $ we also get%
\begin{equation*}
\delta \left( G^{\left( 4\right) }\left( 1\right) \right) =\left( 
\begin{array}{c}
Q_{1} \\ 
-\frac{7}{4}Q_{0}\delta +Q_{1}\lrcorner \varphi%
\end{array}%
\right)
\end{equation*}%
and hence, 
\begin{eqnarray*}
\frac{3}{2}\bar{A}\left( \delta G^{\left( 4\right) }\left( 1\right) \right)
&=&\left( 
\begin{array}{c}
a_{0} \\ 
-\alpha%
\end{array}%
\right) \left( 
\begin{array}{c}
Q_{1} \\ 
-\frac{7}{4}Q_{0}\delta +Q_{1}\lrcorner \varphi%
\end{array}%
\right) \\
&=&\left( 
\begin{array}{c}
a_{0}Q_{1}-\frac{7}{4}Q_{0}\alpha -Q_{1}\times \alpha \\ 
\left( -\frac{7}{4}a_{0}Q_{0}+\left\langle Q_{1},\alpha \right\rangle
\right) \delta +\left( a_{0}Q_{1}+\frac{7}{4}Q_{0}\alpha \right) \lrcorner
\varphi -\psi \left( Q_{1},\alpha \right) -\alpha Q_{1}-Q_{1}\alpha%
\end{array}%
\right) .
\end{eqnarray*}%
Thus, combining everything, we do indeed obtain (\ref{FA}).
\end{proof}

\begin{lemma}
\label{LemSkewWPWt}Suppose $\hat{P}$ is a symmetric $\func{Im}\mathbb{O}$%
-valued $1$-form and $W=w_{0}+\hat{w}\in \Gamma \left( \mathbb{O}M\right) $.
Then, 
\begin{equation}
\func{Skew}\left( \func{Im}\left( W\left( \hat{P}\bar{W}\right) ^{t}\right)
\right) =\frac{1}{2}P\left( w\right) \wedge w
\end{equation}
\end{lemma}

\begin{proof}
Using the definition (\ref{Atranspose}) of the transpose for sections of $%
\Omega ^{1}\left( \mathbb{O}M\right) $, we have 
\begin{eqnarray*}
\left( \hat{P}\bar{W}\right) _{a}^{t} &=&P\left( w\right) _{a}+\left\langle 
\hat{P}_{b}\bar{W},\delta _{a}\right\rangle \delta ^{b} \\
&=&P\left( w\right) _{a}+P^{bc}\left\langle \delta _{c}\bar{W},\delta
_{a}\right\rangle \delta _{b}
\end{eqnarray*}%
Hence, 
\begin{equation*}
W\left( \hat{P}\bar{W}\right) ^{t}=P\left( w\right) _{a}W+P^{bc}\left\langle
\delta _{c}\bar{W},\delta _{a}\right\rangle W\delta _{b}
\end{equation*}%
So, 
\begin{eqnarray*}
\func{Im}\left( W\left( \hat{P}\bar{W}\right) ^{t}\right) _{ad}^{\ }
&=&P\left( w\right) _{a}\left\langle W,\delta _{d}\right\rangle
+P^{bc}\left\langle \delta _{c}\bar{W},\delta _{a}\right\rangle \left\langle
W\delta _{b},\delta _{d}\right\rangle \\
&=&P\left( w\right) _{a}w_{d}+P^{bc}\left\langle \delta _{c}\bar{W},\delta
_{a}\right\rangle \left\langle \delta _{b}\bar{W},\delta _{d}\right\rangle
\end{eqnarray*}%
Thus, skew-symmetrizing, we get 
\begin{equation*}
\func{Im}\left( W\left( \hat{P}\bar{W}\right) ^{t}\right) _{[ad]}^{\
}=P\left( w\right) _{[a}w_{d]}.
\end{equation*}
\end{proof}

\begin{lemma}
\label{LemOctSol}Suppose $\hat{P}$ is a symmetric $\func{Im}\mathbb{O}$%
-valued $1$-form and $W\in \Gamma \left( \mathbb{O}M\right) $. Then, for any 
$\func{Im}\mathbb{O}$-valued $1$-form $K$, the equation 
\begin{equation}
\hat{P}\bar{W}-\left\vert W\right\vert ^{2}\left( \hat{P}\bar{W}\right)
^{t}-2w_{0}\hat{P}=K  \label{QKeq}
\end{equation}%
has a unique solution if and only if $W^{2}\neq -1.$ In that case, the
solution is 
\begin{equation}
\hat{P}=-\frac{\left( K+\left\vert W\right\vert ^{2}K^{t}\right) \left( \bar{%
W}+\left\vert W\right\vert ^{2}W\right) }{\left( 1+\left\vert W\right\vert
^{2}\right) \left\vert W\right\vert ^{2}\left\vert 1+W^{2}\right\vert ^{2}}
\label{QKsol}
\end{equation}
\end{lemma}

\begin{proof}
Take the transpose of (\ref{QKeq}), then we get 
\begin{equation}
-\left\vert W\right\vert ^{2}\left( \hat{P}\bar{W}\right) +\left( \hat{P}%
\bar{W}\right) ^{T}-2w_{0}\hat{P}=K^{T}  \label{QKeqT}
\end{equation}%
Then, multiplying (\ref{QKeqT}) by $\left\vert W\right\vert ^{2}$ and adding
it to (\ref{QKeq}), we eliminate $\left( \hat{P}\bar{W}\right) ^{T},$ and
obtain%
\begin{equation*}
\left( 1-\left\vert W\right\vert ^{4}\right) \left( \hat{P}\bar{W}\right)
-2w_{0}\left( 1+\left\vert W\right\vert ^{2}\right) \hat{P}=K+\left\vert
W\right\vert ^{2}K^{T}
\end{equation*}%
and thus, 
\begin{equation*}
\hat{P}\left( \left( 1-\left\vert W\right\vert ^{2}\right) \bar{W}%
-2w_{0}\right) =\frac{K+\left\vert W\right\vert ^{2}K^{T}}{1+\left\vert
W\right\vert ^{2}}
\end{equation*}%
However 
\begin{equation*}
\left( 1-\left\vert W\right\vert ^{2}\right) \bar{W}-2w_{0}=-\left(
W+\left\vert W\right\vert ^{2}\bar{W}\right)
\end{equation*}%
and 
\begin{eqnarray*}
W+\left\vert W\right\vert ^{2}\bar{W} &=&W\left( 1+\left\vert W\right\vert
^{2}\bar{W}W^{-1}\right) \\
&=&W\left( 1+\bar{W}^{2}\right)
\end{eqnarray*}%
Hence, 
\begin{equation*}
\left\vert W+\left\vert W\right\vert ^{2}\bar{W}\right\vert ^{2}=\left\vert
W\right\vert ^{2}\left\vert 1+W^{2}\right\vert ^{2}
\end{equation*}%
Therefore indeed, there is a unique solution if and only if $W^{2}\neq -1,$
and that solution is given by (\ref{QKsol}).
\end{proof}

\bibliographystyle{jhep2}
\bibliography{refs2}

\end{document}